\documentclass[12pt]{article}

\usepackage[colorlinks,citecolor=blue,linkcolor=blue,urlcolor=blue]{hyperref}
\usepackage{etoolbox}

\usepackage[margin=1in]{geometry}
\usepackage{natbib}
\usepackage{amsmath}
\usepackage{amssymb}
\usepackage{mathrsfs}
\usepackage{enumerate}
\usepackage{mathtools}
\usepackage{dsfont}
\usepackage{amsthm}
\usepackage{mathabx}
\usepackage{cleveref}
\usepackage{subcaption}
\usepackage{multirow}
\usepackage{graphicx}

\theoremstyle{plain}
\newtheorem{theorem}{Theorem}[section]
\newtheorem{lemma}[theorem]{Lemma}
\newtheorem{corollary}[theorem]{Corollary}
\newtheorem{proposition}[theorem]{Proposition}

\theoremstyle{definition} 
\newtheorem{example}[theorem]{Example}
\newtheorem{definition}[theorem]{Definition}
\newtheorem{remark}[theorem]{Remark}
\newtheorem{assumption}[theorem]{Assumption}

\crefname{definition}{Definition}{Definitions}
\crefname{lemma}{Lemma}{Lemmatta}
\crefname{theorem}{Theorem}{Theorems}
\crefname{proposition}{Proposition}{Propositions}
\crefname{equation}{Equation}{Equations}
\crefname{section}{Section}{Sections}
\crefname{example}{Example}{Examples}
\crefname{corollary}{Corollary}{Corollaries}
\crefname{remark}{Remark}{Remarks}
\crefname{figure}{Figure}{Figures}
\crefname{enumi}{Condition}{Conditions}
\crefname{footnote}{Footnote}{Footnotes}
\crefname{assumption}{Assumption}{Assumptions}

\DeclareMathOperator*{\argmax}{argmax}

\DeclareMathOperator{\succz}{succ}

\DeclareMathAlphabet{\mathpzc}{OT1}{pzc}{m}{it}

\newcommand{\eps}{\varepsilon}
\newcommand{\st}{\text{ such that }}
\newcommand{\ow}{\text{otherwise}}

\newcommand{\R}{\mathbb{R}}

\newcommand{\card}[1]{\left|#1\right|}
\newcommand{\sbr}[1]{\left[#1\right]}
\newcommand{\cbr}[1]{\left\{#1\right\}}
\newcommand{\p}[1]{\left(#1\right)}
\newcommand{\set}[1]{\cbr{1,\ldots,#1}}
\newcommand{\prob}{\mathds{P}}
\newcommand{\expect}{\mathds{E}}

\newcommand{\exintquantity}{X}
\newcommand{\exinttransfer}{T}

\newcommand{\choice}{W}
\newcommand{\lowchoice}{\underline{\choice}}
\newcommand{\highchoice}{\overline{\choice}}

\newcommand{\player}{\xi}

\newcommand{\experiment}{\psi}
\newcommand{\action}{a}
\newcommand{\potentialBidders}{B}
\newcommand{\realBidders}{R}
\newcommand{\shillBidders}{S}
\newcommand{\numBidders}{N}
\newcommand{\pReal}{p}
\newcommand{\valDist}{F}

\newcommand{\val}{v}
\newcommand{\numPossVals}{M}
\newcommand{\type}{\theta}
\newcommand{\typeDist}{F}
\newcommand{\typepmf}{f}
\newcommand{\typeSpace}{\vartheta}
\newcommand{\game}{G}
\newcommand{\actionQuantity}{x}
\newcommand{\actionTransfer}{t}
\newcommand{\util}{u}
\newcommand{\strategy}{\sigma}
\newcommand{\strategyParameterized}{\strategy(\type; \potentialBidders)}
\newcommand{\reserveIndex}{\rho}
\newcommand{\reserve}{\type^\reserveIndex}
\newcommand{\optimalReserve}{\type^{\reserveIndex^*}}
\newcommand{\signal}{s}
\newcommand{\signalSet}{\mathcal{S}}
\newcommand{\quantity}{\tilde{\actionQuantity}}
\newcommand{\transfer}{\tilde{\actionTransfer}}
\newcommand{\actionSet}{A}
\newcommand{\devActionSet}{\mathcal{\actionSet}^{\experiment}}

\newcommand{\devValSet}{\vartheta^{\experiment}}

\newcommand{\menuRule}{\mu}
\newcommand{\possibleValues}{\Theta}
\newcommand{\highValue}{\overline{{\possibleValues}}}
\newcommand{\lowValue}{\underline{{\possibleValues}}}
\newcommand{\nextPlayer}{\vec{\player}}
\newcommand{\directAuction}{(\quantity,\transfer,\menuRule,\player_0)}
\newcommand{\optDirectAuction}{(\quantity^*,\transfer,\menuRule,\player_0)}
\newcommand{\effDirectAuction}{(\quantity^E,\transfer,\menuRule,\player_0)}
\newcommand{\reportedAction}[1]{\action^{\leadsto #1}}
\newcommand{\reportedValue}[1]{\type^{\leadsto #1}}
\newcommand{\publicExperiment}{\experiment = \mathrm{Id}}
\newcommand{\privateExperiment}{\experiment = \emptyset}
\newcommand{\ordering}{\rhd}
\newcommand{\contDist}{\mathpzc{\typeDist}}
\newcommand{\contpdf}{\mathpzc{\typepmf}}
\newcommand{\worstCaseQueries}{Q}
\newcommand{\optTransfer}{\transfer^*(\quantity,\menuRule,\player_0,\val,\typeDist)}

\newcommand{\possibleHistories}{H}
\newcommand{\history}{h}
\newcommand{\informationSet}{\mathcal{I}}
\newcommand{\possibleTerminalHistories}{Z}
\newcommand{\terminalState}{z}
\newcommand{\recentAction}{\mathcal{A}}

\newcommand{\playerFunction}{P}

\newcommand{\affOrder}{\succeq_{\mathrm{Aff}}}
\newcommand{\infoOrder}{\succeq_{\mathrm{Info}}}
\newcommand{\valOrder}{\succeq_{\mathrm{Com}}}

\newcommand{\evenSpacing}{\delta}
\newcommand{\screenValSuper}{Y}
\newcommand{\screenVal}{\type^{\screenValSuper}}

\newcommand{\strong}{strong}
\newcommand{\Strong}{Strong}
\newcommand{\strongly}{strongly}
\newcommand{\Strongly}{Strongly}
\newcommand{\weak}{weak}
\newcommand{\Weak}{Weak}
\newcommand{\weakly}{weakly}
\newcommand{\Weakly}{Weakly}
\newcommand{\strategyproof}{ex-post incentive compatible}
\newcommand{\STrategyproof}{Ex-Post Incentive Compatible}
\newcommand{\strategyproofness}{ex-post incentive compatiblility}

\newtoggle{compressLines}
\toggletrue{compressLines}

	\title{Shill-Proof Auctions\thanks{The authors thank Mohammad Akbarpour, Eric Budish, Jeremy Bulow, Peter Cramton, Jacob Leshno, Shengwu Li, Joshua Mollner, Stephen Morris, Marek Pycia, Alex Rees-Jones, Roberto Saitto, Clayton Thomas, and seminar audiences at a16z crypto, EC '25, MIT, SITE, and Stony Brook for helpful comments.
			Komo gratefully acknowledges support from the NSF Graduate Research Fellowship Program.
			Kominers gratefully acknowledges support from the Ng Fund and the Mathematics in Economics Research Fund of the Harvard Center of Mathematical Sciences and Applications. 
			Roughgarden's research at Columbia University is  supported in part by NSF awards
			CCF-2006737 and CNS-2212745, and research awards from the Briger Family Digital Finance Lab and the Center
			for Digital Finance and Technologies.
			Part of this work was conducted during the Simons Laufer Mathematical Sciences Institute Fall 2023 program on the Mathematics and Computer Science of Market and Mechanism Design, which was supported by NSF Grant No.\ DMS-1928930 and by the Alfred P.\ Sloan Foundation under grant G-2021-16778.
			Kominers and Roughgarden are members of the Research Team at a16z crypto (for general a16z disclosures, see https://www.a16z.com/disclosures/). Komo was an intern on the a16z crypto Research Team during part of this research. Notwithstanding, the ideas and opinions expressed herein are those of the authors, rather than of a16z or its affiliates. 
			Kominers and Roughgarden also advise companies on marketplace and incentive design, including auction design. 
			Any errors or omissions remain the sole responsibility of the authors.
	}}
	\author{Andrew Komo\footnote{MIT Economics, Email: \href{komo@mit.edu}{komo@mit.edu}} \and Scott Duke Kominers\footnote{Harvard Business School, Harvard University, and a16z crypto, Email: \href{kominers@fas.harvard.edu}{kominers@fas.harvard.edu}} \and Tim Roughgarden\footnote{a16z crypto and Columbia University, Email: \href{tim.roughgarden@gmail.com}{tim.roughgarden@gmail.com}}}
	
	\date{\today}

\begin{document}
 	
\maketitle

\begin{abstract}
	We characterize single-item auction formats that are \textit{shill-proof} in the sense that a profit-maximizing seller has no incentive to submit shill bids. We distinguish between \textit{\strong} shill-proofness, in which a seller with full knowledge of bidders' valuations can never profit from shilling, and \textit{\weak} shill-proofness, which requires only that the expected equilibrium profit from shilling is non-positive.
	The Dutch auction (with a suitable reserve) is the unique (revenue-)optimal and \strongly~shill-proof auction. Any deterministic auction can satisfy only two properties in the set $\{$static, strategy-proof, \weakly~shill-proof$\}$.
	Our main results extend to settings with affiliated and interdependent values.
\end{abstract}

\newpage

\section{Introduction}

\subsection{Shill Bidding in Auctions}\label{ss:shill}

\textbf{Shill Bidding in Practice. }
Auction theory typically assumes that an auction is carried out as
described (by the seller or a third party) and focuses solely on
the bidders' incentives.
Reality is often different. For example, while major
auction houses like Christie's or Sotheby's may appear to be carrying
out textbook English (ascending) auctions, a degree of skullduggery is often
afoot.  According to a \textit{New York Times} article from 2000:
\begin{quote}
Some tricks of the trade, like an auctioneer's drumming up excitement
by acknowledging nonexistent bids only he hears and potential buyers
who bid with nearly imperceptible secret signals, have been around for
decades. Making up bids, for instance, is known as ``bidding off the
chandelier'' from an era when the grand auction rooms were adorned
with ornate lighting.\footnote{See ``Genteel Auction Houses Turning Aggressive,'' \textit{New York Times}, April 24, 2000.}
\end{quote}
The practice
continues to this day: Christie's Conditions of Sale for their
flagship New York location, in a section titled ``Auctioneer's
Discretion,'' states (among other things) that ``The auctioneer
can\ldots move the bidding backwards or forwards in any way he or she
may decide\ldots.''\footnote{See \url{https://www.christies.com/help/buying-guide-important-information/conditions-of-sale}.}

Such chandelier bids or \textit{shill} bids---bids submitted by the seller
in order to manipulate the final selling price---appear to be
particularly common in online auctions.
For example, eBay has long gone out of its way to emphasize that shill
bidding is forbidden and will be punished:
\begin{quote}
 We want to maintain a fair marketplace for all our users, and as
 such, shill bidding is prohibited on eBay. [\ldots] eBay has a number of
 systems in place to detect and monitor bidding patterns and
 practices. If we identify any malicious behavior, we'll take steps to
 prevent it.\footnote{See \url{https://www.ebay.com/help/policies/selling-policies/selling-practices-policy/shill-bidding-policy?id=4353}.}
\end{quote}
According to many eBay users, however, shill bidding remains
rampant. Here's a sample quote from the eBay discussion forums:
\begin{quote}
  The Sellers post a Buy Now price 3--4 times the actual cost of the
  item.  Then they place the item on an auction at \$0.01. This is to get
  as many views as possible.  The shill comes in shortly after the
    auction starts and
\ldots is there to prevent the item from being sold below
  their profit margin.\footnote{See
\url{https://community.ebay.com/t5/Buying/My-experience-with-Shill-bidders/td-p/30402514}.}
\end{quote}
\cite{chenEtAl20} find that nearly 10\% of all bidders on eBay Motors are shill bidders.

\textbf{Shill-Proof Auctions. }
Much of auction theory to date encourages truthful
bidding through careful auction design, while punting on challenges
like seller deviations and collusion via appeal to unmodeled concepts
such as the out-of-mechanism enforcement of rules.
Anecdotes about eBay and other online
platforms suggest that such methods are only partially
effective at deterring seller deviations. Thus, it makes sense to ask: To what extent can these deviations
instead be disincentivized through an auction's design?

The goal of this paper is to understand which auction formats
are ``shill-proof'' in the sense that a seller cannot profit through the
submission of shill bids. 
We show that shill bidding can matter even in private-value auctions.
The reader might wonder why shill bids can have an impact in the private-values case---assuming that the choice of reserve price does not affect
participation (as it does in the eBay example), isn't a shill bid the
same thing as a reserve price?

The answer depends on when the seller has an opportunity to shill and
the information available to them at that time.  For example, consider
an English auction in which the seller also participates via shill
bidding. Suppose the valuations of the (real) bidders are private and drawn i.i.d.\
from a regular distribution~$\valDist$ and that the opening bid of the
auction is set optimally (for revenue), to the monopoly price~$\optimalReserve$
of~$\valDist$. As the auction proceeds, with the offered price~$p$ starting
at~$\optimalReserve$ and increasing from there (in increments of~$\epsilon$, say),
the seller can shill bid at any time. Suppose that the only additional
information known to the seller at a given round of the auction is
that the remaining bidders are willing to pay at least~$p$. Then, the
seller asks himself: ``now that I know how many bidders are willing to
pay at least~$p$, do I want to shill and reset the reserve price
to~$p+\epsilon$?'' Under our assumption that~$F$ is
regular, the answer is ``no,'' and an expected revenue--maximizing
seller will never shill.\footnote{Auction theory experts will now immediately recognize that the English auction with an optimally chosen reserve price is \textit{not} generally shill-proof in this sense when the valuation distribution is not regular or the values are not private. See, for example, Footnote 18 of \cite{MW82} for more discussion of the latter point.}

Now suppose that the seller has full knowledge of bidders' realized
valuations. In this scenario, the seller will certainly, in some cases, want to
shill in an English auction to push the price up to just below the
highest of the bidders' valuations. Lest this informational assumption---that the seller knows the full valuation profile---seem impossibly demanding, consider
the Dutch (descending) auction (with an arbitrary
reserve price). Here, any shill bid by the seller terminates the
auction immediately, leaving the seller holding the item and earning
zero revenue. Therefore, even if the seller knows the full valuation profile, he would not want to shill bid.

We map out a theory of ``shill-proof'' auctions, focusing
on the following basic questions:
\begin{itemize}

\item Which auction formats are ``\strongly~shill-proof'' in the sense
  of the Dutch auction, i.e., with shill bidding being unprofitable even with full
  knowledge of bidders' realized valuations?

\item Which auction formats are ``\weakly~shill-proof'' in the sense of
  the English auction (with bidders' valuations drawn i.i.d.~from a
  regular distribution and an optimally chosen reserve price),  i.e., with
  shill bidding being unprofitable in expectation at equilibrium?

\item To what extent are \strong~and \weak~shill-proofness compatible
  with other desirable properties such as optimality, efficiency,
  \strategyproofness, and sealed-bid implementation? To what extent is shill-proofness dependent on the bidders' valuation structure?
  
\end{itemize}

\subsection{Overview of Results}\label{ss:overview}

Iterative auction formats like Dutch and English auctions play a central role in our theory, and accordingly we study (real and shill) bidding in the extensive-form game that is induced by a choice of auction format, relying on a framework for extensive-form auction analysis developed by~\citet{li17} and~\citet{akbarpourLi20}.
We consider single-item auctions with~$N$ bidders. 
A subset of these are shill bidders, which we model as bidders with zero private value for the item and with utility equal to the seller's revenue.\footnote{\label{fn:modeling}The prior literature has sometimes modeled shill bidding via an unknown number of bidders, with some subset of the bidders who end up participating in the auction being shills. Our framework is essentially equivalent: we can take $N$ to be large and require~$0$ to be in the support of the valuation distribution; and a bidder with value~$0$ is equivalent (in terms of outcomes) to a bidder not arriving.} For most of the paper, we assume that non-shill bidders have private valuations drawn i.i.d.\ from a known distribution. In \cref{sec:aff_and_sp}, we generalize to any affiliated type distribution and any interdependent value function satisfying some curvature assumptions (\cref{assum:interdep}). We also assume that shill bidders observe all actions. An auction is then \textit{\weakly~shill-proof} (\cref{def:wsp}) if there exists an equilibrium of the induced extensive-form game in which the shill bidders never shill (i.e., always bid their true private value of~$0$). An auction is \textit{\strongly~shill-proof} (\cref{def:ssp}) if, moreover, shill bidders' equilibrium strategies are ex-post strategies.
In our first result, we focus on \textit{public} auctions (\cref{def:public}), meaning auctions in which every bidder's action is publicly observable. This is arguably the most natural model for the analysis of typical iterative auctions such as Dutch and English auctions.
We then turn to arbitrary information structures to prove our other results.

Next, we summarize the main results of this paper; see also \cref{f:summary}.

\textbf{\Strongly~Shill-Proof Auctions. }
Our main result (\cref{thm:dutch_ssp}) uniquely characterizes \strongly~shill-proof auctions: the Dutch auction (with consistent tie-breaking and monopoly reserve price) is \strongly~shill-proof and optimal (i.e., maximizes the seller's expected revenue), and is the \textit{only} such auction in the public setting. In particular, \strongly~shill-proof optimal auctions cannot avoid using a large number of rounds,
and they cannot be \strategyproof~(for real bidders). The rough intuition for the proof of this result is that: (i) No matter the information structure, strongly shill-proof auctions must be pay-as-bid (\cref{lem:ssp_pab}); (ii) for any auction format other than a Dutch auction, there exists a history in which some bidder~$i$ can indicate that her value is higher than $0$ without the auction ending immediately; (iii) incentive compatibility in tandem with the public setting then implies that this information effectively induces the auction to revise its reserve price upward or increases perceived competitiveness for the item being sold, which reduces bid shading; and thus (iv) there exist valuations for the bidders such that, if bidder~$i$ is a shill bidder, shilling will increase the seller's revenue.\footnote{We focus on single-item auctions, but our uniqueness results a fortiori provide an upper bound on what is possible for multiunit auctions, as well. We leave formal study of multi-unit auctions and other settings, such as sequential auctions, to future work.
}

\textbf{Weakly Shill-Proof and \STrategyproof~Static Auctions.}
We next turn to investigating the richer design space of \weakly~shill-proof auctions.
Our main result about weakly shill-proof auctions (\cref{thm:wsp_trilemma}), focuses on \textit{single-action} auctions, meaning auction formats that induce extensive-form games in which each bidder moves exactly once.
In our framework, a single-action auction is \textit{static} when each bidder moves simultaneously.
We prove that no \weakly~shill-proof, single-action auction can satisfy even a very weak \strategyproofness~condition (\cref{def:mild}). 
Thus, an auction can satisfy two and only two of the properties in the set
$\{$single-action, \strategyproof, \weakly~shill-proof$\}$.\footnote{Assuming a regular valuation distribution and a corresponding optimal reserve price, a second-price auction is single-action and \strategyproof; a first-price auction is single-action and \weakly~shill-proof; and an English auction is \strategyproof~and \weakly~shill-proof.}

\textbf{\Weakly~Shill-Proof and Efficient Auctions. }
We then investigate \textit{efficient} (and \weakly~shill-proof) public auctions.
We prove that a Dutch auction with a reserve price of~$0$ is the unique prior-independent auction (in the sense of~\citet{DRYgeb}, with no dependence whatsoever on the valuation distribution) that is both efficient and \weakly~shill-proof (\cref{cor:dutch_robust_wsp}).
In the regular, IPV case, we show (in \cref{prop:dutch_robust_wsp}) that fixing the monopoly price, a prior-independent, \weakly~shill-proof and efficient auction must conclude with a Dutch auction when all bidders' valuations are known to be below said price.
A format such as beginning with an English auction at the monopoly price and then, should there be no takers, concluding with a Dutch auction is an example of an auction that is \weakly~shill-proof and efficient given the monopoly price.\footnote{In fact, this auction format closely resembles the Honolulu--Sydney fish auction documented by \citet{hafalirEtAl23}.}

\textbf{\Weakly~Shill-Proof and \STrategyproof~Optimal Auctions. }
The previous two results imply that \strategyproof~auctions cannot be both \strongly~shill-proof and optimal, nor can they be (robustly) \weakly~shill-proof and efficient. The English auction (with an optimal reserve price) is, as we've noted, \weakly~shill-proof, optimal, and \strategyproof~in the regular, IPV case.
Is it the unique such auction?
Does this combination of properties require a potentially large number of rounds? 
Our next result (\cref{prop:wsp_sp_upper_bound}) shows that, in general, the answer is no:
In fact, although we can never find a \weakly~shill-proof, optimal, and \strategyproof~auction that finishes in one round, we can always find a weakly shill-proof auction and a valuation distribution such that the worst-case number of auction rounds required is an arbitrarily small fraction of the number of rounds needed in the English auction.

\textbf{Affiliated and Interdependent Values.}
Our last theorem, \cref{prop:aff_shill_order}, shows that as the type distribution becomes less affiliated (\cref{def:aff_order}) and the values become less commonly valued (\cref{def:val_order}), then both the set of strongly shill-proof auctions and the set of weakly shill-proof auctions expand.
We then apply \cref{prop:aff_shill_order} to show most of our characterization results  (\cref{thm:dutch_ssp}, \cref{thm:wsp_trilemma}, and \cref{cor:dutch_robust_wsp}) hold in a more general, affiliated environment.
\Cref{prop:aff_shill_order} means that shill-proofness admits a partial order with respect to both the value function and the affiliation structure.
Our theorem holds for any affiliation structure, most commonly used interdependent value functions (the function must satisfy \cref{assum:interdep}), and any extensive-form auction where an optimal transfer rule is used.
The last restriction is necessary because we consider discrete types and so there are a multitude of incentive compatible transfers that could yield different revenue.
We generalize \cref{lem:ssp_pab} in \cref{prop:ssp_pab} to show that even with affiliation, a strongly shill-proof and optimal auction must be pay-as-bid with a reserve structure on the allocation.

\begin{figure}
\centering
\begin{subfigure}[t]{0.48\textwidth}
    \footnotesize\setlength\tabcolsep{3pt}
    \resizebox{\textwidth}{!}{
    \begin{tabular}{|c|c|c|}
        \hline
        & Static & Not Static \\
        \hline
        Strategy- & Impossible & Ascending, Screening Auction \\
        Proof & (\cref{thm:wsp_trilemma}) & (\cref{prop:wsp_sp_upper_bound}) \\
        \hline
        Not Strategy- & First-Price Auction & Dutch Auction \\
        Proof & ~ & (\cref{thm:dutch_ssp}) \\
        \hline
    \end{tabular}}
    \caption{\Weakly~shill-proof and optimal auctions}
\end{subfigure}~
\begin{subfigure}[t]{0.48\textwidth}
    \footnotesize \setlength\tabcolsep{3pt}
    \resizebox{\textwidth}{!}{
	\begin{tabular}{|c|c|c|}
        \hline
        & Static & Not Static \\
        \hline
        Strategy- & Impossible & Not Robustly \\
        Proof & (\cref{thm:wsp_trilemma}) & (\cref{prop:dutch_robust_wsp}) \\ 
        \hline
        Not Strategy- & Not Robustly & Dutch Auction \\
        Proof & (\cref{prop:dutch_robust_wsp}) & (Robustly Unique, \cref{prop:dutch_robust_wsp}) \\
        \hline
    \end{tabular}}
    \caption{\Weakly~shill-proof and efficient auctions}
\end{subfigure}
\par
\begin{subfigure}[t]{0.48\textwidth}
    \footnotesize \setlength\tabcolsep{3pt}
    \resizebox{\textwidth}{!}{
    \begin{tabular}{|c|c|c|}
        \hline
        & Static & Not Static \\
        \hline
        Strategy- & Impossible  & Impossible  \\  
        Proof &  (\cref{thm:dutch_ssp}) &  (\cref{thm:dutch_ssp}) \\  
        \hline
        Not Strategy- & Impossible  & Dutch Auction  \\
        Proof &  (\cref{thm:dutch_ssp}) &  (Unique, \cref{thm:dutch_ssp}) \\
        \hline
    \end{tabular}}
    \caption{\Strongly~shill-proof and optimal auctions}
\end{subfigure}~
\begin{subfigure}[t]{0.48\textwidth}
    \footnotesize \setlength\tabcolsep{3pt}
    \resizebox{\textwidth}{!}{
	\begin{tabular}{|c|c|c|}
        \hline
        & Static & Not Static \\
        \hline
        Strategy- & Impossible  & Impossible  \\  
        Proof &  (\cref{prop:dutch_robust_wsp}) &  (\cref{prop:dutch_robust_wsp}) \\  
        \hline
        Not Strategy- & Impossible  & Dutch Auction  \\
        Proof & (\cref{prop:dutch_robust_wsp}) & (Unique, \cref{prop:dutch_robust_wsp}) \\
        \hline
    \end{tabular}}
    \caption{\Strongly~shill-proof and efficient auctions}
\end{subfigure}  
\caption{\textbf{Summary of Results.} Characterization of single-item auction formats that are \strongly~or \weakly~shill-proof, along with other properties such as optimality, efficiency, \strategyproofness, and sealed-bid implementations.} \label{f:summary}
\end{figure}

\subsection{Related Work}

While the idea and practice of shill bidding by a seller have long been well known, the auction theory literature on the topic is surprisingly thin.
\citet{chakrabortyKosmopoulou04} consider common value auctions and focus on technological barriers (as opposed to auction formats) that can mitigate shill bidding.
\citet{lamy09} studies shill bidding specifically in English auctions in which bidders' valuations are affiliated in the sense of \citet{MW82}, and proves that shill bidding effectively cancels out the effects of affiliation in equilibrium due to real bidders conditioning on bids being fake (see also \citet{I04}).
We also consider shill bidding in affiliated environments, but consider a much larger class of possible extensive-form games.

\cite{PS05} consider a model similar to a second-price auction, motivated by ``cheating'' by online platforms that can announce a manipulated auction outcome subsequent to collecting all of the bidders' bids.
More recently, a number of works (e.g., \cite{roughgarden21,LSZ19,beos,CS23}) have considered shill bidding in the context of blockchain transaction fee mechanism design, with an emphasis on knapsack auctions that are \strategyproof, shill-proof, and robust to various forms of collusion. \citet{ausubelMilgrom06} and \citet{dayMilgrom08} consider shill bids by \textit{bidders} in a multi-item auction, who are looking to exploit complementarities to lower their payments in VCG-type mechanisms---as opposed to shill bids by a seller looking to increase revenue, as is the case of this paper.\footnote{More broadly there is a literature on sybil-resistance referred to as ``false-name proofness'' (see, e.g., \cite{conitzerEtAl10} for a reference).}
Contemporaneous work by \cite{shinozaki24} and \cite{zeng24} also study shill bidding, but take a different approach to the problem.
While we study ex-interim deterrence against shill bidding in extensive form games, they group auctions into equivalence classes based on the outcome and primarily focus on ex-ante deterrence against the seller inserting additional bidders into the auction to increase perceived competition.
The results in those papers are complementary to ours: While we show that the dynamics of the auction are important for preventing shill bidding, they show that when the shill bidders can increase perceived competition outside of taking actions in the auction, only the posted-price mechanism is non-manipulable (see also \cref{foot:zeng}).

Our theory of shill-proof auctions is similar in spirit to the theory of credible mechanisms developed by \citet{akbarpourLi20}, and leverages their framework for extensive-form auction analysis. 
That said, shill-proofness differs conceptually from credibility as shill-proofness focuses on the auctioneer's incentives to insert fake bids, whereas credibility focuses on the auctioneer's incentive to truthfully report the actions of a bidder to other bidders.
Furthermore, the results in this paper are also qualitatively different.
For example, there are a multitude of credible auctions, but only one \strongly~shill-proof auction and there are a multitude of strategy-proof and \weakly~shill-proof auctions, but only one strategy-proof and credible auction (see \cref{ss:shill_v_cred} for more discussion).
In \cref{prop:cred_nesting}, we prove that strong shill-proofness is a stronger condition than credibility and weak shill-proofness is a weaker condition in the sense that if an auction is strongly shill-proof, then it is credible, which in turn implies that it is weakly shill-proof.

More recent research on credible mechanisms, usually with a focus on evading the impossibility results of \citet{akbarpourLi20} under extra assumptions (such as adding cryptographic tools), includes the work of \citet{essaidiEtAl22}, \citet{ferreiraWeinberg20}, and \citet{chitraEtAl23}.
More distantly related papers include that of \citet{hauptHitzig22}, who prove a uniqueness result for the Dutch auction under contextual privacy constraints.

\subsection{Outline of the Paper}
In \cref{sec:model}, we present our formal model of shill bidding in auctions.
\Cref{sec:ssp} studies strong shill-proofness and shows the ways in which Dutch auctions are uniquely optimal at disincentivizing shill bidding.
\Cref{sec:wsp} explores which formats are weakly shill-proof and provides a trilemma result.
\Cref{sec:aff_and_sp} generalizes the model to a setting with affiliation and interdependent values, proves that shill-proofness admits, and ordering with respect to affiliation and interdependent values, and generalizes our results.
In \cref{sec:discussion}, we conclude the paper by discussing extensions.

\section{Model}\label{sec:model}
In this paper, we consider extensive-form, single item auctions.
An extensive-form game $\game$ is a tuple of possible histories $\possibleHistories$, and, for each history $\history \in \possibleHistories$, functions mapping $\history$ to: (i) a player taking an action, $\playerFunction(\history)$; (ii) a set of possible actions, $\actionSet(\history)$; (iii) an information set, $\informationSet(\history)$; and (iv) the most recent action taken, $\recentAction(\history)$.
As further notation, we denote the starting history of the game by $\history_{\emptyset}$ and the set of terminal histories as $\possibleTerminalHistories$; we say $\history' \prec \history$ if $\history'$ precedes $\history$, i.e., there exists a sequence of actions that lead from $\history'$ to $\history$.

We restrict attention to single-item auctions, which means that for every terminal history $\terminalState \in \possibleTerminalHistories$, we can associate an allocation and transfer vector: $\terminalState = (\actionQuantity,\actionTransfer)$, with $\sum_{i=1}^{\numBidders} \actionQuantity_i \le 1$ and $x_i \in \cbr{0,1}$ for all $i$.
Abusing notation slightly, we use $\actionQuantity(\terminalState),\actionTransfer(\terminalState)$ to mean the vectors $(\actionQuantity,\actionTransfer)$ associated with the terminal history $\terminalState$.
We also assume perfect recall and finite depth. (\cref{def:game} in the Supplemental Appendix gives a formal, thorough, and standard definition of extensive-form games.)

\subsection{Bidders -- Real and Shill}
In the auction, there is a set of potential bidders $\potentialBidders$, with $\card{\potentialBidders} = \numBidders$, who might participate.
We assume that the seller's value is commonly known to be $0$.
Of these potential bidders, a set of real bidders $\realBidders$ actually participate.
Each bidder $i \in \potentialBidders$ has an independent probability $\pReal$ of participating, $\prob\sbr{i \in \realBidders} = \pReal$.\footnote{This randomness  plays little role in our analysis---we impose it only so that the overarching structure of our model has bidders with ex-ante, symmetric, independent private values. See also \cref{fn:modeling}.} 
The other bidders, $\shillBidders = \potentialBidders \backslash \realBidders$, are shill bidders whose incentives are completely aligned with the seller/auctioneer's, i.e., their utility is defined by the sum of real bidders' transfers: for $i \in \shillBidders, \util_i(\terminalState) = \sum_{j \in \realBidders} \actionTransfer_j(\terminalState)$.
Each real bidder $i \in \realBidders$ has value $\type_i$ for the item being sold where $\type_i \sim \typeDist$ independently for each $i$.
Each real bidder has quasi-linear utility: for $i \in \realBidders$,
\iftoggle{compressLines}{$\util_i(\terminalState) = \actionQuantity_i(\terminalState) \type_i - \actionTransfer_i(\terminalState).$}{\[ \util_i(\terminalState) = \actionQuantity_i(\terminalState) \type_i - \actionTransfer_i(\terminalState). \]}
We assume $\typeDist$ is discrete, with support $\typeSpace$ consisting of the ordered atoms $0 = \type^1 < \type^2 < \ldots < \type^{\numPossVals}$, and we define $\typepmf(\type^k) = \prob_{w \sim \typeDist}\sbr{w = \type^k}$ to be the probability mass function (pmf) of the distribution.
As notation, for each shill bidder $i \in \shillBidders$, we assign $\type_i = 0$ and let $\type = (\type_1,\ldots,\type_{\numBidders})$.
The choice of values for shill bidders does not affect their incentives, and by supposing that their values are $0$, we can define efficiency and optimality (revenue-maximization) in terms of only $\type$ instead of $\type$ and $\realBidders$.\footnote{When considering optimal auctions, we naturally assume the seller only cares about raising revenue from real bidders rather than shill bidders, since shill bidders are proxies for the seller themself.}
Observe that given how $\type$ is generated, we are in the standard, symmetric, single-item independent private values (IPV) setting. 

Note that (i) whether a given bidder is real and (ii) the bidders' values for the item are not built into the extensive-form game $\game$.
Instead, bidders' strategies are a function of their types.
Like most papers in the extensive-form auction literature, we study games with a finite type space because defining auctions with a continuum of types requires defining a general class of continuous-time games.
To the authors' knowledge, there is no theory of continuous-time games that rivals the generality and flexibility of extensive-form games.

Real bidders have no ex-ante information about who else is a real bidder; they only know that each other bidder is real with probability $\pReal$.\footnote{Unlike the assumption that $\numBidders$ is fixed, this assumption is an economically substantive one: The only way for shill bidders to manipulate the outcome of the auction is for shill bidders to take actions in the auction. If the bidders instead knew who were the real bidders, shill bidders could have an incentive to appear as real bidders in order to increase perceived competition.
\label{foot:zeng}}
However, throughout the course of the auction, they can update their beliefs about which bidders are real and adjust their actions accordingly.
We assume that shill bidders know the set of shill bidders and observe all previous actions taken.
Formally, for any history $\history$, if $\playerFunction(\history) \in \shillBidders$, then $\informationSet(\history) = \cbr{h}$.
This assumption rules out games with simultaneity (including static games) from the perspective of the shill bidders, but not real bidders.\footnote{In extensive form games, simultaneity is modeled as the information set of a player having multiple elements. Without cryptography or other unmodeled technologies, we view it as reasonable to assume that while the auction may appear simultaneous to the real bidders, actions are taking place sequentially and the seller can observe those actions.}

Our equilibrium concept is pure-strategy Perfect Bayesian Equilibrium; a formal definition of the auction equilibrium $(\game,\strategy)$ can be found in \cref{def:game_eq}.
We write $\strategy(\type;\realBidders)$ for the strategy profile when the value profile is $\type$ and the realized set of real bidders is $\realBidders$.
Perfect Bayesian equilibria are defined without consideration of group deviations by the shill bidders.
However, with this setup, it is without loss to focus only on individual shill bidder's incentives as opposed to any group actions because all the shill bidders have the same objective function and information available to them.

\subsection{Auction Environment}
Throughout the paper, we focus on auction equilibria that are ex-post monotone and individually rational:
The auction equilibrium $(\game,\strategy)$ is \textbf{monotone} if, for all $i$, $j$, $\type_{-j}$, and $\type_j > \type_j'$,
\iftoggle{compressLines}{$\sbr{\actionTransfer_i\p{\strategyParameterized} > 0 \implies \actionTransfer_i\p{\strategyParameterized} \ge \actionTransfer_i\p{\strategy(\type_j',\type_{-j};\potentialBidders)}},$}{\[\actionTransfer_i\p{\strategyParameterized} > 0 \implies \actionTransfer_i\p{\strategyParameterized} \ge \actionTransfer_i\p{\strategy(\type_j',\type_{-j};\potentialBidders)},\]}
and is \textbf{individually rational} (IR) if, for all $\val$ and $i \in \potentialBidders$,
\iftoggle{compressLines}{$\actionQuantity_i(\strategyParameterized)\val_i - \actionTransfer_i(\strategyParameterized) \ge 0.$}{\[ \actionQuantity_i(\strategyParameterized)\val_i - \actionTransfer_i(\strategyParameterized) \ge 0. \]}
The monotonicity condition on transfers is satisfied by all standard single-item auction formats such as English auctions, Vickrey auctions, Dutch auctions, and sealed first-price auctions.
Moreover, note that we impose the monotonicity condition primarily for convenience; all our results continue to hold if we instead just assume that total transfers are weakly higher whenever any bidder reports a value higher than $0$.
The ex-post IR condition rules out all-pay auctions and ensures that only the winner pays the seller.\footnote{Unlike with continuous types, with discrete types there is positive probability of multiple bidders having the same highest value and so an optimal or efficient auction make randomize between bidders. For many of our results, we will assume \cref{def:orderly} to rule out this case.}
We also make this assumption primarily for convenience, versions of all our main results would still hold if we were to relax to ex-interim IR.

\subsection{Shill-Proofness}
Next, we define our key shill-proofness desiderata.
We are interested in auction equilibria in which shill bidders do not shill. 
Formally, this corresponds to requiring that shill bidders always act like real bidders who have value $0$ for the item---since real bidders who have value $0$ will never enter non-trivial bids in equilibrium, requiring shill bidders to have the same actions in equilibrium in effect means that shilling does not occur.

\begin{definition}\label{def:wsp}
	An auction equilibrium $(\game,\strategy)$ is \textbf{\weakly~shill-proof} if $\strategy$ is invariant to the realization of $\shillBidders$, i.e., for all $\type$ and $\shillBidders,\shillBidders' \subseteq \cbr{i : \type_i = 0}$:
	\iftoggle{compressLines}{$\strategy\p{\type;\potentialBidders\backslash\shillBidders} = \strategy\p{\type;\potentialBidders\backslash\shillBidders'}$.}{\[ \strategy\p{\type;\potentialBidders\backslash\shillBidders} = \strategy\p{\type;\potentialBidders\backslash\shillBidders'}. \]}
\end{definition}

Note that \cref{def:wsp} is a statement about an equilibrium of an auction---it is possible (although we have not found an example of this) that an auction may have both shill-proof equilibria and non-shill-proof equilibria.

We can also strengthen our no-shilling criterion from equilibrium to ex-post strategy:
\begin{definition}\label{def:ssp}
	An auction equilibrium $(\game,\strategy)$ is \textbf{\strongly~shill-proof} if it is \weakly~shill-proof and $\strategy$ is an ex-post strategy profile for shill bidders, i.e., for all $\strategy'$, $\shillBidders$, and $\type_{-\shillBidders}$,
	\[  \sum_{j \in \realBidders} \actionTransfer_j(\strategy(0,\type_{-\shillBidders};\realBidders)) \ge \sum_{j \in \realBidders} \actionTransfer_j\p{\strategy_{\shillBidders}',\strategy_{-\shillBidders}(0,\type_{-\shillBidders};\realBidders)}. \]
\end{definition}
\Strong~shill-proofness is obviously preferable to \weak~shill-proofness (all else equal), especially if there are concerns about a seller acquiring information about real bidders' valuations beyond what is encoded by the prior. As we will see, however, the design space of \weak~shill-proof auctions is meaningfully larger than that of \strong~shill-proof auctions.

\subsection{Revelation Principle}

In order to make progress in understanding shill-proof auction formats, the following revelation principle will be helpful: for every auction equilibrium $\p{\game,\strategy}$, there exists a direct auction that can be summarized by a direct allocation rule $\quantity : \typeSpace^{\numBidders} \to [0,1]^{\numBidders}$, a direct transfer rule $\transfer: \typeSpace^{\numBidders} \to \R^{\numBidders}$, a menu rule
\[\menuRule: \mathcal{P}\p{\typeSpace^{\numBidders}} \times \potentialBidders \to \bigcup_{L=2}^{\numPossVals} \p{\cbr{\mathcal{T} \in \Pi(\typeSpace) \mid \card{\mathcal{T}} = L} \times \potentialBidders^{L}}\]
where $\Pi(\typeSpace)$ is the set of all possible partitions over the type space, and a starting player $\player_0 \in \potentialBidders$.\footnote{This revelation principle is similar to those found in, for example, \citet{ashlagiGonczarowski18,mackenzie20,mackenzieZhou22, pyciaTroyan23}.}
The first input to the menu rule $\menuRule$ is a set~$\possibleValues$ of valuation profiles of the form $\possibleValues = \possibleValues_1 \times \possibleValues_2 \times \cdots \times \possibleValues_{\numBidders}$ with~$\possibleValues_i \subseteq \typeSpace$ for all~$i$---intuitively, the valuation profiles that are, in equilibrium, consistent with a particular history. 
The second input is a player~$\player$ who is to move next. 
The output of the rule is a collection $\cbr{\p{\choice_{\ell}, \nextPlayer_{\ell}}}_{\ell \in \set{L}}$, 
where the $\choice_{\ell}$'s are a partition of~$\possibleValues_{\player}$ (the equilibrium strategy $\strategy$ determines the partition; player~$\player$ will truthfully choose a subset 
based on her valuation) and~$\nextPlayer_{\ell}$ indicates the next player to move should 
player~$\player$ choose~$\choice_{\ell}$.
Under $\strategy$, the player $\player$ will always select the partition $\choice_{\ell}$ such that $\type_{\player} \in \choice_{\ell}$.
If $\nextPlayer_{\ell} = \emptyset$, then the game ends should choice $\ell$ be selected by the bidder~$\player$.
For a typical iterative auction, one generally has $L=2$ with the two sets corresponding to types above and below some value, respectively.
Or, for a single-action auction, the $W_{\ell}$'s are generally singletons, with one per type in $\possibleValues_i$.

We show that for any implementable outcome $(\quantity,\transfer)$ of the auction, we can always find a menu rule that is ``informative"---the set of possible outcomes differs across partition selections\footnote{This notion of informativeness is very similar to the pruned condition from \citet{akbarpourLi20}.}---that also implements the same outcome.
So, without loss of generality, we restrict menu rules in this way and then describe an auction equilibrium by $\directAuction$. (See \cref{lem:aug_rev_prin} in the Appendix for a more formal treatment.)
We refer to $\directAuction$ as an \textit{auction} when convenient.
As is always the case with direct mechanisms, the auction encompasses both the game form and the equilibrium, i.e., by appealing to the revelation principle we have implicitly selected the equilibrium.

Finally, as notation for later sections, for a set $\possibleValues$ of valuation profiles, define $\highValue_i = \max\cbr{\type_i : \type_i \in \possibleValues_i}$ to be the maximum possible value of bidder $i$; $\highValue_{-i}$ to be the maximum possible value of bidders $j \ne i$; and $\highValue = \max_i \cbr{\highValue_i}$.
We define $\lowValue_i,\lowValue_{-i}$, and $\lowValue$ as the corresponding values for minima instead of maxima.

\section{Strongly Shill-Proof Auctions}\label{sec:ssp}
\subsection{Direct Mechanisms}
In this subsection, we first show that all \strongly~shill-proof auctions must be pay-as-bid and then show that under an assumption that real bidders observe all past actions, we can precisely pin down the Dutch auction as the the only \strongly~shill-proof auction.

\begin{lemma}[Pay-as-bid]\label{lem:ssp_pab}
	If an auction $\directAuction$ is \strongly~shill-proof, then it must be a pay-as-bid auction. Formally, for all $\player,\type_{\player}$, and $\type_{-\player},\type_{-\player}'$,
	\[ \quantity_{\player}\p{\type_{\player},\type_{-\player}} = \quantity_{\player}\p{\type_{\player},\type_{-\player}'} \implies \transfer_{\player}\p{\type_{\player},\type_{-\player}} = \transfer_{\player}\p{\type_{\player},\type_{-\player}'}. \]
\end{lemma}

The proof of \cref{lem:ssp_pab} and all other results can be found in the Appendix.
Observe that \cref{lem:ssp_pab} holds (as does \cref{thm:dutch_ssp}) even if we relax the assumption on shill bidders' information sets because \strong~shill-proofness means that shill bidders want to report $0$ even if they know the precise valuations of other bidders ex-ante.
To see why \cref{lem:ssp_pab} is true, consider the case where $\realBidders = \cbr{\player}$. Then, the shill bidders will report whichever values maximize $\transfer_{\player}$ and so $\transfer_{\player}$ must be constant across all outcomes with the same allocation.

Recall that with continuous types, two bidders have the same type with probability $0$ and the optimal allocation rule is uniquely defined up to measure-zero events.
However, since types are discrete, the probability of ties is non-zero and so we have to define a tie-breaking rule.
We adopt the notion of orderliness introduced by \citet{akbarpourLi20}: there exists a fixed priority order---independent of values---over which bidder wins an item if there is a tie.\footnote{For example, if ties are broken lexicographically, then the auction is orderly.}
We clarify their definition by explicitly linking the priority order to the ex-post allocation rule.
\begin{definition} \label{def:orderly}
	An auction equilibrium $(\game,\strategy)$ is \textbf{orderly} if there exists a total ordering $\ordering$ over $(\type_i,i)$ with the following properties: For all $i$ and $j$,
	\begin{enumerate}[(i)]
		\item $\type_i > \type_j \implies (\type_i,i) \ordering (\type_j,j)$;
		\item if there exists $m$ such that $(\type^m,i) \ordering (\type^m,j)$, then for all $k$, $(\type^k,i) \ordering (\type^k,j)$; and
		\item for all $\type$, if $\quantity_i(\type) \in \cbr{0,1}$ and $(\type_i, i) \ordering (\type_j,j)$, then $
	\quantity_i(\type) > \quantity_j(\type)$.
	\end{enumerate}
\end{definition}

We can now define the orderly allocation rule for any auction with reserve type $\reserve$ as
\[ \quantity^{\reserveIndex}_i(\type) = \mathds{1}\cbr{\type_i \ge \reserve, \p{\type_i,i} = \max_{\ordering} \cbr{\p{\type_j,j}}_{j \in \potentialBidders}}. \]
Then, the orderly, efficient allocation rule is $\quantity^E = \quantity^1$ and the orderly, optimal allocation rule is $\quantity^* = \quantity^{\reserveIndex^*}$ for some $\reserveIndex^*$.

We can now explicitly define the revenue-maximizing pay-as-bid bidding functions. 
Note that we need to select among multiple IC bidding functions because the discrete type space implies that the IC constraints need not bind with equality.
Using \cref{lem:ex_int_funcs} from the Appendix, the bidding function for bidder $i$ is given by
\[ b^1_i(\type^m) = \type^m - \sum_{k: \type^k < \type^m} (\type^{k+1} - \type^k)\frac{\p{\typeDist(\type^{k})}^{i-1}\p{\typeDist(\type^{k-1})}^{\numBidders-i-1}}{\p{\typeDist(\type^{m})}^{i-1}\p{\typeDist(\type^{m-1})}^{\numBidders-i-1}}. \]
The transfer rule is then $\transfer^E = b^1 \cdot \quantity^E$ and $\transfer^* = b^1 \cdot \quantity^*$ for the efficient and optimal auctions, respectively.
Observe that the asymmetry in the bidding and transfer functions arises purely from the discrete types and the tie-breaking rule, not from the definition of shill-proofness.

\subsection{Indirect Mechanisms}
Now that we have pinned down the direct mechanism, we turn to the indirect implementation.
The information structure available to real bidders matters for which mechanisms are strongly shill-proof.
If there is no information leakage between bidders, then the classic first-price auction is strongly shill-proof.
However, if any bidder can learn any outcome-relevant information about another bidder's actions, then the first-price auction is not strongly shill-proof.
In \cref{ss:dutch}, we return to the primary setting studied in this work, public auctions.

To expand, the single-action, first-price auction where each bidder sequentially makes a bid will be strongly shill-proof when the game is static, i.e., when each bidder takes her single action with no information about what other bidders have done.
However, we can show that the single-action, first-price auction is strongly shill-proof only in static games---if there is ever a chance that a bidder learns some outcome-informative information about what actions previous bidders have taken, then the auction is no longer strongly shill-proof.
Formally, in an orderly auction, for some $i,j$, and $\type$ with $(\type^{\numPossVals},j) \ordering (\type_i,i) \ordering (\optimalReserve,j)$, let $h_{\type}$ be the (unique) history reached in equilibrium by $\sigma(\type; \potentialBidders)$ where $P(h) = j$ and let $h_0$ be the corresponding history reached by types $(0,\type_{-i})$.
We say that an information structure is \textbf{strictly more informative} than a static game if there exist $i$, $j$, and $\type$ such that $\informationSet(h_{\type}) \ne \informationSet(h_0)$ and $\quantity_j(\type) = 1$.
Note that if there exists such a type vector, then there must exist a type vector where $\quantity_j(\type) = 0$.
\begin{proposition}\label{prop:fpa_ssp}
	Under a static information structure, the single-action, first-price auction is strongly shill-proof.
	However, for any information structure strictly more informative than a static game, the single-action, first price auction is not strongly shill-proof.
\end{proposition}

A core intuition underlying \cref{prop:fpa_ssp} is that the stronger a bidder perceives her competition, the less she shades her bid in a first-price auction.
So, if a shill bidder $i$ is able to distinguish that she has a value that could change the outcome of the auction to bidder $j$, then there must be some value $\type_j$ under bidder $j$ perceives the auction as more competitive; this is occurs if and only if there is outcome-relevant information leaked from bidder $i$ to $j$.
This will increase revenue if $j$ wins and so the first-price auction is not strongly shill-proof.

As noted in \cref{sec:model}, implementing no information leakage (or, equivalently, simultaneity) between bidders is often an unrealistic assumption.
Even with some technology that might deter information leakage in the actual mechanism, information leakage through side channels is inevitable in many settings.
So, we take \cref{prop:fpa_ssp} as evidence therefore that the first-price auction is, in many settings, potentially subject to manipulation by shill bidders.

\subsection{Dutch Auctions}\label{ss:dutch}
Now that we have shown that the extensive form of an auction affects whether or not it is strongly shill-proof, we turn to studying which extensive forms are always strongly shill-proof.
We show that Dutch auctions are strongly shill-proof, no matter the information structure, and show that that there exists a natural information structure---the public setting---such that Dutch auctions are uniquely strongly shill-proof.

To begin, we formally define the public information structure:
\begin{definition}\label{def:public}
	An auction equilibrium $\p{\game,\strategy}$ is \textbf{public} if the information set at any history is all previous actions taken.
	Formally, for any history $\history$, $\informationSet(\history) = \cbr{h}$.
\end{definition}

Public auctions are common in practice: from open air fish markets, to auctions on eBay, participants often can see every action other bidders take before choosing what to do.\footnote{Non-examples of public auctions include the FCC spectrum auctions, where bidders typically only learn information on other bidders' actions in rounds (see \citet{milgromSegal17} for more information).}
A second interpretation of an auction being public is that all information is totally leakable: If an auction is strongly shill-proof and public, then shill bidders can credibly signal to all other bidders any actions they have taken in the auction and it will not affect the outcome.
In general, we believe that it is natural to assume that shill bidders actions are visible in settings where considering shill bidding is important.
Otherwise, shill bidding can only manipulate the outcome ex-post, which differs from the normal interpretation of shill bidding and its applications.

Next, we define the Dutch auction with reserve price $\reserve$.
The Dutch auction is defined as the auction which begins by offering each bidder $i$ the item at $b^1_i(\type^{\numPossVals})$, and then if no bidder chooses to buy the item at that price, the item is offered for $b^1_i(\type^{\numPossVals-1})$ and so on until either a bidder has chosen to buy the item or the price to be offered drops below $b^1_i(\reserve)$.
Note that the optimal Dutch auction is the Dutch auction with reserve price $\optimalReserve$.
We consider only orderly auctions and therefore, at each price level, bidders are offered the opportunity to buy the item in priority order.
Formally:
\begin{definition}\label{def:dutch}
	The \textbf{Dutch auction with reserve price $\reserve$} is defined by the allocation rule $\quantity^{\reserveIndex}$, first-price transfer rule $\transfer^1 = \quantity^{\reserveIndex} \cdot b^1$, initial player $\p{\cdot,\player_0} = \max_{\ordering} \cbr{\p{0,i}}$, and menu
	\begin{align*}
		\menuRule^D_{\reserveIndex}(\possibleValues,\player)& = \cbr{\p{\choice_{\mathsf{L}},\nextPlayer_{\mathsf{L}}},\p{\choice_{\mathsf{H}},\nextPlayer_{\mathsf{H}}}}, \\
		\text{where } &\choice_{\mathsf{H}} = \cbr{\highValue_{\player}}, \choice_{\mathsf{L}} = \possibleValues_{\player}\backslash\cbr{\highValue_{\player}}, \nextPlayer_{\mathsf{H}}=\emptyset, \text{ and } \\
		&		\nextPlayer_{\mathsf{L}} = \begin{cases}
			\p{\cdot,\tilde{\player}} = \max_{\ordering} \cbr{\p{\highValue_i,i} : i \ne \player} & \exists i \ne \player \st \card{\possibleValues_i} > 1 \text{ and } \highValue_{i} \ge \reserve \\
			\emptyset & \ow
		\end{cases}.
	\end{align*}
\end{definition}

\begin{theorem} \label{thm:dutch_ssp}
	A Dutch auction with any reserve price is \strong~shill-proof.
	Furthermore, if a public, orderly and optimal auction is strongly shill-proof, then it is the Dutch auction with reserve price $\optimalReserve$.
\end{theorem}

\begin{proof}[Proof Sketch.]
The Dutch auction is \strongly~shill-proof because any shill bid immediately ends the auction---and in that case there would be no transfers from other bidders.
To gain an intuition for why uniqueness holds, note that the key defining property of the Dutch auction is that any bid immediately ends the auction.
Then, we prove the result in four cases.
First, we consider an auction where the player rotation differs from a Dutch auction; we then show that if a shill bidder indicates she has the highest possible remaining value, this has the effect of ex-interim increasing the effective reserve price.
Since the auction is public and only the highest value is allocated, when the types are realized so that some bidder with a higher priority order has the same value (and non-shill bidders have a low value), raising the reserve must increase revenue from the winning bidder.\footnote{If the auction were not public, then real bidders' beliefs might not change ex-interim and so their transfers might not change either. For example, in a sealed first-price auction, a shill bidder's actions have no effect on the transfers of other bidders.}
The next case considers what happens when bidders are queried in priority order, and the minimum of a partition is greater than the reserve price, but less than the highest value of other bidders.
This raises the ex-interim reserve price and so by the same argument as above, shill bidding is profitable.
In the third case, we consider what happens when the minimum in the partition is equal to the highest value of other bidders.
By construction, this can only occur when a bidder is the lowest priority and so, applying our orderliness assumption, this raises the ex-interim reserve to be the highest type still possible in equilibrium for other bidders without immediately allocating the good to the shill bidder.
We can then apply the same argument as the first two cases to conclude that shill bidding would be profitable in this case.
In the final case, we consider an auction where the minimum in the partition is below the reserve price.
In this case, the ex-interim reserve has not changed, so the logic from the previous cases does not apply.
Instead, shill bidding in this case makes it appear as if there is more competition for the item, and we prove that this causes ``less bid shading,'' i.e., higher final transfers from the winner.
\end{proof}

Observe that although the standard intuition about shill bidding is that its' purpose is to influence other bidders' perception about the common value for an item, \cref{thm:dutch_ssp} shows that just being able to influence other bidders perception about the probability they win the item without directly winning the item \textit{in a single case} is enough to limit the number of possible auction formats to just Dutch auctions.
Moreover, under public information, such a case always exists for non-Dutch auctions.
The public information structure is the most informative, while a static information structure is the least informative.
We show in the Supplemental Appendix that as the information structure becomes less informative, the number of strongly shill-proof auctions weakly increases (\cref{prop:info_shill_order}).

\section{Weakly Shill-Proof Auctions}\label{sec:wsp}
In \cref{lem:ssp_pab}, we demonstrate that strong shill-proofness uniquely pins down the transfer rule; and in \cref{thm:dutch_ssp}, we demonstrate that once we consider public auctions, we uniquely pin down the extensive form as well.
In this section, we turn to studying our more permissive definition, weak shill-proofness, and first show the limits of possible extensive forms that are weakly shill-proof and then exhibit extensive-form auctions that are weakly shill-proof and strategy-proof, including a new auction format that can be arbitrarily faster than the English auction. 

\subsection{Trilemma}\label{ss:single_action}
In this subsection, we prove that it is impossible to find an auction that is weakly shill-proof, strategically simple for even one of the real bidders, and \textit{fast} in the sense that each bidder takes exactly one action.
To make this trilemma as general as possible, we do not even require the auction to be optimal or efficient.
Instead, we are only using condition (iii) of our orderliness definition (\cref{def:orderly}), which requires that a bidder who does not have the highest type must not win the item.
Note that such a condition is without loss for any symmetric and deterministic mechanism.

We begin by defining a \textbf{single-action auction}. 
An auction is considered single-action when each bidder takes precisely one action in the auction (under all possible histories). 
More formally, for any $\history_{\numBidders} \in Z$, let $\history_{\emptyset} \prec \history_1 \prec \ldots \history_{N-1} \prec \history_{\numBidders}$ be the sequence of preceding histories. 
Then, an auction is \textbf{single-action} if for all $i \in \potentialBidders$, there exists a unique $n \le \numBidders$ such that $i = \playerFunction(\history_{n})$.
Without loss, we label the bidders $1,\ldots,\numBidders$, in the order that they move and label the action taken by bidder $i$ as $\action_i$.\footnote{The bidder ordering and therefore the labeling can be endogenous to actions taken.}

For expositional purposes, instead of tracking the information set $\informationSet_i$ of a bidder $i \in \realBidders$, we assume that the information $i$ has when taking an action is a signal $\signal_i \in \signalSet_i$.
This signal is generated via a deterministic function $\experiment_i: \p{\bigtimes_{j < i} \actionSet_j} \to \signalSet_i$ called an \textbf{experiment}.\footnote{Abusing notation, we also sometimes take $\experiment_i: \typeSpace^{i-1} \to \signalSet_i$, i.e., the experiment maps values to signals instead of actions.}
For notational convenience, we assume that $\experiment_i$ is surjective for all $i$.
We can think of the experiment as a garbling of the previous bidders' actions---the experiment can pool together multiple actions from previous bidders to a single signal and so a signal is not always perfectly informative of previous actions.
We can recover the public setting with a fully informative experiment, i.e., $\publicExperiment$, the identity mapping.
We can capture classical static game settings via an uninformative experiment that always returns the same output, $\privateExperiment$.
We use $\experiment^{-1}_i(\signal_i)$ to denote the set of $\type_{-i}$ that are possible from the perspective of bidder $i$ given its signal. 

A revelation principle holds in this setting: for any single-action auction, we can define the direct allocation and transfer rules as $\quantity(\type)$ and $\transfer(\type)$, respectively, with the appropriate incentive compatibility and individual rationality constraints for real bidders (\cref{lem:one_shot_rev_principle} in the Appendix), and appropriate IC constraints for \weak~and \strong~shill-proofness (\cref{lem:wsp_val_equiv_general,lem:ssp_val_equiv}, respectively, in the Appendix).

To finish defining all the terms necessary for our main result of this section, we formally define ex-post incentive compatibility and then give a weaker notion of \strategyproofness~in the single-action auction setting: \strategyproofness~for at least a single bidder.
\begin{definition} \label{def:strategy_proof}
	An auction $\directAuction$ is \textbf{\strategyproof} if it is an ex-post strategy for real bidders to report their values truthfully: for all real bidders $i \in \realBidders$, $\type$, and $\type_i'$,
	\[ \quantity(\type) \cdot \type_i - \transfer(\type) \ge \quantity(\type_i',\type_{-i}) \cdot \type_i - \transfer(\type_i',\type_{-i}). \]
\end{definition}
\begin{definition} \label{def:mild}
	A single-action auction $\directAuction$ is \textbf{mildly \strategyproof} if there exists a real bidder $i < \numBidders$ such that truthfulness is an ex-post strategy for $i$ conditional on the realization of her signal: there exists bidder $i < \numBidders$, such that for all $\type_i$, $\type_i'$, $\signal_i$, and $\type_{-i},\type_{-i}' \in \experiment_i^{-1}(\signal_i)$: $ \quantity_i(\type) \cdot \type_i - \transfer_i(\type) \ge \quantity_i\p{\type'} \cdot \type_i - \transfer_i\p{\type'}$.\footnote{We exclude the last bidder who takes an action from our definition because a take-it-or-leave-it offer to that bidder can be optimal and \strategyproof. \cref{thm:wsp_trilemma} would still hold if we instead defined mild \strategyproofness~to mean \strategyproof~for at least two bidders.}
\end{definition}

\begin{theorem}	\label{thm:wsp_trilemma}
	There exists no orderly auction that is single-action, mildly \strategyproof,~and \weakly~shill-proof. 
\end{theorem}
\begin{proof}[Proof Sketch]
	Consider any real bidder $i < \numBidders$.
	By \weak~shill-proofness, the transfer from bidder $i$, conditional on winning (or losing) the auction, is invariant to the values of bidders who take actions after her (\cref{lem:wsp_form} in the Appendix). 
	If this were not the case, then if every bidder $j>i$ is a shill bidder, the shill bidders would report the values that would maximize the transfer from the winning bidder. 
	By mild \strategyproofness, the transfer from bidder $i$, conditional on winning the auction, is invariant to her value (\cref{lem:msp_form} in the Appendix).
	This is because if there were multiple winning reports with different transfer amounts, only the smallest transfer amount would make truthful reporting of the value an ex-post strategy. 
	So, in every single-action, optimal auction, the transfer from the winning bidder $i$, can depend only on the values reported by bidders before $i$.
	But, this means that if a bidder has positive utility for winning the item (as would be the case if $\type_i > \type_j$ for all $j < i$), then she should report $\type^{\numPossVals}$ to maximize the probability of winning (without changing the transfer paid upon winning).
	Thus, the auction must treat bidder $i$ as if she reported $\type^{\numPossVals}$, which violates the allocation rule of an optimal auction.
\end{proof}

Observe that while there are no auctions at the intersection of all three conditions in \cref{thm:wsp_trilemma}, we can find multiple auctions at the intersection of any two of those conditions.
For example, assuming an IPV and regular (\cref{def:regular}) environment, the second-price auction is single-action and mildly \strategyproof; the first-price auction is single-action and \weakly~shill-proof; and the English auction is mildly \strategyproof~and \weakly~shill-proof.
None of the examples just described are unique; there exist other single-action and mildly \strategyproof~auctions, single-action and weakly shill-proof auctions, and mildly \strategyproof~and \weakly~shill-proof auctions.

\subsection{Other Weakly Shill-Proof Auctions}
Now that we have shown an impossibility result regarding weakly shill-proof auctions, in this subsection, we will provide some characterization results for weakly shill-proof auctions.
Recall from the introduction that even in an IPV environment, an English auction need not be weakly shill-proof when the type distribution is irregular.
So, for tractability, for the rest of this subsection, we suppose that bidders have IPV types and regular type distributions.
We use the definition of regularity for discrete types found in \cite{elkind07}:
\begin{definition}\label{def:regular}
	A distribution $F$ is \textbf{regular} if for all $k$, the virtual value
	\iftoggle{compressLines}{$\varphi^k = \type^k - (\type^{k+1} - \type^k) \frac{1 - \typeDist(\type^k)}{\typepmf(\type^{k})}$}{\[ \varphi^k = \type^k - (\type^{k+1} - \type^k) \frac{1 - \typeDist(\type^k)}{\typepmf(\type^{k})} \]}
	is non-decreasing.
\end{definition}

\subsubsection{\Weakly~Shill-Proof and Efficient Auctions} \label{ss:wsp_dutch_auction}
To further analyze when concerns about shill-bidding lead to Dutch auctions, we now examine the case of efficient auctions.
Shill bidding is important to consider in the efficient auction context because the designer may be interested in allocating goods efficiently even while sellers are trying maximize revenue.
For example, in two-sided marketplaces like eBay, the auctioneer/market designer places positive welfare weight on both buyers and sellers, whereas sellers typically do not---and sellers may certainly try to shill bid in these settings, as discussed in the Introduction.
Our next result shows that in order for an auction to be robustly \textit{\weakly} shill-proof and efficient, part of its game tree must be a Dutch auction.
If an auction is not \textit{robustly} weakly shill-proof, we mean that we can find a value distribution such that the auction is not \weakly~shill-proof.
More formally, if the auction $\directAuction$ is parameterized by the optimal reserve $\optimalReserve$, the number of atoms below the reserve $\underline{\numPossVals}$ and the number of atoms weakly above the reserve $\overline{\numPossVals}$,\footnote{Here, $\underline{\numPossVals}+\overline{\numPossVals} = \numPossVals$.} then the following result holds.
\begin{definition}\label{def:semi_dutch}
	An efficient auction $\effDirectAuction$ is a \textbf{semi-Dutch auction with cutoff $\optimalReserve$} if 
	for any $\val$ such that $\max_i\cbr{\val_i} < \optimalReserve$:
	\begin{enumerate}[(i)]
		\item $\widecheck{\possibleValues} = \cbr{w : w < \optimalReserve}^{\numBidders}$ is reached; and \label{cond:semi_dutch_cutoff}
		\item $\menuRule(\possibleValues,\player) = \menuRule^D(\possibleValues,\player)$ for any player $\player$ and possible values $\possibleValues \subseteq \widecheck{\possibleValues}$ where $\menuRule^D$ is the menu rule for the Dutch auction with reserve price $0$ from \cref{def:dutch}. \label{cond:semi_dutch_menu}
	\end{enumerate}
\end{definition}

\begin{proposition} \label{prop:dutch_robust_wsp}
	For every public and efficient auction that is not a semi-Dutch auction with cutoff $\optimalReserve$, there exists a regular value distribution with optimal reserve $\optimalReserve$ under which the auction is not \weakly~shill-proof.
\end{proposition}

The key step in the proof of \cref{prop:dutch_robust_wsp} resembles the proof of \cref{thm:dutch_ssp}---in any non-Dutch auction, shill bidders can ex-interim ``raise the reserve price'' by changing their actions.
However, given that we are interested in \weak~shill-proofness instead of \strong~shill-proofness, we have to examine shill bidders' incentives when we take expectations over real bidders' values instead of conditioning directly on their values.
Regularity implies that above~$\optimalReserve$, shill bidders do not have an incentive to shill bid in auction formats such as the English auction (see \cref{ss:shill}).
However, below $\optimalReserve$, we can always find a regular distribution such that the ex-interim expected value of raising the reserve price is always positive; in particular, we can find a value distribution where the atoms are far enough apart that raising the reserve price a single ``level'' generates a large amount of additional revenue.
In the Appendix, we construct the claimed sub-class of regular distributions (see \cref{def:sparse}).
So, below~$\optimalReserve$, the auction must resemble the Dutch auction; the class of all such auctions is precisely all semi-Dutch auctions with cutoff $\optimalReserve$.
In the Supplemental Appendix (\cref{ex:robust_wsp}), we provide an example of an auction format that is not semi-Dutch and that is weakly shill-proof for some value distributions but not others; and the following example presents a real-world setting that roughly fits the premises of \cref{prop:dutch_robust_wsp} where a semi-Dutch auction (that is not a Dutch auction) is used.

\begin{example}
	The Honolulu-Sydney fish auctions and Istanbul flower auctions documented by \citet{hafalirEtAl23} blend elements of the Dutch and English auctions:
	The auction begins at some intermediate price and if anyone bids, then the price ascends like in the English auction. 
	If no one bids, the price descends until someone bids like in the Dutch auction.\footnote{Honolulu-Sydney auction, once someone bids, other bidders can counter-bid and raise the price once more. However, in practice there is little counter-bidding. On the theoretical side, in an IPV setting, there exists an equilibrium where there is no counter-bidding. Counter-bidding once the Dutch auction starts is not allowed in the Istanbul flower auction.}
	
	The Honolulu-Sydney auction plausibly fits the technical assumptions made in \cref{prop:dutch_robust_wsp}:
	The auctions are public, as they take place in person and all bidders can see other bidders' actions.
	Market participants are interested in efficient outcomes because the goods are perishable and there are positive disposal costs for the sellers.
	We do not mean to imply that the Honolulu-Sydney auction was instituted precisely because it is shill-proof, but we highlight it as further evidence that in markets where it is difficult to monitor shill bidding, shill-proof mechanisms may arise.
\end{example}

Note that \cref{prop:dutch_robust_wsp} relies on the auction format being able to condition on the true, optimal reserve.
If we instead require the auction format to be completely prior-independent, then the only public, efficient, and \weakly~shill-proof auction is the Dutch auction.

\begin{corollary} \label{cor:dutch_robust_wsp}
	For any public and efficient auction that is not a Dutch auction, there exists a regular value distribution under which the auction is not \weakly~shill-proof.
\end{corollary}
\begin{proof}
	Observe that a semi-Dutch auction with cutoff $\optimalReserve = \type^{\numPossVals}$ is simply a Dutch auction.
	Then, apply \cref{prop:dutch_robust_wsp} for a regular distribution with optimal reserve $\optimalReserve = \type^{\numPossVals}$.
\end{proof}

\subsubsection{\Weakly~Shill-Proof and Strategy-Proof Auctions} \label{ss:strategy_proof}
We have shown that the only optimal auction with a feasible equilibrium for real bids and an ex-post strategy for shill bidders not to shill (\strong~shill-proofness) is the Dutch auction.
We now investigate the dual question: what optimal auctions have an ex-post strategy for real bidders (\strategyproofness) and an equilibrium under which no shill bidding occurs (\weak~shill-proofness)?

An optimal auction is \strategyproof~if and only if has the second-price transfer rule \cite[Proposition 8]{akbarpourLi20}:
\[ \transfer^2_i(\type) = \quantity_i^*(\type) \cdot \max\cbr{\optimalReserve,\text{second-highest value in $\{\type_1,\ldots,\type_{\numBidders}\}$}}. \]
Note that if a shill bidder knew the valuations of all other bidders, then shill bidding would turn a second-price auction into a first price auction, which bounds the expected profit for a shill bidder from shilling.
So, in order to find a \strategyproof~and \weakly~shill-proof auction, we must find a menu rule that implements a second-price auction where the expected gain from shill bidding is sufficiently small at shill bidders' information sets.
As discussed in \cref{ss:shill}, for regular value distributions, one \weakly~shill-proof, \strategyproof, and optimal auction is the English auction.
We formalize the English auction in our framework as follows:

\begin{definition}
	The \textbf{English auction with reserve price $\optimalReserve$} is defined as the auction with the optimal allocation rule $\quantity^*$, second-price transfer rule $\transfer^2$, initial player $\p{\cdot,\player_0} = \min_{\ordering} \cbr{\p{0,i}}_{i \in \potentialBidders}$, and menu
	\begin{align*}
		\menuRule^{E}(\possibleValues,\player) &= \cbr{\p{\choice_{\mathsf{L}},\nextPlayer_{\mathsf{L}}},\p{\choice_{\mathsf{H}},\nextPlayer_{\mathsf{H}}}}, \\
		\text{where } &\choice_{\mathsf{L}} = \cbr{\type \in \possibleValues_{\player}: \type < \optimalReserve} \cup \cbr{\lowValue_{\player}}, \choice_{\mathsf{H}} = \possibleValues_{\player} \backslash \choice_{\mathsf{L}}, \\
		&\nextPlayer_{\mathsf{L}} = \nextPlayer_{\mathsf{H}} = \begin{cases}
			\p{\cdot,\tilde{\player}} = \min_{\ordering} \cbr{\p{\lowValue_i,i} : i \ne \player, \highValue_i = \type^{\numPossVals}, \card{\possibleValues_i} > 1} & \highValue_{-\player} = \type^{\numPossVals} \\
			\emptyset & \ow
		\end{cases}.
	\end{align*}
\end{definition}
\begin{remark}
	The English auction with reserve price $\optimalReserve$ is \weakly~shill-proof, \strategyproof, and optimal (when the value distribution is regular).
	Depending on the information bidders have when taking actions, that auction can also be made dominant-strategy.\footnote{See \cref{ss:xp_v_ds} for a discussion of how our results translate into dominant-strategy contexts.}
\end{remark}

The English auction is not the only \strategyproof~and \weakly~shill-proof auction.
While the English auction is used frequently, one drawback is that it is ``slow"---each bidder can be queried on their willingness-to-pay on the order of $\numPossVals$ times.
Specifically, let $\worstCaseQueries^{E}(\typeDist) = M - \reserveIndex^* + 1$ be the worst-case number of times a bidder must be queried.
To explore whether there are \weakly~shill-proof and \strategyproof~auctions that require fewer rounds of communication, we introduce a natural ``compression'' of the English auction that comprises the following two phases:
\begin{enumerate}
	\item An English auction is run from $\optimalReserve$ to some $\screenVal$.
	\item If necessary, a second-price auction is then run among players who have not dropped out before the value level of $\screenVal$.
\end{enumerate}

\begin{definition}
	The \textbf{ascending, screening auction} with screen level $\screenVal$ is defined by the optimal allocation rule $\quantity^*$, second-price transfer rule $\transfer^2$, initial player $\p{\cdot,\player_0} = \min_{\ordering} \cbr{\p{\highValue_i,i}}$, and menu
	\begin{align*}
		&\menuRule(\possibleValues,\player) = \begin{cases}
			\menuRule^{E}(\possibleValues,\player) & \exists i \st \lowValue_i \le \screenVal \text{ and } \highValue_i = \type^{\numPossVals} \\
			\cbr{\p{\cbr{\type^k},\nextPlayer^{\, k} }}_{k \in \cbr{\screenValSuper+1,\ldots,\numPossVals}} & \ow
		\end{cases}, \\
		&\text{where } \p{\cdot,\nextPlayer^{\, k} } = \max_{\ordering}\cbr{(0,i) : \card{\possibleValues_i} > 1, \lowValue_i = \screenVal} \text{ for $k < \numPossVals$ and } \nextPlayer^{\, \numPossVals} = \emptyset.
	\end{align*}
\end{definition}
The ascending, screening auction reduces the maximum number of times each bidder can be queried to
\iftoggle{compressLines}{$\worstCaseQueries^{AS,\screenValSuper}(\typeDist) = \screenValSuper - \reserveIndex^* + 2.$}{\[ \worstCaseQueries^{AS,\screenValSuper}(\typeDist) = \card{\cbr{k : \optimalReserve \le \type^k \le \screenVal}} + 1. \]}
Because the transfer rule is $\transfer^2$, the ascending, screening auction is \strategyproof~and optimal.

We use the ascending, screening auction format to explore how fast a \weakly~shill-proof, \strategyproof~and optimal auction can be, as a function of the underlying value distribution.
Our next result shows that the ascending, screening auction can be \weakly~shill-proof, \strategyproof, optimal, and take an arbitrarily small fraction of queries as compared to the English auction depending on the value distribution:
\begin{proposition} \label{prop:wsp_sp_upper_bound}
	For all $\eps > 0$, there exists a value distribution $\valDist$ and screen level $\screenVal$ such that $\worstCaseQueries^{AS,\screenValSuper}(\valDist)/\worstCaseQueries^E(\valDist) < \eps$ and the ascending, screening auction with screening level $\screenVal$ is \weakly~shill-proof, \strategyproof, and optimal.
\end{proposition}
The ascending, screening auction is orderly and optimal by construction; and it is \strategyproof~because the English auction phase and the second-price phase both induce the same (\strategyproof) allocation and transfer rule.
The larger $\screenVal$ is, the less that can be extracted in expectation from shill bidding and the more likely it is that a shill bidder will win the item if she shill bids.
We provide a sufficient minimum bound on $\screenVal$ based on a few moments of a distribution (not its number of atoms), such that for distributions with ``thin-enough'' right tails---in particular monotone hazard rate distributions---the ascending, screening auction is weakly shill-proof (\cref{lem:wsp_sp_bound}).
We can then construct a sequence of distributions with increasing numbers of atoms and constant $\screenVal$ to complete the proof.
The distributional requirements we need here differ from the class we find to prove \cref{prop:dutch_robust_wsp}, demonstrating that designing weakly shill-proof auctions has rich interactions with properties of the prior distribution.

\section{Affiliation, Interdependent Values, and Shill-Proofness}\label{sec:aff_and_sp}
In this section, we generalize our model to allow for interdependent values and a more general distribution of shill bidders.
We will then show that shill-proofness admits an order with respect to the type distribution and value function: the higher the level of affiliation and level of common values that is present in the environment, the harder it is for an auction format to be shill-proof, either weakly or strongly.
Next, we show that \cref{lem:ssp_pab} holds under our general model and pins down the optimal allocation rule for strongly shill-proof auctions in affiliated environments.
Finally, we show that our main results (\cref{thm:dutch_ssp,thm:wsp_trilemma}) still hold under our more general model.

\subsection{Generalized Model}
We maintain the assumptions that there is a set of potential bidders $\potentialBidders$, with $\card{\potentialBidders} = \numBidders$, who might participate and that the seller's value is commonly known to be $0$.
We also maintain our previous assumptions on real and shill bidder utilities and information sets.
However, instead of assuming that each potential bidder has probability $\pReal$ of participating, we instead assume that the set of real bidders is drawn symmetrically and randomly, i.e., for any $\realBidders',\realBidders'' \subset \potentialBidders$, if $|\realBidders'| = |\realBidders''|$, then $\prob\sbr{\realBidders = \realBidders'} = \prob\sbr{\realBidders = \realBidders''}$.
For parsimony, we also assume full support, but the probability of certain sets of real bidders can be arbitrarily small; this rules out mechanisms that condition on large coalition of bidders in unusual ways in order to become weakly shill-proof.\footnote{Our results for strong shill-proofness do not require a full support distribution over shill bidder sets.}
Letting the probability of certain realizations go to zero (or relaxing the full support assumption entirely), this treatment of the set of shill bidders nests many natural models of shill bidders.
For example, it nests a model of shill bidding where there is at most or exactly one shill bidder.
It also nests our model of shill bidding from \cref{sec:model} where each bidder has an i.i.d.\ chance of being a real bidder.

In our general model, we interpret the type $\type_i \in \typeSpace$ to flexibly represent both a private value for the item and a signal of the common value for the item.
We will expand on exactly how the type enters the valuation for real bidders below.
As notation, for each shill bidder $i \in \shillBidders$, we assign their type to be $\type_i = 0$ to represent them having the lowest possible valuation for the item.
Then, we can say types are jointly drawn from a full-support, symmetric $\typeDist$ and we define $\typepmf(\type) = \prob_{w \sim \typeDist}\sbr{w = \type}$ to be the pmf~of the distribution.
We assume that the type distribution is \textbf{affiliated} in the \cite{MW82} sense: For all $\theta, \theta' \in \typeSpace^N$,
\[ \log\typepmf(\theta \vee \theta') + \log\typepmf(\theta \wedge \theta') \ge \log\typepmf(\theta) + \log\typepmf(\theta'). \]
Affiliation captures the idea that types are positively correlated: the probability of types/signals all being high or low is higher than the probability of some signals being high and some being low.

There is a value function $\val: \typeSpace \times \typeSpace^{\numBidders-1} \to \R$ such that for each bidder $i$, her value for the item is $\val\p{\type_i, \cbr{\type_j}_{j \ne i}}$; note that this form implies that bidders have symmetric preferences over the types of other bidders.
We also make a few functional form restrictions on $\val$ below.
\begin{assumption}\label{assum:interdep}
	The value function $\val$ for a bidder satisfies the following properties:
	\begin{enumerate}[(i)]
		\item For all $k$, $\val(\type^k,0) = \type^k$;
		\item $\val$ is weakly increasing in $\type_{-i}$;
		\item If $\type_i \ge \type_j$, then $\val\p{\type_i, \cbr{\type_j, \type_{-(i,j)}}} \ge \val\p{\type_j, \cbr{\type_i, \type_{-(i,j)}}}$; and
		\item For all $\type$ such that $\type_i \ge \max_{j \ne i} \cbr{\type_j}$, $\val$ is weakly super-modular and has weakly decreasing differences.
	\end{enumerate}
\end{assumption}
Condition (i) is a (without loss) normalization conditional on the other parts of the assumption.
Condition (ii) means that $\type$ encodes common value preferences---a bidder values an item more when other bidders value it more.
Condition (iii) means that bidders value their own high-type realizations more than they value high-type realizations for other bidders. (This is the standard single-crossing assumption needed for a responsive, incentive compatible mechanism to exist.)
Condition (iv) means that conditional on having the highest ex-ante signal of value, a bidder has higher value for other bidders' signals whenever her type is high and has diminishing marginal returns to high signals.

Note that the standard, private values setting where for all $\type$, $\val\p{\type_i, \type_{-i}} = \type_i$ satisfies \cref{assum:interdep}.
Most other value functions in the literature also satisfy these conditions; for example, generalized, additive, interdependent values $\val(\type) = \type_i + \kappa \sum_{j \ne i} \type_j$ with $\kappa \le 1$ and maximum common values $\val(\type) = \max_i \cbr{\type_i}$.

\subsection{Shill-Proof Order}
To prove that increasing the level of affiliation and interdependence of values makes it more difficult for an auction to be shill-proof, let us begin by formally defining what it means for the value distribution to be ``more affiliated'' and for the value function to be ``more commonly valued.''

\begin{definition}\label{def:aff_order}
	Consider two distributions $\typeDist, \typeDist' \in \Delta(\typeSpace^{\numBidders})$ such that all marginal distributions are the same: $\typeDist_i = \typeDist'_i$, for all $i$.
	A distribution $\typeDist'$ is \textbf{more affiliated} than $\typeDist$, $\typeDist' \affOrder \typeDist$, if for all $x,y \in \typeSpace^{\numBidders}$, $\log\typepmf'(x \vee y) - \log\typepmf'(x) \ge \log\typepmf(y) - \log\typepmf(x \wedge y)$.
\end{definition}
Our definition of the affiliation order is standard and from \citet{karlinRinott80}.
When $\typeDist' \affOrder \typeDist$, values between bidders are more highly correlated under $\typeDist'$ than $\typeDist$.
We fix the marginal distributions to emphasize that our definition focuses on the correlation between bidders and not the relative strength of the distributions.
All affiliated distributions are more affiliated than the type distribution where bidders are independent: If $\typeDist$ is affiliated, then $\typeDist \affOrder \prod_i \typeDist_i$.

\begin{definition}\label{def:val_order}
	A value function $\val'$ is \textbf{more commonly valued} than $\val$, $\val' \valOrder \val$, if $\val'$ is weakly more super-modular than $\val$ and for all $\type$, $\val'(\type) \ge \val(\type)$.
\end{definition}

The point-wise comparison in \cref{def:val_order} makes sense because $\val, \val'$ are normalized so that $\val(\type^k,0) = \val'(\type^k,0) = \type^k$---and therefore the point-wise comparison means that other bidders' signals are valued more as the value function becomes more commonly valued.
The increase in super-modularity means that we also require types to be more complementary when we say that the value function is more commonly valued.
We view this as reasonable because if signals became less complementary, there would exist $\val' \valOrder \val$ such the marginal change in bidder $i$'s value based on her own type would be less under $\val'$ than under $\val$.
Any value function $\val$ satisfying \cref{assum:interdep} is more commonly valued than the value function $\val_{\mathrm{private}}(\type_i,\type_{-i}) = \type_i$ arising under the private values model.

Letting $\optTransfer$ be defined as some optimal transfer rule conditional on the extensive-form game and the primitives of the environment, we are now in position to formally state our main result for this section:
\begin{theorem}\label{prop:aff_shill_order}
	Consider affiliated type distributions $\typeDist$ and $\typeDist'$, and value functions $\val$ and $\val'$ satisfying \cref{assum:interdep}.
	Suppose $(\quantity,\optTransfer,\menuRule,\player_0)$ is orderly and strongly shill-proof. Then, if $\val' \valOrder \val$ and $\typeDist \affOrder \typeDist'$, it is the case that $(\quantity,\transfer^*(\quantity,\menuRule,\player_0,\val',\typeDist'),\menuRule,\player_0)$ is strongly shill-proof.
	The same statement holds for weak shill-proofness.
\end{theorem}
\begin{proof}[Proof Sketch.]
To begin the proof sketch, let us first observe that the ex-interim transfers from each bidder is pinned down based on the ex-interim allocation rule via a bidder's incentive compatibility constraint.
\begin{lemma} \label{lem:aff_ex_int_funcs}
	For every auction $(\quantity,\optTransfer,\menuRule,\player_0)$, the ex-interim transfer rule for bidder $i$ is
	\begin{gather*}
		\exinttransfer_i(\type_i;\possibleValues) = \exintquantity_i(\type_i;\possibleValues)\val^*(\type_i,\type_i) - \sum_{m : \type^{j_m} < \type_i}  \exintquantity_i(\type^{j_m};\possibleValues) \cdot (\val^*(\type^{j_{m}},\type^{j_{m+1}};\possibleValues) - \val^*(\type^{j_{m}},\type^{j_{m}};\possibleValues)), \\ 
		\text{ where } \val^*_i(\type^{k},\type^{k'};\possibleValues) = \expect\sbr{\val(\type^k,\type_{-i}) \mid \type_i = \type^{k'}, \type_{-i} \in \possibleValues_{-i}, \quantity_i(\type^k,\type_{-i}) = 1}
	\end{gather*}
	is the expected value of the item for bidder $i$ of type $\type^{k'}$ conditional on winning
	the item and playing as if she were type $\type^k$ and $\cbr{\type^{j_m}}_{m}$ are the ordered atoms of $\possibleValues_i$.
\end{lemma}
Note that feasibility normally pins down the transfer rule via the envelope theorem.
However, we are working with discrete types instead of continuous types and so the incentive compatibility constraint is slackened as the change in allocation probability from misreporting is discrete.
So, we use optimality to select among this multitude of ex-interim transfers that maintain incentive compatibility.
(This is an adaptation of the standard envelope theorem formulation of ex-interim transfers adapted to the extensive-form game by noting that if, at any point in the game, $\exinttransfer_i$ does not take this form, there exists a profitable deviation by bidder $i$ to commit to acting as a lower type throughout the rest of the game.)

Next, we prove the following technical lemma.
\begin{lemma}\label{lem:aff_set_order}
	Consider sets $S$ and $S'$ such that $S \cap S' = \emptyset$ and $\min_{v \in S} \cbr{v} > \min_{v \in S'} \cbr{v}$.
	Then, for any weakly super-modular, non-decreasing function $g$ with decreasing differences, if $\typeDist' \affOrder \typeDist$, then for all $i$,
	\begin{equation}\label{eq:aff_set_order}
		\expect_{\type \sim \typeDist'}\sbr{g(\type) \mid \type_i \in S} - \expect_{\type \sim \typeDist'}\sbr{g(\type) \mid \type_i \in S'} \ge \expect_{\type \sim \typeDist}\sbr{g(\type) \mid \type_i \in S} - \expect_{\type \sim \typeDist}\sbr{g(\type) \mid \type_i \in S'}.
	\end{equation}
\end{lemma}

The proof then proceeds by contrapositive for affiliation and then common values.
We consider that the auction is not weakly shill-proof for distribution $\typeDist$ and consider $\typeDist' \affOrder \typeDist$.
We show that if there exists a profitable deviation under $\typeDist$, then that same deviation must be profitable under $\typeDist'$.
Any profitable deviation either manipulates a bidder's ex-interim expected transfer or ex-post manipulates her payment.
Given that we have fixed the extensive form, any deviation that is possible under $(\val,\typeDist)$ is possible under $(\val',\typeDist')$.
So, ex-post manipulation will be profitable under both and so we only need to consider ex-interim manipulations.
We can re-write expected profit of deviation as the difference between the expectation of a sum taken over the true and manipulated type reports.
We show that the sum satisfies the conditions of \cref{lem:aff_set_order} and so we can conclude that deviating increases profits.
Note that the probability that bidder $i$ wins the item is constant between the two type distributions because the allocation rule is held constant and so we have shown that the auction is not weakly shill-proof under $\typeDist'$.
The argument when strongly shill-proof is the exact same.
The proof when considering common values is also by contrapositive.
In this case, we apply the fact that $\val'$ is more super-modular and point-wise larger to directly conclude that the profit from the deviation is larger.
\end{proof}

\subsection{Extending our Characterizations}
We can now use \cref{prop:aff_shill_order} to generalize our main characterization results from previous sections.
To begin, we observe that while we know the optimal allocation rule in an IPV environment and the revenue equivalence theorem implies that any transfer rule yields the same revenue, the same need not hold with interdependent values and affiliation.
So, we must simultaneously search for an allocation and transfer rule that satisfies our desiderata.

In order to find the optimal allocation and transfer rules, we first pin down the transfer rule and then the allocation rule.
Our proof of \cref{lem:ssp_pab} holds in our general environment and so strong shill-proofness implies that the transfer rule must be pay-as-bid.
Then, having pinned down the transfer rule, we can see first that the unconstrained efficient allocation is feasible, and also that in the optimal auction we will never want to allocate to a bidder who does not have the highest type.
Thus, we can solve the monopolist screening problem to find the optimal reserve in this auction.
The formal result is as follows.
\begin{proposition}\label{prop:ssp_pab}
	Under any affiliated distribution and value function satisfying \cref{assum:interdep}, if an auction $\directAuction$ is strongly shill-proof, then it is pay-as-bid:
	For all $\player$, $\val_{\player}$, $\type_{-\player}$, and $\type_{-\player}'$,
	\[ \quantity_{\player}\p{\type_{\player},\type_{-\player}} = \quantity_{\player}\p{\type_{\player},\type_{-\player}'} \implies \transfer_{\player}\p{\type_{\player},\type_{-\player}} = \transfer_{\player}\p{\type_{\player},\type_{-\player}'}. \]
	All efficient mechanisms have an allocation rule such that
	\[ \sum_i \quantity^E_i(\type) = 1 \text{ and } i \notin \argmax_{j} \cbr{\type_j} \implies \quantity^E_i(\type) = 0. \]
	All optimal mechanisms have an allocation rule such that
	\[ \sum_i \quantity^*_i(\type) = \begin{cases}
		1 & \max_j \cbr{\type_j} \ge \optimalReserve \\
		0 & \ow 
	\end{cases} \text{ and } i \notin \argmax_{j} \cbr{\type_j} \implies \quantity^E_i(\type) = 0 \]
	for some $\optimalReserve$.
\end{proposition}

We can now generalize \cref{thm:dutch_ssp}.
Modifying the transfer rule in the Dutch auction to be the appropriate pay-as-bid rule, the following corollary holds.
\begin{corollary}\label{cor:gen_dutch_ssp}
	Under any affiliated type distribution and value function satisfying \cref{assum:interdep}, the statement of \cref{thm:dutch_ssp} holds.
\end{corollary}
\begin{proof}
	The proof that any Dutch auction is strongly shill-proof is the exact same as in the IPV case.
	By \cref{prop:ssp_pab}, we know what the orderly and optimal allocation rule is.
	By \cref{thm:dutch_ssp}, we know that in the IPV case, no non-Dutch auction will be strongly shill-proof.
	We then apply \cref{prop:aff_shill_order} to conclude that all non-Dutch auctions are not strongly shill-proof and optimal in general.
\end{proof}

Many of our results for weakly shill-proof auctions continue to hold in our general model when we restrict to optimal transfer rules.
The generalization of \cref{thm:wsp_trilemma} is as follows.
\begin{corollary}
	Under any affiliated type distribution $\typeDist$ and value function $\val$ satisfying \cref{assum:interdep}, there exists no orderly auction $(\quantity,\optTransfer,\menuRule,\player_0)$ that is single-action, mildly ex-post incentive compatible and weakly shill-proof.
\end{corollary}
\begin{proof}
	Applying \cref{prop:aff_shill_order} to \cref{thm:wsp_trilemma} with a restriction on possible transfer rules completes the proof.
\end{proof}
While regularity is not as applicable with affiliation, our robustness result (\cref{cor:dutch_robust_wsp}) for the Dutch auction still hold.
\begin{corollary}
		Under any affiliated type distribution $\typeDist$ and value function $\val$ satisfying \cref{assum:interdep} and for any public and efficient auction $(\quantity,\optTransfer,\menuRule,\player_0)$, if the auction is not a Dutch auction with reserve price $0$, there exists a value distribution under which the auction is not weakly shill-proof.
\end{corollary}
\begin{proof}
	Since the Dutch auction with reserve price $0$ is strongly shill-proof, it must be weakly shill-proof.
	We then apply \cref{prop:aff_shill_order} to \cref{cor:dutch_robust_wsp} to rule out any other auction format being robustly weakly shill-proof and complete the proof.
\end{proof}

We conclude this section by observing that \cref{prop:wsp_sp_upper_bound} does not necessarily hold in the general model.
In fact, as we mentioned in the introduction, even the English auction is not necessarily weakly shill-proof under an IPV type distribution if the distribution is not regular.
We leave as a future research question to find conditions on type distributions such that some strategy-proof and weakly shill-proof mechanism exists.

\section{Discussion} \label{sec:discussion}
\subsection{Shill-Proofness vs. Credibility} \label{ss:shill_v_cred}
As discussed in the review of literature, another notion of ``cheating'' by the auctioneer is that of \textit{in-credibility} introduced by \citet{akbarpourLi20}.
An auction is credible if a revenue-maximizing auctioneer has no incentive to lie about what other players are doing.
That information environment differs from this paper because we assume that bidders correctly (though perhaps not fully) perceive the actions of other players and where in the game tree they are.
In the Online Appendix, we formally define credibility in our setting (\cref{def:credible}) and prove the following implications:
\begin{proposition}\label{prop:cred_nesting}
	Suppose $(\game,\strategy)$ is an orderly auction.
	If the auction is \strongly~shill-proof, then it must be credible.
	If the auction is credible, then it must be \weakly~shill-proof.
\end{proposition}

We also define \textit{$\experiment$-credibility} for single-action auctions (\cref{def:gen_credible}) as a generalization of credibility allowing for bidders to have exogenous signals (where $\experiment$ is our notation from \cref{ss:single_action}) about the actions of other bidders as well as additional communication from the auctioneer.
Recall that we defined $\publicExperiment$ to mean that the signals reveal the actions of previous bidders and $\privateExperiment$ to be static auction setting.
We prove that the implications of \cref{prop:cred_nesting} still hold and give conditions under which credibility coincides with \strong~and \weak~shill-proofness:
\begin{proposition}\label{prop:sa_cred_nesting}
	Suppose $(\game,\strategy)$ is a single-action, orderly auction.
	If the auction is \strongly~shill-proof, then it must be $\experiment$-credible.
	If it is $(\privateExperiment)$-credible, then it is \strongly~shill-proof.
	If the auction is $\experiment$-credible, then it must be \weakly~shill-proof.
	If it is \weakly~shill-proof, then it is $(\publicExperiment)$-credible.
\end{proposition}
\cref{prop:sa_cred_nesting} implies that \cref{thm:wsp_trilemma} is a generalization of the credibility trilemma \cite[Theorem 1]{akbarpourLi20}.

\subsection{Dominant Strategies} \label{ss:xp_v_ds}

This paper focuses on ex-post strategies.
However, all our results can be extended to dominant strategies as well.
To extend \cref{thm:dutch_ssp} from an ex-post strategy not to shill to a dominant strategy is straight-forward: shill bidding in the Dutch auction always leads to $0$ revenue, which means it is a weakly dominated strategy, regardless of what other bidders do.
Further, since there exist no other auctions besides the Dutch auction that have an ex-post strategy not to shill bid, there can exist no other auction with a dominant strategy not to shill bid.

To extend \cref{prop:wsp_sp_upper_bound} from an ex-post equilibrium to a dominant strategy equilibrium for real bidders, some care must be taken in considering the information sets of different bidders when they take actions.
However, we can provide dominant-strategy equilibria versions of the English and ascending, screening auctions by assuming bidders move simultaneously each round of the English auction, as well in the second-price auction phase of the ascending, screening auction.
\cref{thm:wsp_trilemma} holds if we were to instead consider dominant strategies for real bidders as dominant strategy incentive compatibility is a stronger condition than ex-post incentive compatibility.

\bibliographystyle{ecta-fullname}
\bibliography{bib}

\begin{thebibliography}{31}
\newcommand{\enquote}[1]{``#1''}
\expandafter\ifx\csname natexlab\endcsname\relax\def\natexlab#1{#1}\fi

\bibitem[\protect\citeauthoryear{Akbarpour and Li}{Akbarpour and
  Li}{2020}]{akbarpourLi20}
\textsc{Akbarpour, Mohammad and Shengwu Li} (2020): \enquote{Credible auctions:
  A trilemma,} \emph{Econometrica}, 88 (2), 425--467.

\bibitem[\protect\citeauthoryear{Ashlagi and Gonczarowski}{Ashlagi and
  Gonczarowski}{2018}]{ashlagiGonczarowski18}
\textsc{Ashlagi, Itai and Yannai~A. Gonczarowski} (2018): \enquote{Stable
  matching mechanisms are not obviously strategy-proof,} \emph{Journal of
  Economic Theory}, 177, 405--425.

\bibitem[\protect\citeauthoryear{Ausubel and Milgrom}{Ausubel and
  Milgrom}{2006}]{ausubelMilgrom06}
\textsc{Ausubel, Lawrence~M. and Paul Milgrom} (2006): \emph{The lovely but
  lonely {Vickrey} auction}, 22--26.

\bibitem[\protect\citeauthoryear{Bahrani, Garimidi, and Roughgarden}{Bahrani
  et~al.}{2024}]{bahraniEtAl24}
\textsc{Bahrani, Maryam, Pranav Garimidi, and Tim Roughgarden} (2024):
  \enquote{Centralization in Block Building and Proposer-Builder Separation,}
  ArXiv:2401.12120.

\bibitem[\protect\citeauthoryear{Basu, Easley, {O'Hara}, and Sirer}{Basu
  et~al.}{2023}]{beos}
\textsc{Basu, Soumya, David Easley, Maureen {O'Hara}, and Emin~G{\"u}n Sirer}
  (2023): \enquote{{StableFees}: A Predictable Fee Market for
  Cryptocurrencies,} \emph{Management Science}, 69 (11), 6508--6524.

\bibitem[\protect\citeauthoryear{Chakraborty and Kosmopoulou}{Chakraborty and
  Kosmopoulou}{2004}]{chakrabortyKosmopoulou04}
\textsc{Chakraborty, Indranil and Georgia Kosmopoulou} (2004):
  \enquote{Auctions with shill bidding,} \emph{Economic Theory}, 24, 271--287.

\bibitem[\protect\citeauthoryear{Chen, Liang, Chang, Liu, Yin, and Yu}{Chen
  et~al.}{2020}]{chenEtAl20}
\textsc{Chen, Kong-Pin, Ting-Peng Liang, Ted Chang, Yi-chun Liu, Shou-Yung Yin,
  and Ya-Ting Yu} (2020): \enquote{How Serious is Shill Bidding in Online
  Auctions? {Evidence} from {eBay} Motors,} Working Paper.

\bibitem[\protect\citeauthoryear{Chitra, Ferreira, and Kulkarni}{Chitra
  et~al.}{2023}]{chitraEtAl23}
\textsc{Chitra, Tarun, Matheus V.~X. Ferreira, and Kshitij Kulkarni} (2023):
  \enquote{Credible, optimal auctions via blockchains,} ArXiv:2301.12532.

\bibitem[\protect\citeauthoryear{Chung and Shi}{Chung and Shi}{2023}]{CS23}
\textsc{Chung, Hao and Elaine Shi} (2023): \enquote{Foundations of Transaction
  Fee Mechanism Design,} in \emph{Proceedings of the 34th Annual ACM-SIAM
  Symposium on Discrete Algorithms (SODA)}, 3856--3899.

\bibitem[\protect\citeauthoryear{Conitzer, Immorlica, Letchford, Munagala, and
  Wagman}{Conitzer et~al.}{2010}]{conitzerEtAl10}
\textsc{Conitzer, Vincent, Nicole Immorlica, Joshua Letchford, Kamesh Munagala,
  and Liad Wagman} (2010): \enquote{False-name-proofness in social networks,}
  in \emph{Internet and Network Economics: 6th International Workshop, WINE
  2010, Stanford, CA, USA, December 13-17, 2010. Proceedings 6}, Springer,
  209--221.

\bibitem[\protect\citeauthoryear{Day and Milgrom}{Day and
  Milgrom}{2008}]{dayMilgrom08}
\textsc{Day, Robert and Paul Milgrom} (2008): \enquote{Core-selecting package
  auctions,} \emph{International Journal of Game Theory}, 36, 393--407.

\bibitem[\protect\citeauthoryear{Dhangwatnotai, Roughgarden, and
  Yan}{Dhangwatnotai et~al.}{2015}]{DRYgeb}
\textsc{Dhangwatnotai, Peerapong, Tim Roughgarden, and Qiqi Yan} (2015):
  \enquote{Revenue Maximization with a Single Sample,} \emph{Games and Economic
  Behavior}, 91, 318--333.

\bibitem[\protect\citeauthoryear{Elkind}{Elkind}{2007}]{elkind07}
\textsc{Elkind, Edith} (2007): \enquote{Designing and learning optimal finite
  support auctions,} in \emph{Proceedings of the Eighteenth Annual ACM-SIAM
  Symposium on Discrete Algorithms (SODA)}, SIAM, 736--745.

\bibitem[\protect\citeauthoryear{Essaidi, Ferreira, and Weinberg}{Essaidi
  et~al.}{2022}]{essaidiEtAl22}
\textsc{Essaidi, Meryem, Matheus V.~X. Ferreira, and S.~Matthew Weinberg}
  (2022): \enquote{Credible, strategyproof, optimal, and bounded expected-round
  single-item auctions for all distributions,} ArXiv:2205.14758.

\bibitem[\protect\citeauthoryear{Ferreira and Weinberg}{Ferreira and
  Weinberg}{2020}]{ferreiraWeinberg20}
\textsc{Ferreira, Matheus V.~X. and S.~Matthew Weinberg} (2020):
  \enquote{Credible, truthful, and two-round (optimal) auctions via
  cryptographic commitments,} in \emph{Proceedings of the 21st ACM Conference
  on Economics and Computation}, 683--712.

\bibitem[\protect\citeauthoryear{Hafalir, Kesten, Sherstyuk, and Tao}{Hafalir
  et~al.}{2023}]{hafalirEtAl23}
\textsc{Hafalir, Isa, Onur Kesten, Katerina Sherstyuk, and Cong Tao} (2023):
  \enquote{When Speed is of Essence: Perishable Goods Auctions,} Working Paper.

\bibitem[\protect\citeauthoryear{Haupt and Hitzig}{Haupt and
  Hitzig}{2021}]{hauptHitzig22}
\textsc{Haupt, Andreas and Zo{\"e} Hitzig} (2021): \enquote{Contextually
  Private Mechanisms,} ArXiv:2112.10812.

\bibitem[\protect\citeauthoryear{Izmalkov}{Izmalkov}{2004}]{I04}
\textsc{Izmalkov, Sergei} (2004): \enquote{Shill Bidding and Optimal Auctions,}
  MIT Working Paper.

\bibitem[\protect\citeauthoryear{Karlin and Rinott}{Karlin and
  Rinott}{1980}]{karlinRinott80}
\textsc{Karlin, Samuel and Yosef Rinott} (1980): \enquote{Classes of orderings
  of measures and related correlation inequalities. I. Multivariate totally
  positive distributions,} \emph{Journal of Multivariate Analysis}, 10 (4),
  467--498.

\bibitem[\protect\citeauthoryear{Lamy}{Lamy}{2009}]{lamy09}
\textsc{Lamy, Laurent} (2009): \enquote{The shill bidding effect versus the
  linkage principle,} \emph{Journal of Economic Theory}, 144 (1), 390--413.

\bibitem[\protect\citeauthoryear{Lavi, Sattath, and Zohar}{Lavi
  et~al.}{2022}]{LSZ19}
\textsc{Lavi, Ron, Or~Sattath, and Aviv Zohar} (2022): \enquote{Redesigning
  {B}itcoin's fee market,} \emph{ACM Transactions on Economics and
  Computation}, 10 (1), Art.~5.

\bibitem[\protect\citeauthoryear{Li}{Li}{2017}]{li17}
\textsc{Li, Shengwu} (2017): \enquote{Obviously strategy-proof mechanisms,}
  \emph{American Economic Review}, 107 (11), 3257--3287.

\bibitem[\protect\citeauthoryear{Mackenzie}{Mackenzie}{2020}]{mackenzie20}
\textsc{Mackenzie, Andrew} (2020): \enquote{A revelation principle for
  obviously strategy-proof implementation,} \emph{Games and Economic Behavior},
  124, 512--533.

\bibitem[\protect\citeauthoryear{Mackenzie and Zhou}{Mackenzie and
  Zhou}{2022}]{mackenzieZhou22}
\textsc{Mackenzie, Andrew and Yu~Zhou} (2022): \enquote{Menu mechanisms,}
  \emph{Journal of Economic Theory}, 204, 105511.

\bibitem[\protect\citeauthoryear{Milgrom and Segal}{Milgrom and
  Segal}{2017}]{milgromSegal17}
\textsc{Milgrom, Paul and Ilya Segal} (2017): \emph{Designing the {US}
  incentive auction}, Cambridge University Press, 803--812.

\bibitem[\protect\citeauthoryear{Milgrom and Weber}{Milgrom and
  Weber}{1982}]{MW82}
\textsc{Milgrom, Paul~R. and Robert~J. Weber} (1982): \enquote{A Theory of
  Auctions and Competitive Bidding,} \emph{Econometrica}, 50 (5), 1089--1122.

\bibitem[\protect\citeauthoryear{Porter and Shoham}{Porter and
  Shoham}{2005}]{PS05}
\textsc{Porter, Ryan and Yoav Shoham} (2005): \enquote{On Cheating in
  Sealed-Bid Auctions,} \emph{Decision Support Systems}, 39 (1), 41--54.

\bibitem[\protect\citeauthoryear{Pycia and Troyan}{Pycia and
  Troyan}{2023}]{pyciaTroyan23}
\textsc{Pycia, Marek and Peter Troyan} (2023): \enquote{A theory of simplicity
  in games and mechanism design,} \emph{Econometrica}, 91 (4), 1495--1526.

\bibitem[\protect\citeauthoryear{Roughgarden}{Roughgarden}{2021}]{roughgarden21}
\textsc{Roughgarden, Tim} (2021): \enquote{Transaction fee mechanism design,}
  \emph{ACM SIGecom Exchanges}, 19 (1), 52--55.

\bibitem[\protect\citeauthoryear{Shinozaki}{Shinozaki}{2024}]{shinozaki24}
\textsc{Shinozaki, Hiroki} (2024): \enquote{Shutting-out-proofness in object
  allocation problems with money,} \emph{Working Paper}.

\bibitem[\protect\citeauthoryear{Zeng}{Zeng}{2024}]{zeng24}
\textsc{Zeng, Haoyuan} (2024): \enquote{Identity-Proof Auctions,} Working
  Paper.

\end{thebibliography}

\iftoggle{compressLines}{}{\newpage}

\appendix
\section{Model (\cref{sec:model}) Appendix}
\begin{definition} \label{def:game_eq}
	Consider any set of real bidders $\realBidders$ and tuple $(\game,\strategy)$.
	We restrict the set of potential deviations for shill bidders to
	\[ \Sigma_{\shillBidders} = \cbr{ \strategy_{\shillBidders}' : \forall\type_{-\shillBidders}, \exists \type_{\shillBidders} \st \p{\strategy_{\shillBidders}',\strategy_{-\shillBidders}(\type_{-\shillBidders})} = \strategy(\type_{\shillBidders},\type_{-\shillBidders}; \realBidders=\potentialBidders) }. \]
	Then, the tuple $(\game,\strategy)$ is an \textbf{auction equilibrium} if for all $i \in \realBidders$ and deviating strategies $\strategy_i'$,
	\[ \expect_{\type_{-i}',\tilde{\realBidders}}\sbr{u_i\p{\strategy\p{\type_i,\type_{-i}';\tilde{\realBidders}}}} \ge \expect_{\type_{-i}',\tilde{\realBidders}}\sbr{u_i\p{\strategy_i'\p{\type_i,\type_{-i}';\tilde{\realBidders}},\strategy_{-i}\p{\type_{i},\type_{-i}';\tilde{\realBidders}}}}, \]
	and for all $i \in \shillBidders$ and $\strategy_{\shillBidders}' \in \Sigma_{\shillBidders}$,
	\[ \expect_{\type_{-\shillBidders}'}\sbr{u_i\p{\strategy\p{0,\type_{-\shillBidders}';\realBidders}}} \ge \expect_{\type_{-\shillBidders}'}\sbr{u_i\p{\strategy_{\shillBidders}'\p{0,\type_{-\shillBidders}';\realBidders},\strategy_{-\shillBidders}\p{0,\type_{-\shillBidders}';\realBidders}}}. \]
\end{definition}
In \cref{def:game_eq}, the strategy space $\Sigma_{\shillBidders}$ is restricting shill bidders to only take actions that could be played on-path by a real bidder.
In other words, shill bidders must act ``as-if'' they are real bidders and not take any actions that would definitively prove that they are shill bidders.
This restriction allows us to move to a direct mechanism where shill-proofness is defined as it being an equilibrium (ex-interim for \weak~shill-proofness, ex-post for \strong~shill-proofness) for all shill bidders to report $0$.
Note that if we were to enlarge the strategy space $\Sigma_{\shillBidders}$ to be the set $\hat{\Sigma}_{\shillBidders}$ of all strategy profiles, are main results would not change.
For \cref{thm:dutch_ssp}, \cref{prop:dutch_robust_wsp}, and \cref{prop:wsp_sp_upper_bound}, we are focused on augmented direct games and in such games $\Sigma_{\shillBidders} = \hat{\Sigma}_{\shillBidders}$.
For \cref{thm:wsp_trilemma}, we know that $\Sigma_{\shillBidders} \subset \hat{\Sigma}_{\shillBidders}$ and our impossibility result must still hold if the set of possible deviations by shill bidders is larger; thus, the theorem still holds.

Before we state our revelation principle in this context, we recall (with slight modification of notation) a definition and result from \cite{akbarpourLi20} that will be helpful in the proof.
\begin{definition}[\citet{akbarpourLi20}, Definition 2]\label{def:pruned}
	A game equilibrium $\p{\game,\strategy}$ is \textbf{pruned} if, for any history $h$:
	\begin{enumerate}[(i)]
		\item\label{cond:pruned_1} There exists $\type$ such that $\history \preceq \terminalState(\strategy(\type;\potentialBidders))$.
		\item\label{cond:pruned_2} If $\history \notin \possibleTerminalHistories$, then $\card{\succz(\history)} \ge 2$.
		\item\label{cond:pruned_3} If $\history \notin \possibleTerminalHistories$, then for $i = \playerFunction(\history)$, there exists $\type_{i}$,$\type_i'$, and $\type_{-i}$ such that
		\begin{enumerate}[(a)]
			\item $\history \prec \terminalState(\strategy(\type;\potentialBidders))$,
			\item $\history \prec \terminalState(\strategy(\type_{i}',\type_{-i};\potentialBidders))$, and
			\item $(\actionQuantity,\actionTransfer) \p{\strategy(\type;\potentialBidders)} \ne (\actionQuantity,\actionTransfer) \p{\strategy(\type_{i}',\type_{-i};\potentialBidders)}$.
		\end{enumerate}
	\end{enumerate}
\end{definition}

\begin{lemma}[\citet{akbarpourLi20}, Proposition 1] \label{lem:pruned}
	If $\p{\game,\strategy}$ is a game equilibrium, then there exists a game equilibrium $\p{\game',\strategy'}$ that is pruned and for all $\type$, $(\actionQuantity,\actionTransfer) \p{\strategyParameterized} = (\actionQuantity',\actionTransfer') \p{\strategy'(\type;\potentialBidders)}$.
\end{lemma}

\begin{lemma}[Augmented Revelation Principle] \label{lem:aug_rev_prin}
	For every game equilibrium $\p{\game,\strategy}$ there exists an auction $\directAuction$ that meets the following conditions:
	\begin{enumerate}[(i)]
		\item\label{cond:rev_prin_1} There exists a direct mechanism $(\quantity,\transfer)$: for all $\type$,
		\iftoggle{compressLines}{$\quantity(\type) = \actionQuantity(\strategy(\type;\potentialBidders))$ and $\transfer(\type) = \actionTransfer(\strategy(\type;\potentialBidders)).$}{	\[ \quantity(\type) = \actionQuantity(\strategy(\type;\potentialBidders)) \text{ and } \transfer(\type) = \actionTransfer(\strategy(\type;\potentialBidders)). \]}
		
		\item\label{cond:rev_prin_2} There exists a choice menu rule $\menuRule$ that is a function of the potential values $\possibleValues = \possibleValues_1 \times \cdots \times \possibleValues_{\numBidders}$ and bidder $\player$. This rule has an output of $L \ge 2$ choices characterized as $\cbr{\p{\choice_{\ell}, \nextPlayer_{\ell}}}_{\ell \in \set{L}}$ where: 
		\begin{enumerate}[(a)]
			\item\label{cond:rev_prin_2_a} $\cbr{\choice_{\ell}}_{\ell \in \set{L}}$ forms a partition of~$\possibleValues_{\player}$, $\nextPlayer_\ell \in \p{\potentialBidders \cup \cbr{\emptyset}} \backslash \cbr{\player}$, and $\nextPlayer = \emptyset$ signifies the game has ended.
			
			\item\label{cond:rev_prin_2_b} For any $\ell$ such that $\nextPlayer_\ell \ne \emptyset$, let $\hat{\possibleValues}^{\ell} = \p{\possibleValues_1,\ldots,\possibleValues_{\player-1},\choice_\ell,\possibleValues_{\player+1},\ldots,\possibleValues_{\numBidders}}$. Then, for any such $\ell$, there exists $\type_{\player_\ell},\type_{\player_\ell}' \in \hat{\possibleValues}^{\ell}_{\player_\ell}$, and $\type_{-\player_\ell} \in \hat{\possibleValues}^{\ell}_{-\player_\ell}$ such that 
			\iftoggle{compressLines}{$(\quantity,\transfer)\p{\type_{\player_\ell},\type_{-\player_\ell}} \ne (\quantity,\transfer) \p{\type_{\player_\ell}',\type_{-\player_\ell}}.$}{\[(\quantity,\transfer)\p{\type_{\player_\ell},\type_{-\player_\ell}} \ne (\quantity,\transfer) \p{\type_{\player_\ell}',\type_{-\player_\ell}}.\]}
			If $\type_{\player} \in \choice_{\ell}$, then the next player in the game is $\nextPlayer_{\ell}$ and the menu presented to her is $\menuRule(\hat{\possibleValues}^{\ell},\nextPlayer_{\ell})$.
			
			\item\label{cond:rev_prin_2_c} If $\ell$ is such that $\nextPlayer_\ell = \emptyset$, then for all $\type,\type' \in \hat{\possibleValues}_{\ell}$,
			\iftoggle{compressLines}{$\p{\quantity,\transfer} \p{\type} = \p{\quantity,\transfer} \p{\type'}.$}{\[ \p{\quantity,\transfer} \p{\type} = \p{\quantity,\transfer} \p{\type'}. \]}
			
			\item\label{cond:rev_prin_2_d} The first player to take an action is $\player_0$, who is presented the menu $\menuRule(\typeSpace^{\numBidders},\player_0)$.
		\end{enumerate}
	\end{enumerate}
\end{lemma}
\begin{proof}
	To prove \cref{cond:rev_prin_1}, we simply construct $(\quantity,\transfer)$ by iterating over all possible $\type$ and defining $(\quantity,\transfer)$ as the outcome of $\strategyParameterized$ in $\game$.
	
	To prove \cref{cond:rev_prin_2}, we first observe that by \cref{def:game_eq}, shill bidders must act ``as-if'' they were real bidders and that we have restricted to pure strategies.
	Thus, we can always label actions as classes $\p{\choice_{\ell}, \nextPlayer_{\ell}}$ of a partition of the remaining possible values for the current player $\player$ and satisfy \cref{cond:rev_prin_2_a}.
	The fact that $L \ge 2$ is equivalent to \cref{cond:pruned_1,cond:pruned_2} of \cref{def:pruned}, and \cref{cond:pruned_3} of \cref{def:pruned} is equivalent to \cref{cond:rev_prin_2_b} here.
	We can then apply \cref{lem:pruned} to find a game that satisfies these properties.
	\cref{cond:rev_prin_2_c} follows from the fact that $\game$ is well-defined (with each terminal history associated with a single outcome).
	\cref{cond:rev_prin_2_d} is simply mapping the first player in $\game$ to $\player_0$ and the auctioneer has no information on bidders' values yet.
\end{proof}

\section{\Strongly~Shill-Proof Auctions (\cref{sec:ssp}) Appendix}
\begin{lemma} \label{lem:winner_paying}
	An optimal auction $(\game,\strategy)$ is \textbf{winner-paying}: For all $i$ and $\val$,
	\[ \actionQuantity_i\p{\strategyParameterized} = 0 \implies \actionTransfer_i\p{\strategyParameterized} = 0. \]
\end{lemma}
\begin{proof}
	By the ex-post IR constraint, when
	$\actionQuantity_i\p{\strategyParameterized} = 0$, we have $\actionTransfer_i\p{\strategyParameterized} \le 0.$
	It then follows from optimality that $\actionTransfer_i\p{\strategyParameterized} = 0$.
	To see this, note that for bidder $j \ne i$, equilibrium constraints on bidder $j$ slacken when moving from $\actionTransfer_i < 0$ to $\actionTransfer_i = 0$ and so her play will remain the same.
	Meanwhile the transfer from bidder $i$ strictly increases moving from $\actionTransfer_i < 0$ to $\actionTransfer_i = 0$.
\end{proof}

\subsubsection*{Proof of \cref{lem:ssp_pab}.}
Towards contradiction, suppose there exists a \strongly~shill-proof auction $\directAuction$, player $\player$, and values $\type_{\player},\type_{-\player},\type_{-\player}'$ such that $\quantity_{\player}(\type) = \quantity_{\player}(\type_{\player},\type_{-\player}')$, but $\transfer_{\player}(\type) \ne \transfer_{\player}(\type_{\player},\type_{-\player}')$.
WLOG, suppose $\transfer_{\player}(\type) > \transfer_{\player}(\type_{\player},\type_{-\player}')$.
Now by \cref{lem:winner_paying}, $\player$ can only have two different transfers if that player wins the item under the allocation. 
Then, take $\realBidders = \cbr{\player}$ and by monotonicity, $\transfer_{\player}(\type) > \transfer_{\player}(\type_{\player},\type_{-\player}') \ge \transfer(\type_{\player},0)$ and thus shilling increases revenue and the auction is not \strongly~shill-proof.
\qed

Given our auction $\directAuction$, we define $\exintquantity_i(\type_i;\possibleValues)$ and $\exinttransfer_i(\type_i;\possibleValues)$ to be the ex-interim quantity and transfer rules, respectively, when bidder $i$ has value $\type_i$ and the set of potential values for all bidders is $\possibleValues$. 
\begin{lemma} \label{lem:ex_int_funcs}
	For every optimal auction $\optDirectAuction$, the ex-interim transfer rule for bidder $i$ is
	\[\exinttransfer_i(\type_i;\possibleValues) = \exintquantity_i(\type_i;\possibleValues) \type_i - \sum_{m : \type^{j_m} < \type_i}  \sbr{\exintquantity_i(\type^{j_m};\possibleValues) \cdot (\type^{j_{m+1}} - \type^{j_m})}, \]
	where  $\cbr{\type^{j_m}}_{m}$ are the ordered atoms of $\possibleValues_i$.
\end{lemma}
\begin{proof}
	To prove that $\exinttransfer$ has the claimed form, we will consider a specific non-truthful reporting: if a bidder has value $\type^m$, she commits to mis-reporting (selecting partitions) $\type^{m'}$ for the rest of the game.
	We now follow the proof of Theorem 1 of \cite{elkind07}.
	Since our direct mechanism is an equilibrium for real bidders, we must have that
	\begin{gather*}
		\exintquantity_i(\type^{j_m};\possibleValues)\type^{j_m} - \exinttransfer_i(\type^{j_m};\possibleValues) \ge \exintquantity_i(\type^{j_{m-1}};\possibleValues)\type^{j_m} - \exinttransfer_i(\type^{j_{m-1}};\possibleValues), \text{ and} \\
		\exintquantity_i(\type^{j_{m-1}};\possibleValues)\type^{j_{m-1}} - \exinttransfer_i(\type^{j_{m-1}};\possibleValues) \ge \exintquantity_i(\type^{j_m};\possibleValues)\type^{j_{m-1}} - \exinttransfer_i(\type^{j_m};\possibleValues).
	\end{gather*}
	Defining $U_i$ to be the ex-interim utility for bidder $i$, the preceding expressions become:
	\begin{gather*}
		U_i(\type^{j_m};\possibleValues) \ge U_i(\type^{j_{m-1}};\possibleValues) + (\type^{j_m} - \type^{j_{m-1}})\exintquantity_i(\type^{j_{m-1}};\possibleValues), \text{ and} \\
		U_i(\type^{j_{m-1}};\possibleValues) \ge U_i(\type^{j_{m}};\possibleValues) - (\type^{j_m} - \type^{j_{m-1}})\exintquantity_i(\type^{j_m};\possibleValues).
	\end{gather*}
	Thus,
	\iftoggle{compressLines}{$(\type^{j_m} - \type^{j_{m-1}})\exintquantity_i(\type^{j_{m-1}};\possibleValues) \le U_i(\type^{j_m};\possibleValues) - U_i(\type^{j_{m-1}};\possibleValues) \le (\type^{j_m} - \type^{j_{m-1}})\exintquantity_i(\type^{j_m};\possibleValues).$}{\[ (\type^{j_m} - \type^{j_{m-1}})\exintquantity_i(\type^{j_{m-1}};\possibleValues) \le U_i(\type^{j_m};\possibleValues) - U_i(\type^{j_{m-1}};\possibleValues) \le (\type^{j_m} - \type^{j_{m-1}})\exintquantity_i(\type^{j_m};\possibleValues) .\]}
	Hence, any IC mechanism is such that
	\begin{align*}
		U_i(\type^{j_m};\possibleValues) &= U_i(\type^{j_1}; \possibleValues) + \sum_{k=2}^{m} (\type^{j_m} - \type^{j_{m-1}})\tilde{\exintquantity}_i(\type^{j_m};\possibleValues) \\
		&\text{where } \tilde{\exintquantity}_i(\type^{j_m};\possibleValues) \in \sbr{\exintquantity_i(\type^{j_{m-1}};\possibleValues),\exintquantity_i(\type^{j_m};\possibleValues)}.
	\end{align*}
	Therefore, we have that
	\begin{equation} \label{eq:ic_transfer}
		\exinttransfer_i(\type^{j_m};\possibleValues) = \exintquantity_i(\type^{j_m};\possibleValues)\type^{j_m} - U_i(\type^{j_1}; \possibleValues) - \sum_{k=2}^{m} (\type^{j_m} - \type^{j_{m-1}})\tilde{\exintquantity}_i(\type^{j_m};\possibleValues).
	\end{equation}
	By the ex-post IR condition, we have $U_i(\type^{j_1}; \possibleValues) \ge 0$ for all $\possibleValues$.
	So, solving for the optimal transfer rule from \cref{eq:ic_transfer},
	\begin{align*}
		\exinttransfer_i^*(\type^{j_m};\possibleValues) &= \max_{U_i,\tilde{\exintquantity}} \sbr{\exintquantity_i(\type^{j_m};\possibleValues)\type^{j_m} - U_i(\type^{j_1}; \possibleValues) - \sum_{k=2}^{m} (\type^{j_m} - \type^{j_{m-1}})\tilde{\exintquantity}_i(\type^{j_m};\possibleValues)}  \\
		& \st U_i(\type^{j_1}; \possibleValues) \ge 0 \text{ and } \tilde{\exintquantity}_i(\type^{j_m};\possibleValues) \in \sbr{\exintquantity_i(\type^{j_{m-1}};\possibleValues),\exintquantity_i(\type^{j_m};\possibleValues)}.
	\end{align*}
	The solution to this maximization is
	\iftoggle{compressLines}{$U_i(\type^{j_1}; \possibleValues) = 0, \tilde{\exintquantity}_i(\type^{j_m};\possibleValues) = \exintquantity_i(\type^{j_{m-1}};\possibleValues).$}{\[ U_i(\type^{j_1}; \possibleValues) = 0, \tilde{\exintquantity}_i(\type^{j_m};\possibleValues) = \exintquantity_i(\type^{j_{m-1}};\possibleValues). \]}
	Thus, \cref{eq:ic_transfer} becomes
	\begin{align*}
		\exinttransfer_i(\type_i;\possibleValues) &= \exintquantity_i(\type_i;\possibleValues) \type_i - \sum_{m : \type^{j_m} < \type_i}  \sbr{\exintquantity_i(\type^{j_m};\possibleValues) \cdot (\type^{j_{m+1}} - \type^{j_m})}. &&\qedhere
	\end{align*}
\end{proof}

For any value choice $(\choice,\cdot) \in \menuRule(\cdot,\cdot)$, let us define $\lowchoice = \min_{w \in \choice} \cbr{w}$ and $\highchoice =\max_{w \in \choice} \cbr{w}$, respectively.
\begin{lemma}[Extended Pay-as-Bid] \label{lem:ssp_epab}
	Consider a \strongly~shill-proof and optimal auction $\optDirectAuction$. Fix $\possibleValues,\player$ and consider any $(\choice,\nextPlayer) \in \menuRule(\possibleValues,\player)$. If there exists $\type,\type' \in \possibleValues$ such that $\type_{\player},\type_{\player}' \in \choice$ and $\menuRule(\possibleValues,\player)$ is the last action $\player$ takes, then,
	\[ \quantity^*_{\player}(\type) = \quantity^*_{\player}(\type') = 1 \implies \transfer_{\player}(\type) = \transfer_{\player}(\type') = \frac{\exinttransfer_{\player}(\lowchoice;\possibleValues)}{\exintquantity_{\player}(\lowchoice;\possibleValues)},  \]
	i.e., transfers are constant conditional on allocation and are pinned down by the ex-interim outcome functions from the lowest type in the partition.
\end{lemma}
\begin{proof}
	Since an auction cannot distinguish between values in the same choice set, we can apply \cref{lem:ssp_pab} to conclude that if $\menuRule(\possibleValues,\player)$ is the last action $\player$ takes, then $\transfer(\type)=\transfer(\type')$.
	To conclude the proof, we note that $\player$ wins no matter what her value is in $\choice$ and then apply \cref{lem:winner_paying} to conclude that $\exinttransfer_{\player}(\lowchoice;\possibleValues) = \transfer_{\player}(\type) \cdot \exintquantity_{\player}(\lowchoice;\possibleValues)$.
\end{proof}

\subsubsection*{Proof of \cref{prop:fpa_ssp}}
Under a static information structure, the single-action, first-price auction directly corresponds to the direct mechanism.
In a pay-as-bid mechanism, a shill bidder's actions have no way to influence the payments except by changing the allocation and by the definition of orderly, if a shill bidder changes the allocation she wins the auction, which is dominated.
To prove that any more informative structure is not strongly shill-proof, consider some $i,j$ and $\type$ where $\informationSet(h) \ne \informationSet(h_0)$.
By definition, information sets form a partition and so $\informationSet(h) \cap \informationSet(h_0) = \emptyset$.
Observe that $j$ can have any type at both $h$ and $h_0$ given that each bidder takes a single action.
Now, consider the minimal $m$ such that $j$ believes ex-interim that she could win with type $\type^m$.
By assumption, $\type^m > \optimalReserve$.
Because we are considering strong shill-proofness, we can assume without loss that $j$ does in fact win the auction at $\type^m$.
So, at $h$, we can then apply \cref{lem:ssp_epab} to complete the proof that $j$ must pay more at $h$ than at $h_0$ and so if $i \in \shillBidders$, shill bidding is profitable at this type realization.
\qed

\subsubsection*{Proof of \cref{thm:dutch_ssp}.}
	We first show that the Dutch auction with reserve price $\reserve$ is a well-defined, i.e., that the stopping rule allows for the auction to be orderly.\footnote{The auction is public by definition.}
	We then show that it is \strongly~shill-proof.
	Finally, we show that there are no other public, \strongly~shill-proof, orderly, optimal auctions.
	
	\textbf{The Dutch Auction is Orderly, Optimal, and \Strongly~Shill-Proof. }
	The Dutch auction quantity rule is orderly and the transfer rule is ex-post IR and monotone.
	Indeed, by construction, the next player $\nextPlayer$ is always the player with the potentially highest value (including for tie-breaking). 
	So, if that player indicates that she is of the highest possible type, the outcome (allocation and transfer) is fully determined and the auction ends.
	The auction ends once there are no players who could have values weakly greater than~$\reserve$.
	
	We now prove that the Dutch auction is \strongly~shill-proof.
	Towards contradiction, suppose not.
	So, there must exist some $\shillBidders$, $\player \in \shillBidders$, and $\possibleValues$ such that $\cbr{\highValue_{\player}}$ is selected from the menu $\menuRule(\possibleValues,\player)$.
	But by construction, this means that the auction immediately ends and the good is allocated to the shill bidder who misreported.
	By \cref{lem:winner_paying}, the revenue from this deviation is $0$, which must be weakly less than any other possible transfer. 
	
	\textbf{Uniqueness. }
	Towards contradiction, suppose there exists a menu rule $\tilde{\menuRule} \ne \menuRule^D_{\reserveIndex^*}$ that is associated with a public, \strongly~shill-proof, orderly, optimal auction.
	Therefore, there exists $\possibleValues$ and $\player$ such that $\tilde{\menuRule}(\possibleValues,\player) \ne \menuRule^D_{\reserveIndex^*}(\possibleValues,\player)$.
	Without loss, we will suppose that $\possibleValues$ is the first time in the game tree that $\tilde{\menuRule}$ differs from $\menuRule^D_{\reserveIndex^*}$.
	Formally, for all $\hat{\possibleValues} \supsetneq \possibleValues$, $\tilde{\menuRule}\p{\hat{\possibleValues},\player} = \menuRule^D\p{\hat{\possibleValues},\player}$.
	We now proceed in cases.
	
	\textit{Case 1 (Different Next Player Choice).} Suppose
	\iftoggle{compressLines}{$\tilde{\menuRule}(\possibleValues,\player) = \cbr{\p{\choice_{\mathsf{L}},\tilde{\nextPlayer}_{\mathsf{L}}},\p{\choice_{\mathsf{H}},\tilde{\nextPlayer}_{\mathsf{H}}}},$}{\[ \tilde{\menuRule}(\possibleValues,\player) = \cbr{\p{\choice_{\mathsf{L}},\tilde{\nextPlayer}_{\mathsf{L}}},\p{\choice_{\mathsf{H}},\tilde{\nextPlayer}_{\mathsf{H}}}},\]}
	where $\tilde{\nextPlayer}_{\mathsf{L}} \ne \nextPlayer_{\mathsf{L}}$ or $\tilde{\nextPlayer}_{\mathsf{H}} \ne \nextPlayer_{\mathsf{H}}$.
	If $\tilde{\nextPlayer}_{\mathsf{L}} = \nextPlayer_{\mathsf{L}}$, then $\tilde{\nextPlayer}_{\mathsf{H}} = \nextPlayer_{\mathsf{H}}$ because the outcome is fully resolved once a bidder selects the high partition.
	So, we need only consider the case where $\tilde{\nextPlayer}_{\mathsf{L}} \ne \nextPlayer_{\mathsf{L}}$.
	By \cref{def:orderly}, $\tilde{\nextPlayer}_{\mathsf{H}} \ne \emptyset$ (even if the $\player$ chooses $\cbr{\highValue_{\player}}$).
	Now, there must exist some bidders $b_1$ and $b_2$ where $b_1$ is called before $b_2$ but $(\cdot,b_2) \ordering (\cdot,b_1)$ as otherwise the auction calls players in the same order as the Dutch auction, which we assumed was not the case.
	Let $\realBidders = \potentialBidders\backslash\cbr{b_1}$, and for some $m$, $\type_{b_2} = \type^m$ and $\type_i = \type^{m-2}$ for all bidders $i \notin \cbr{b_1,b_2}$.
	Taking the expression from \cref{lem:ex_int_funcs} and dividing both sides by $\exintquantity_{b_2}$, we get
	\[ \frac{\exinttransfer_{b_2}(\type^m;\possibleValues)}{\exintquantity_{b_2}(\type^m;\possibleValues)} = \type^m - \sum_{k: \type^{j_k} < \type^m} \frac{\exintquantity_{b_2}\p{\type^{j_{k}};\possibleValues}}{\exintquantity_{b_2}\p{\type^m;\possibleValues}}\cdot(\type^{j_{k+1}}-\type^{j_{k}}) < \type^m. \]
	(Note that there must be at least one such $k$ in the summation because otherwise $\possibleValues_{b_2} = \cbr{\type^m}$ and $b_2$ would not take an action.)
	
	The last choice $b_2$ makes is to select a partition $\choice$ such that $\lowchoice \ge \type^{m-2}$.
	We can therefore apply \cref{lem:ssp_epab} to conclude that the transfer if $b_1$ reports $0$ must be $\frac{\exinttransfer_{b_2}(\lowchoice;\possibleValues)}{\exintquantity_{b_2}(\lowchoice;\possibleValues)} \le \frac{\exinttransfer_{b_2}(\type^m;\possibleValues)}{\exintquantity_{b_2}(\type^m;\possibleValues)} < \type^m$.
	If bidder $b_1$ instead reports $\type^m$, then bidder $b_2$ will win and the revenue will be $\type^m$ and so the auction will not be \strongly~shill-proof.
	
	\textit{Case 2 (Different Partitions).} Suppose there exists $(\choice,\cdot) \in \tilde{\menuRule}(\possibleValues,\player)$ such that $\choice \notin \cbr{\choice_{\mathsf{L}}, \choice_{\mathsf{H}}}$.
	We need not consider the case where $\lowchoice = 0$ because in that case either we can consider some other choice $\choice' \notin \cbr{\choice_{\mathsf{L}}, \choice_{\mathsf{H}}}$ or $\choice = \possibleValues_{\player}$ which would violate \cref{lem:aug_rev_prin}.
	We need not consider $\lowchoice > \highValue_{-\player}$ because $\possibleValues$ is the first time $\tilde{\menuRule}$ differs from $\menuRule^D_{\reserveIndex^*}$ and for all $\possibleValues$, $\player$, and $(\tilde{\choice},\cdot) \in \menuRule^D_{\reserveIndex^*}(\possibleValues,\player)$, it is the case that $\tilde{\lowchoice} \le \highValue_{-\player}$.
	So, there are now only three sub-cases we must consider: $\lowchoice \in [\optimalReserve,\highValue_{-\player})$, $\lowchoice = \highValue_{-\player}$, and $\lowchoice \in (0,\optimalReserve)$.
	
	\textit{Case 2a ($\lowchoice \in [\optimalReserve,\highValue_{-\player})$).} In this sub-case, there exists $m^*$ such that $\optimalReserve \le \lowchoice_{\ell} \le \type^{m^*} < \highValue_{-\player}$.
	Since $\possibleValues$ is the first time that $\tilde{\menuRule}$ differs from $\menuRule^D_{\reserveIndex^*}$, we can suppose there exists $i$ such that $(\type^{m^*},\player) \ordering (\optimalReserve,i)$ because otherwise the outcome would already be resolved or the player rotation would be the only difference (Case 1).
	Then, suppose bidder $i$ is such that $i \in \realBidders$ and $\type_i \ge \type^{m^*+1}$.
	Take bidder $\player \in \shillBidders$ to shill $\type^{m^*}$; and for $k \notin \cbr{i,\player}$, take $\type_k = \lowValue_k < \type^{m^*}$.
	Therefore, by \cref{lem:ex_int_funcs}, observe that for the last action $i$ takes, her ex-interim transfer must be higher when shill $\player$ reports $\type^{m^*}$ than when she reports $0$.
	Thus by \cref{lem:ssp_epab}, when $\type^{m^*} < \highValue_{-\player}$, there exists a valuation vector $\type$ such that a shill bidder would want to deviate away from reporting $0$---and therefore such an auction is not \strongly~shill-proof.
	
	\textit{Case 2b ($\lowchoice = \highValue_{-\player}$).} In this sub-case, we know $\optimalReserve \le \lowchoice_{\ell} = \highValue_{-\player}$.
	Since $\possibleValues$ has been generated via a Dutch auction so far, the current player has the lowest tie-breaking priority, i.e., $\player$ is such that for all $j \ne \player$, $(\highValue_j,j) \ordering (\highValue_j,\player)$.
	Letting $j \in \realBidders$ and $\type_j = \highValue_{-\player}$, take bidder $\player \in \shillBidders$ to report $\highValue_{-\player}$; and for all $k \notin \cbr{j,\player}$, take $\type_k = \lowValue_k < \highValue_{-\player}$.
	As noted, $(\highValue_j,j) \ordering (\highValue_j,\player)$ and so bidder $j$ is allocated the item and not shill bidder $\player$.
	Therefore, by the same argument as Case 2a, shill bidding will increase revenue.
	
	\textit{Case 2c ($\lowchoice \in (0,\optimalReserve)$).} By \cref{cond:rev_prin_2_b} of \cref{lem:aug_rev_prin}, there must be some chance that $\player$ could win the auction in order to affect outcomes.
	In particular, by definition of $\quantity^*$, there exists $j$ and $\type_j \in \possibleValues_j$ such that $(\type_j, j) \ordering (\optimalReserve, \player)$ and by \cref{lem:ssp_pab}, there exists $\type_j' \in \possibleValues_j$ such that $(\highchoice, \player) \ordering (\type_j',j)$.
	If any bidder is ever offered a choice with $\lowchoice \ge \optimalReserve$, then the previous cases imply that a shill bidder can profitably deviate and so we only have to consider the instances where no such choice is offered.
	Now, suppose $\realBidders = \cbr{j}$.
	Suppose all shill bidders play the strategy of selecting the partition $\tilde{\choice}$ such that $0 < \tilde{\lowchoice} < \optimalReserve$ if such a choice is available.
	Let the final move that $j$ takes to be $\choice^{0,j,\mathrm{last}}$ and $\choice^{S,j,\mathrm{last}}$ under the shill bidders' strategy of selecting $0$ and not, respectively.
	Similarly, define $\possibleValues^{0,\mathrm{last}},\possibleValues^{S,\mathrm{last}}$ as the possible values and $\transfer^{0}_j,\transfer^S_j$ as the transfers under these respective strategies.
	Observe that it is without loss to assume that $\type^{\omega} \equiv \lowchoice^{0,j,\mathrm{last}} \le \lowchoice^{S,j,last}$ because in the latter case the shill bidders are always acting as if they have higher values than in the former case.
	Next, for $c \in \cbr{0,S}$, let $p_i^{m,c} = \prob\sbr{(\type^m,j) \ordering \p{\type_i,i} \mid \possibleValues_{i}^{c,\mathrm{last}}}$ and define $\zeta_{i,m} = \frac{p_i^{m,0}p_i^{\omega,S}}{p_i^{m,S}p_i^{\omega,0}}$.
	Observe that for all $i \ne j$ and $m \le \omega$, it is the case that $\zeta_{i,m} \ge 1$ with at least one strict inequality because $\highValue_i^{S,\mathrm{last}} \ge \highValue_i^{0,\mathrm{last}}$ for all $i$ with strict inequality for at least one $i$ and $m$.
	So, by \cref{lem:ex_int_funcs,lem:ssp_epab},\footnote{We may assume that the possible values for $j$ are sequential above the reserve otherwise we could consider $j$ as a shill bidder for some other bidder instead by the cases above.}
	
	\begin{align*}
		\transfer^S_j - \transfer^{0}_j &\ge \frac{\exinttransfer_{\player}(\type^{\omega};\possibleValues^{S,\mathrm{last}})}{\exintquantity_{\player}(\type^{\omega};\possibleValues^{S,\mathrm{last}})} - \frac{\exinttransfer_{\player}(\type^{\omega};\possibleValues^{0,\mathrm{last}})}{\exintquantity_{\player}(\type^{\omega};\possibleValues^{0,\mathrm{last}})} \\
		&= \sum_{\reserveIndex^* \le k < \omega} \p{(\type^{k+1}-\type^{k})\prod_{i \ne j} \frac{p_i^{k,0}}{p_i^{\omega,0}}} - \sum_{\reserveIndex^* \le k < \omega} \p{(\type^{k+1}-\type^{k})\prod_{i \ne j} \frac{p_i^{k,S}}{p_i^{\omega,S}}} \\
		&= \frac{1}{\prod_{i \ne j} p_i^{\omega,0} p_i^{\omega,S}} \cdot \Bigg(\sum_{\reserveIndex^* \le k < \omega} (\type^{k+1}-\type^{k}) \cdot \p{\prod_{i \ne j} p_i^{k,S}p_i^{\omega,0}(\zeta_{i,k} - 1) } \Bigg) > 0.
	\end{align*}
	Thus, we have described a profitable shill bidding strategy in this sub-case.
	\qed

\section{Weakly Shill-Proof Auctions (\cref{sec:wsp}) Appendix}
\subsection{Single-Action Auctions (\cref{ss:single_action}) Appendix}
\begin{lemma} \label{lem:one_shot_rev_principle}
	For any single-action auction, there exist unique $\quantity: \typeSpace^{\numBidders} \to \{0,1\}^{\numBidders}$ and $\transfer: \typeSpace^{\numBidders} \to \R^{\numBidders}$ such that:
	\begin{enumerate}[(i)]
		\item\label{cond:one_shot_corr} (Correspondence) For all $\val \in \typeSpace^{\numBidders}$,
		\iftoggle{compressLines}{$\quantity^*(\type) = \actionQuantity(\strategyParameterized)$ and $\transfer(\type) = \actionTransfer\p{\strategyParameterized}$.}{\[\quantity^*(\type) = \actionQuantity(\strategyParameterized) \text{ and } \transfer(\type) = \actionTransfer\p{\strategyParameterized}.\]}
		\item\label{cond:one_shot_ir} (Individual Rationality) For all $i \in \potentialBidders$ and $\type$,
		\iftoggle{compressLines}{$\quantity^*_i(\type)\type_i - \transfer_i(\type) \ge 0.$}{\[\quantity^*_i(\type)\type_i - \transfer_i(\type) \ge 0.\]}
		\item (Incentive Compatibility) For all $\realBidders$, $i \in \realBidders$, $\type_i$, and $\type_i'$,
		\begin{multline} \label{eq:one_shot_ic}
			\expect_{\type_{-i},\tilde{\realBidders}} \sbr{\quantity^*_i\p{\strategy(\type;\tilde{\realBidders})}\type_i - \transfer_i(\strategy(\type;\tilde{\realBidders})) \mid \type_i,\signal_i = \experiment(\type_{j<i}) } \\\ge 
			\expect_{\type_{-i},\tilde{\realBidders}} \sbr{\quantity^*_i\p{\strategy(\type_i',\type_{-i}; \tilde{\realBidders})}\type_i - \transfer_i\p{\strategy(\type_i',\type_{-i}; \tilde{\realBidders})} \mid \type_i,\signal_i = \experiment(\type_{j<i})}.
		\end{multline}
	\end{enumerate}
\end{lemma}
\begin{proof}
	To begin, let us note that we can uniquely define $(\quantity,\transfer)$ point-wise based on the outcomes in $\game$ from playing $\strategyParameterized$.
	Next, the IR constraint (\cref{cond:one_shot_ir}) follows immediately from the ex-post IR condition and our construction of $(\quantity,\transfer)$.
	Finally, \cref{eq:one_shot_ic} comes from \cref{def:game_eq} and recalling that we restrict shill bidders to actions that could have been taken by real bidders.
\end{proof}

\begin{lemma} \label{lem:wsp_val_equiv_general}
	If a single-action auction is \weakly~shill-proof, then for all $\realBidders$, $\type_{j < \min \shillBidders}$,\footnote{$\type_{j < \min \shillBidders} \equiv \cbr{\type_j : j < \min_{i \in \shillBidders}\cbr{i}}$.} and $\cbr{\type_i}_{i \in \shillBidders}$,
	\[\expect_{\type} \sbr{\sum_{k \in \realBidders} \transfer_k\p{\cbr{\type_i}_{i \in \shillBidders},\cbr{\type_i}_{i \notin \shillBidders}} \mid \type_{j < \min \shillBidders}} \le \expect_{\type} \sbr{\sum_{k \in \realBidders} \transfer_k\p{0,\cbr{\type_i}_{i \notin \shillBidders}} \mid \type_{j < \min \shillBidders}}. \]
\end{lemma}
\begin{proof}
	Towards contradiction, suppose there exists $\realBidders$, $\type_{j < \min \shillBidders}$, and $\cbr{\type_i}_{i \in \shillBidders}$ such that
	\[\expect_{\type} \sbr{\sum_{k \in \realBidders} \transfer_k\p{\cbr{\type_i}_{i \in \shillBidders},\cbr{\type_i}_{i \notin \shillBidders}} \mid \type_{j < \min \shillBidders}} > \expect_{\type} \sbr{\sum_{k \in \realBidders} \transfer_k\p{0,\cbr{\type_i}_{i \notin \shillBidders}} \mid \type_{j < \min \shillBidders}}. \]
	We now prove that the deviation by the coalition $\shillBidders$ where they report $\cbr{\type_i}_{i \in \shillBidders}$ is profitable and therefore that the auction is not \weakly~shill-proof.
	By assumption, a shill bidder observes actions by all bidders who take actions before her.
	So, $\cbr{\type_i}_{i \in \shillBidders}$ can condition on $ \type_{j < \min \shillBidders}$ when making decisions.
	Then, the strategy by $\shillBidders$ of committing to report $\cbr{\type_i}_{i \in \shillBidders}$ regardless of what other bidders play after $\min_{i \in \shillBidders} \cbr{i}$ must be strictly profitable compared to always reporting $0$.
	Thus, we have found a strategy that does strictly better than always reporting $0$: 
	When the values before $\min_{i \in \shillBidders} \cbr{i}$ are reported as $\type_{j < \min \shillBidders}$, report $ \cbr{\type_i}_{i \in \shillBidders}$. 
	Otherwise, report $0$.
	This strategy in the direct game immediately translates to a profitable deviation in the auction by \cref{def:game_eq} and \cref{lem:one_shot_rev_principle} and thus the equilibrium is not \weakly~shill-proof.
\end{proof}

Now, when discussing single-action auctions, we focus on the direct mechanisms associated to \weakly~shill-proof auctions and so we will refer to an auction as $(\quantity,\transfer,\experiment)$ without reference to $\realBidders$.

\begin{lemma} \label{lem:wsp_form}
	Suppose a single-action, orderly auction $(\quantity,\transfer,\experiment)$ is \weakly~shill-proof. 
	Then, for all $i$, $\type$, and $\type_{j > i}'$,
	\iftoggle{compressLines}{$\sbr{\quantity_i \p{\type}=\quantity_i\p{\type_{j \le i}, \type_{j > i}'} \implies \transfer_i\p{\type}=\transfer_i\p{\type_{j \le i}, \type_{j > i}'}}.$}{\[\quantity^*_i \p{\type}=\quantity^*_i\p{\type_{j \le i}, \type_{j > i}'} \implies \transfer_i\p{\type}=\transfer_i\p{\type_{j \le i}, \type_{j > i}'}.\]}
\end{lemma}
\begin{proof}
	Towards contradiction, suppose there exists $i$, $\type$, and $\type_{j > i}'$, such that
	\iftoggle{compressLines}{$\quantity_i \p{\type}=\quantity_i\p{\type_{j \le i}, \type_{j > i}'}$, but $\transfer_i\p{\type; \signal} > \transfer_i\p{\type_{j \le i}, \type_{j > i}'; \signal}$.}{\[ \quantity_i \p{\type}=\quantity_i\p{\type_{j \le i}, \type_{j > i}'}, \text{but } \transfer_i\p{\type; \signal} > \transfer_i\p{\type_{j \le i}, \type_{j > i}'; \signal}. \]}
	Because the auction is orderly, we can apply \cref{lem:winner_paying} to conclude that $\quantity_i \p{\type}=\quantity_i\p{\type_{j \le i}, \type_{j > i}'}=1$.
	Let $\realBidders = \set{i}$. 
	Then,
	\begin{align*}
		\expect_{\type} \sbr{\sum_{k \in \realBidders} \transfer_k(\cbr{\type_i}_{i \in \shillBidders},\cbr{\type_i}_{i \notin \shillBidders}) \mid \type_{j \le \min \shillBidders}} &= \transfer_i\p{\type} > \transfer_i\p{\type_{j \le i}, \type_{j > i}'} \ge \transfer_i\p{\type_{j \le i}, 0}.
	\end{align*}
	This violates \cref{lem:wsp_val_equiv_general}, and so we have reached a contradiction.
\end{proof}

\begin{lemma} \label{lem:msp_form}
	Suppose a single-action, orderly auction $(\quantity,\transfer,\experiment)$ is mildly \strategyproof~and \weakly~shill-proof.
	Then, there exists $i < \numBidders$ such that for all $\type_i$, $\type_i'$, and $\type_{-i} \in \experiment^{-1}_i(\type_{j < i})$,
	\iftoggle{compressLines}{$\sbr{\quantity_i(\type) = \quantity_i(\type_i',\type_{-i}) \implies \transfer_i(\type) = \transfer_i(\type_i',\type_{-i})}.$}{\[ \quantity_i(\type) = \quantity_i(\type_i',\type_{-i}) \implies \transfer_i(\type) = \transfer_i(\type_i',\type_{-i}). \]}
\end{lemma}
\begin{proof}
	Consider $i < \numBidders, \realBidders \ni i, \type_i, \type_i'$, and $\type_{-i} \in \experiment^{-1}_i(\type_{j < i})$ such that $\quantity_i(\type) = \quantity_i(\type_i',\type_{-i})$.
	WLOG, suppose $\type_i > \type_i'$.
	By monotonicity, $\transfer_i(\type) \ge \transfer_i(\type_i',\type_{-i})$.
	Towards contradiction, suppose $\transfer_i(\type) > \transfer_i(\type_i',\type_{-i})$.
	By the winner-paying property, $\transfer_i(\type) > \transfer_i(\type_i',\type_{-i})$ implies that $\quantity_i(\type) = \quantity_i(\type_i',\type_{-i}) = 1$.
	However, note that $\transfer_i(\type) > \transfer_i(\type_i',\type_{-i})$ would mean that the utility of reporting $\type_i'$ would be higher than truthful reporting under true value $\type_i$ which would violate the mildly \strategyproofness~and thus $\transfer_i(\type_i',\type_{-i}) = \transfer_i(\type_i',\type_{-i})$.
\end{proof}

\subsubsection*{Proof of \cref{thm:wsp_trilemma}.}
Towards contradiction, suppose such an auction did exist.
Fix $i < \numBidders$ and $\type_{j < i}$ and then suppose that $\type_i < \type^{\numPossVals}$.
Combining \cref{lem:wsp_form,lem:msp_form}, we can see that for all $\type_i'$ and $\type_{-i},\type_{-i}' \in \experiment^{-1}_{i}(\type_{j<i})$, if $\quantity_i(\type) = \quantity_i(\type')$, then $\transfer_i(\type) = \transfer_i(\type')$.
So, define $\transfer_i$ as the (constant) $\transfer_i(\type)$ for all $\type$ such that $\quantity_i(\type) = 1$.

In order for bidder $i$ to have an ex-post strategy, when $\quantity(\type) = 1$, it must also be the case that $\quantity_i(\type^{\numPossVals},\type_{-i}) = 1$.
So, applying the winner-paying property (and suppressing that the expectation is conditioned on $\signal_i = \experiment_i(\type_{j < i})$), we have
\begin{multline} \label{eq:high_report_util}
	\expect_{\type_{-i}}\sbr{\quantity_i\p{\type^{\numPossVals},\type_{-i}}\type_i - \transfer_i\p{\type^{\numPossVals},\type_{-i}} }\\ = 
	\expect_{\type_{-i}}\sbr{\type_i - \transfer_i^* \mid \quantity_i(\type) = 1} + \expect_{\type_{-i}}\sbr{\type_i - \transfer_i^* \mid \quantity_i(\type) = 0, \quantity_i\p{\type^{\numPossVals},\type_{-i}} = 1},
\end{multline}
\begin{equation} \label{eq:true_report_util}
	\text{and }\expect_{\type_{-i}}\sbr{\quantity_i\p{\type}\type_i - \transfer_i\p{\type}} = \expect_{\type_{-i}}\sbr{\type_i - \transfer_i^* \mid \quantity_i(\type) = 1}.
\end{equation}
Taking the difference between \cref{eq:high_report_util} and \cref{eq:true_report_util}, we see that
\begin{multline} \label{eq:expected_difference}
	\expect_{\type_{-i}}\sbr{\quantity_i\p{\type^{\numPossVals},\type_{-i}}\type_i - \transfer_i\p{\type^{\numPossVals},\type_{-i}}} - \expect_{\type_{-i}}\sbr{\quantity_i\p{\type}\type_i - \transfer_i\p{\type}}  \\=
	\expect_{\type_{-i}}\sbr{\type_i - \transfer_i^* \mid \quantity_i(\type) = 0, \quantity_i\p{\type^{\numPossVals},\type_{-i}} = 1}
\end{multline}
Now, by definition of orderly, $\quantity$ is monotone, and by assumption $\transfer$ is monotone.
If there exists $\type^m$ such that $\prob[\quantity_i(\type^m,\type_{-i}) = 1] < \prob[\quantity_i(\type^{\numPossVals},\type_{-i}) = 1]$, then for such $m$,
\begin{gather*}
	\expect_{\type_{-i}}\sbr{\quantity_i(\type^m,\type_{-i}) \type^m - \transfer_i(\type^m,\type_{-i})} \ge \expect_{\type_{-i}}\sbr{\quantity_i(\type^{m-1},\type_{-i}) \type^m - \transfer_i(\type^{m-1},\type_{-i})} \\
	> \expect_{\type_{-i}}\sbr{\quantity_i(\type^{m-1},\type_{-i}) \type^{m-1} - \transfer_i(\type^{m-1},\type_{-i})} \ge 0,
\end{gather*}
where the last inequality comes from the IR condition.
Thus, the IR constraint does not bind for $\type_i = \type^m$. 
Since the good has to be allocated to the highest type, for all $i < \numBidders$, there exists $\type_{-i}$ such that $\quantity_i(\type) = 0$ and $\quantity_i(\type^{\numPossVals},\type_{-i}) = 1$.
Thus, 
\begin{equation} \label{eq:misreport_profitable}
	\expect_{\type_{-i}}\sbr{\type_i - \transfer_i^* \mid \quantity_i(\type) = 0, \quantity_i\p{\type^{\numPossVals},\type_{-i}} = 1} > 0.
\end{equation}
Combining \cref{eq:expected_difference,eq:misreport_profitable}, we see that
\begin{equation}\label{eq:IC_violated}
	\expect_{\type_{-i}}\sbr{\quantity_i\p{\type^{\numPossVals},\type_{-i}}\type_i - \transfer_i\p{\type^{\numPossVals},\type_{-i}}} > \expect_{\type_{-i}}\sbr{\quantity_i\p{\type}\type_i - \transfer_i\p{\type}}.
\end{equation}
We then apply the \weak~shill-proofness condition to simplify \cref{eq:one_shot_ic} to
\[ \expect_{\type_{-i}} \sbr{\quantity_i\p{\type}\type_i - \transfer_i(\type) } \ge \expect_{\type_{-i}} \sbr{\quantity_i\p{\type_i',\type_{-i}}\type_i - \transfer_i\p{\type_i',\type_{-i}}}. \]
Taking $\type_i' = \type^{\numPossVals}$, \cref{eq:IC_violated} violates the IC constraint from \cref{lem:one_shot_rev_principle}---and thus we have reached a contradiction.
\qed

\subsection{Efficient Auctions (\cref{ss:wsp_dutch_auction}) Appendix}
In order to build towards a proof \cref{prop:dutch_robust_wsp}, we will prove that for a certain class of value distributions, every \textit{\weakly} shill-proof and efficient auction must be a semi-Dutch auction.
Formally, we assume that the value distribution is sparse:

\begin{definition} \label{def:sparse}
	A regular distribution $F$ is \textbf{sparse} if for all $k < \reserveIndex^*$,
	\begin{equation} \label{eq:sparsity}
		\type^k - (\type^{k+1}-\type^k)\frac{f(\type^{k+1})}{f(\type^k)} < 0.
	\end{equation}
\end{definition}
A distribution is sparse if the atoms are sufficiently far apart.
Sparsity can also be a reasonable assumption if the auctioneer has preferences for the auction to be completed quickly, or otherwise finds it costly to distinguish between values that are close to each other.

\begin{lemma} \label{lem:shill_regular_strategy}
	Consider an efficient auction and suppose $\typeDist$ is regular and sparse.
	Let $\realBidders$, $\possibleValues$ such that $\possibleValues_i = \cbr{w : w \in [\lowValue_i,\highValue_i]}$ for all $i \in \realBidders$, and consider $(\underline{\choice},j)$ such that $\underline{\choice} < \optimalReserve, j \notin \realBidders$ and for all $i \in \realBidders, (\highValue_i,i) \ordering (\underline{\choice},j)$. Then, for all $\gamma < \underline{\choice}$,
	\[ \expect\sbr{\sum_{i \in \realBidders} \transfer_i(\type) \mid \possibleValues=\p{\possibleValues_{-j},\cbr{\gamma}}} < \expect\sbr{\sum_{i \in \realBidders} \transfer_i(\type) \mid \possibleValues=\p{\possibleValues_{-j},\cbr{\underline{\choice}}}}. \]
	Thus, the following shilling strategy is profitable compared to always reporting $0$: if there exists $(\possibleValues,\player,\choice)$ such that $(\choice,\cdot) \in \menuRule(\possibleValues,\player)$ and $\lowchoice \in (0,\optimalReserve)$, then select $\choice$. 
	Otherwise, select the partition containing $0$.
\end{lemma}
\begin{proof}
	Consider any $i \in \realBidders$ and $\possibleValues$ and let $C = \p{\sum_{\type^k \in \possibleValues_i} \typepmf(\type^k)}^{-1}$.
	Then, applying \cref{eq:ic_transfer},

	\begin{align*}
		\expect&\sbr{\transfer_i(\type) \mid \possibleValues} + U_i(\type^{j_1}; \possibleValues) = \expect\sbr{\exinttransfer_i(\type_i;\possibleValues)} + U_i(\type^{j_1}; \possibleValues) \iftoggle{compressLines}{=}{\\&=} C \sum_{m} \typepmf(\type^{j_m}) \exinttransfer_i(\type^{j_m};\possibleValues) \\
		&= C \sum_{m: \type^{j_m} \in \possibleValues_i} \typepmf(\type^{j_m}) \p{\exintquantity_i(\type^{j_m};\possibleValues) \type^{j_m} - \sum_{k < m} \sbr{\tilde{\exintquantity}^k_i(\possibleValues) \cdot (\type^{j_{k+1}} - \type^{j_k})}} \\
		&= C \sbr{ \sum_{m: \type^{j_m} \in \possibleValues_i} \typepmf(\type^{j_m}) \exintquantity_i(\type^{j_m};\possibleValues) \type^{j_m} - \sum_{m: \type^{j_m} \in \possibleValues_i} \sum_{k < m} \typepmf(\type^{j_m})\sbr{\tilde{\exintquantity}^k_i(\possibleValues) \cdot (\type^{j_{k+1}} - \type^{j_k})} } \\
		&= C \sum_{m: \type^{j_m} \in \possibleValues_i} \sbr{ \type^{j_m} \exintquantity_i(\type^{j_m};\possibleValues) - (\type^{j_{m+1}}-\type^{j_m})\frac{\typeDist(\highValue_i)-\typeDist(\type^{j_m})}{\typepmf(\type^{j_m})}\tilde{\exintquantity}^m_i(\possibleValues)} \typepmf(\type^{j_m}).
	\end{align*}
	Applying the definition of the efficient allocation rule $\quantity^E$, we know that for $(\type^m,i) \ordering (\gamma,j)$ and $(\type^m,i) \ordering (\gamma',j)$, we can define
	\iftoggle{compressLines}{$\exintquantity_i(\type^m;\possibleValues_{-j}) \equiv \exintquantity_i(\type^m;\possibleValues_{-j},\cbr{\gamma}) = \exintquantity_i(\type^m;\possibleValues_{-j},\cbr{\gamma'}).$}{\[\exintquantity_i(\type^m;\possibleValues_{-j}) \equiv \exintquantity_i(\type^m;\possibleValues_{-j},\cbr{\gamma}) = \exintquantity_i(\type^m;\possibleValues_{-j},\cbr{\gamma'}).\]}
	Note that $\underline{\choice} \le \min_i \cbr{\highValue_i}$ by assumption and therefore, for $\underline{\choice} \in \p{\gamma, \optimalReserve}$,
	\begin{align*}
		\expect&\big[\transfer_i(\type) \mid \possibleValues=\p{\possibleValues_{-j},\cbr{\gamma}}\big] - \expect\sbr{\transfer_i(\type) \mid \possibleValues=\p{\possibleValues_{-j},\cbr{\underline{\choice}}}} \\
		&= C\sum_{m : \gamma \le \type^{j_m} < \underline{\choice}} \sbr{ \type^{j_m} \exintquantity_i(\type^{j_m};\possibleValues) - (\type^{j_{m+1}}-\type^{j_m})\frac{\typeDist(\highValue_i)-\typeDist(\type^{j_m})}{\typepmf(\type^{j_m})}\tilde{\exintquantity}^m_i(\possibleValues)} \typepmf(\type^{j_m}) \\
		&\le C\sum_{m : \gamma \le \type^{j_m} < \underline{\choice}} \sbr{ \type^{j_m} - (\type^{j_{m}+1}-\type^{j_m})\frac{\typepmf(\type^{j_{m}+1})}{\typepmf(\type^{j_m})}} \typepmf(\type^{j_m})\exintquantity_i(\type^{j_m};\possibleValues_{-j}) < 0 
	\end{align*}
	where the final inequality comes from sparsity. And so,
	\[ \expect\sbr{\sum_{i \in \realBidders} \transfer_i(\type) \mid \possibleValues=\p{\possibleValues_{-j},\cbr{\gamma}}} < \expect\sbr{\sum_{i \in \realBidders} \transfer_i(\type) \mid \possibleValues=\p{\possibleValues_{-j},\cbr{\underline{\choice}}}}, \] as claimed in the statement of the lemma.
	Thus, committing to misreport as $\underline{\choice}$ is strictly beneficial compared to any strategy that can only report $\gamma < \underline{\choice}$.
\end{proof}

\begin{lemma} \label{lem:dutch_sparse_wsp}
	If $\typeDist$ is regular and sparse, then every public, \weakly~shill-proof, and efficient auction is a semi-Dutch auction with cutoff $\optimalReserve$.
\end{lemma}
\begin{proof}
	Suppose $\typeDist$ is regular and sparse.
	Consider an arbitrary \weakly~shill-proof and efficient auction, $\effDirectAuction$, and consider any $\type$ such that $\max_i \cbr{\type_i} < \optimalReserve$.
	We prove that both parts of the definition of a semi-Dutch auction are necessary.
	
	\textit{Proof of \cref{cond:semi_dutch_cutoff} of \cref{def:semi_dutch}. }
	First we prove that if, for any player $\player$ and set of possible values $\possibleValues$, there exists $(\choice,\cdot) \in \menuRule(\possibleValues,\player)$ where $0 < \lowchoice < \optimalReserve$, then $\possibleValues \subseteq \widecheck{\possibleValues}$.
	Towards contradiction, suppose that there exists a $(\player,\possibleValues,\choice)$ such that $\possibleValues \not\subseteq \widecheck{\possibleValues}$, $\p{\choice,\cdot} \in \menuRule(\possibleValues,\player)$ $\lowchoice \in \p{0,\optimalReserve}$,
	and $\player \in \shillBidders$.
	
	Let us first prove it is without loss to assume $(\player,\possibleValues,\choice)$ is such that for all $i$, $\lowValue_i = 0$ or $\lowValue_i \ge \optimalReserve$. 
	If there exists $(\player,\possibleValues,\choice)$ and $i$ such that $\lowValue_{i} \in (0,\optimalReserve)$, let us label that set as $\possibleValues^{K}$ and let $\possibleValues^0 \supset \possibleValues^1 \supset \ldots \supset \possibleValues^{K}$ be the sequence of on-path possible value sets preceding $\possibleValues^K$.
	Let the players called along the path be $\player^0,\player^1,\ldots,\player^K$ and the value partition selected by player $k$ to be $\choice^k$.
	Note that $\possibleValues^0 = \typeSpace^{\numBidders}$ and so for all $i$, $\tilde{\lowValue}_i = 0$ or $\tilde{\lowValue}_i \ge \optimalReserve$.
	So, the set $\mathcal{K} = \cbr{k < K : \lowchoice^k \in (0,\optimalReserve)} \ne \emptyset$ and therefore $k^* = \min_{k \in \mathcal{K}}\cbr{k}$ is well-defined.
	If $k$ is such that $\lowchoice^k \notin (0,\optimalReserve)$ and for all $i$, $\lowValue^k_i \notin (0,\optimalReserve)$, then it must be the case that for all $i$, $\lowValue^{k+1}_i \notin (0,\optimalReserve)$.
	Since $k^*$ is the first time in the game that a player selects a partition with $\lowchoice^{k} \in (0,\optimalReserve)$, it must be the case that for all $i$, $\lowValue^k_i = 0$ or $\tilde{\lowValue}_i \ge \optimalReserve$.
	Since $\possibleValues^{k^*} \supset \possibleValues^K$, $\possibleValues^{k^*} \not\subseteq \widecheck{\possibleValues}$.
	Thus, $(\player^{k^*},\possibleValues^{k^*},\choice^k)$ is such that $\possibleValues^{k^*} \not\subseteq \widecheck{\possibleValues}$, $\p{\choice^{k^*},\cdot} \in \menuRule\p{\possibleValues^{k^*},\player^{k^*}}$, $\lowchoice^{k^*} \in \p{0,\optimalReserve}$, and for all $i$, $\lowValue^{k^*}_i = 0$ or $\lowValue^{k^*}_i \ge \optimalReserve$.
	
	So, in order to have $\lowchoice \in (0,\optimalReserve)$, it must be the case that $0 \in \possibleValues_{\player}$.
	Thus it is possible for $\player$ to be a shill bidder while having so far only selected partitions that contain $0$.
	Let $\shillBidders = \cbr{i : \highValue_i < \optimalReserve} \cup \cbr{\player}$.
	By assumption that $\possibleValues \nsubseteq \widecheck{\possibleValues}$, there must exist $i$ such that $\highValue_i \ge \optimalReserve$ and thus we can suppose that $\realBidders \ne \emptyset$.
	By assumption that $\lowchoice \in \p{0,\optimalReserve}$, we can suppose that $\realBidders$ is such that for all $i \in \realBidders$, $\highValue_i > \lowchoice$.
	By \cref{lem:shill_regular_strategy}, this would contradict the hypothesis that the auction is \weakly~shill-proof and so we must have $\possibleValues \subseteq \widecheck{\possibleValues}$ when there exists $(\choice,\cdot) \in \menuRule(\possibleValues,\player)$ such that $0 < \lowchoice < \optimalReserve$.
	
	\textit{Proof of \cref{cond:semi_dutch_menu} of \cref{def:semi_dutch}. }
	We now prove that for any player $\player$ and set of possible values $\possibleValues \subseteq \widecheck{\possibleValues}$, it is the case that $\menuRule(\possibleValues,\player) = \menuRule^D_{\reserveIndex^*}(\possibleValues,\player)$.
	Consider any option $\p{\choice,\cdot} \in \menuRule\p{\possibleValues,\player}$.
	Observe that by \cref{lem:shill_regular_strategy}, it is not the case that $0 < \lowchoice < \highValue_{-\player}$.
	So, $\lowchoice \ge \highValue_{-\player}$.
	Since this is the case for all $\possibleValues$, it must therefore be true that $\lowchoice = \highValue_{\player}$.
	This is because if $\highValue_{\player} > \lowValue \ge \highValue_{-\player}$, then there must have existed some earlier menu $\p{\tilde{\choice},\cdot} \in \menuRule\p{\tilde{\possibleValues},\tilde{\player}}$ for which $\underline{\tilde{\choice}} < \overline{\tilde{\possibleValues}}_{\tilde{\player}}$.
	
	So far we have proven that
	\iftoggle{compressLines}{$\menuRule(\possibleValues,\player) = \cbr{\p{\choice_{\mathsf{L}},\tilde{\nextPlayer}_{\mathsf{L}}},\p{\choice_{\mathsf{H}},\tilde{\nextPlayer}_{\mathsf{H}}}}.$}{\[ \menuRule(\possibleValues,\player) = \cbr{\p{\choice_{\mathsf{L}},\tilde{\nextPlayer}_{\mathsf{L}}},\p{\choice_{\mathsf{H}},\tilde{\nextPlayer}_{\mathsf{H}}}}. \]}
	To complete the proof, we have to prove that $\tilde{\nextPlayer}_{\mathsf{L}} = \nextPlayer_{\mathsf{L}}$ and $\tilde{\nextPlayer}_{\mathsf{H}} = \nextPlayer_{\mathsf{H}}$.
	Towards contradiction, suppose $\tilde{\nextPlayer}_{\mathsf{L}} \ne \nextPlayer_{\mathsf{L}}$ or $\tilde{\nextPlayer}_{\mathsf{H}} \ne \nextPlayer_{\mathsf{H}}$.
	If $\tilde{\nextPlayer}_{\mathsf{H}} \ne \nextPlayer_{\mathsf{H}}$, then, by \cref{lem:aug_rev_prin}, Condition (ii), there exists $i$ such that $\highValue_i = \highValue_{\player}, \p{\highValue_i,i} \ordering \p{\highValue_{\player},\player}$.
	We can let $\realBidders = \cbr{i}$ and then apply \cref{lem:shill_regular_strategy} to contradict the hypothesis that the auction is \weakly~shill-proof.
	If $\tilde{\nextPlayer}_{\mathsf{L}} \ne \nextPlayer_{\mathsf{L}}$, then, as argued in the proof of \cref{thm:dutch_ssp}, the menu presented to $\tilde{\nextPlayer}_{\mathsf{L}}$ must not have the auction end immediately, no matter what partition $\tilde{\nextPlayer}_{\mathsf{L}}$ selects.
	Thus, our previous argument for the case where $\tilde{\nextPlayer}_{\mathsf{H}} \ne \nextPlayer_{\mathsf{H}}$ applies, and we can conclude that $\menuRule(\possibleValues,\player) = \menuRule^D(\possibleValues,\player)$.
\end{proof}

\subsubsection*{Proof of \cref{prop:dutch_robust_wsp}}
	The statement follows as a corollary of \cref{lem:dutch_sparse_wsp}.
	Consider any optimal reserve $\optimalReserve$, $\underline{\numPossVals}$ atoms below the optimal reserve, and $\overline{\numPossVals}$ atoms above the optimal reserve.
	We construct a sparse (and regular) distribution $\tilde{\typeDist}$ with optimal reserve $\optimalReserve$,  $\underline{\numPossVals}$ atoms below the optimal reserve, and $\overline{\numPossVals}$ atoms above the optimal reserve.
	To begin, let $\evenSpacing$ such that $\underline{\numPossVals}\evenSpacing \le \optimalReserve$ and $(\underline{\numPossVals}+1)\evenSpacing > \optimalReserve$.
	Then for all $k \le \underline{\numPossVals}$, let $\type^{k+1}-\type^{k} = \evenSpacing$ and $\tilde{\typepmf}(\type^k) = e^{-\lambda(k-2)\evenSpacing} - e^{-\lambda(k-1)\evenSpacing}.$
	Note that
	\[ \tilde{\varphi}^k = \type^k - (\type^{k+1} - \type^k) \frac{1 - \tilde{\typeDist}(\type^k)}{\tilde{\typepmf}(\type^{k})} = (k-1)\evenSpacing - \evenSpacing \frac{e^{-\lambda(k-1)\evenSpacing}}{e^{-\lambda(k-2)\evenSpacing} - e^{-\lambda(k-1)\evenSpacing}}, \]
	and so
	\iftoggle{compressLines}{$\tilde{\varphi}^{k+1} - \tilde{\varphi}^k = k\evenSpacing - (k-1)\evenSpacing = \evenSpacing > 0;$}{\[ \tilde{\varphi}^{k+1} - \tilde{\varphi}^k = k\evenSpacing - (k-1)\evenSpacing = \evenSpacing > 0; \]}
	hence, $\tilde{\typeDist}$ satisfies the regularity condition for $k \le \underline{\numPossVals}$.
	
	In order for $\optimalReserve$ to be an optimal reserve of $\tilde{\typeDist}$, it must be the case that for $k^*$ such that $\tilde{\varphi}^{k^*} \ge 0$ and $\tilde{\varphi}^{k^*-1} < 0$, it is also the case that $\optimalReserve \in ((k^*-1)\evenSpacing,k^*\evenSpacing]$.
	Such a $k^*$ must be equal to $\lceil \check{k} \rceil$, where 
	\iftoggle{compressLines}{$\check{k} \evenSpacing - \frac{\evenSpacing}{e^{\lambda \evenSpacing} - 1} = 0.$}{\[ \check{k} \evenSpacing - \frac{\evenSpacing}{e^{\lambda \evenSpacing} - 1} = 0. \]}
	Thus, $\underline{\numPossVals}-1 = k^* = \lceil \frac{1}{e^{\lambda \evenSpacing} - 1} \rceil$.
	
	In order for $\tilde{\typeDist}$ to be sparse, it must satisfy \cref{eq:sparsity}, which here simplifies to
	\[ (k-1)\evenSpacing \cdot \p{1 + \frac{e^{-\lambda(k-2)\evenSpacing} - e^{-\lambda(k-1)\evenSpacing}}{e^{-\lambda(k-1)\evenSpacing} - e^{-\lambda(k)\evenSpacing}}} < k\evenSpacing \implies \frac{e^{2\lambda \evenSpacing} - 1}{e^{\lambda \evenSpacing} - 1} < \frac{k}{k - 1}. \]
	So, $\tilde{\typeDist}$ is sparse if
	\begin{equation} \label{eq:sparse_exp_dist}
		e^{2\lambda \evenSpacing} < \frac{2 + (e^{\lambda \evenSpacing}-1)(2(k^*-\check{k})-1)}{1 - ((k^*-\check{k})-1)(e^{\lambda \evenSpacing}-1)}.
	\end{equation}
	Selecting $\lambda$ such that $\underline{\numPossVals} - 1 = \check{k} = k^*$, \cref{eq:sparse_exp_dist} is satisfied.
	
	Finally, to finish constructing $\tilde{\typeDist}$, we simply select atoms $\type^{\underline{\numPossVals}+2},\ldots,\type^{\underline{\numPossVals}+\overline{\numPossVals}}$ and respective probability weights to satisfy
	\begin{gather*}
		\tilde{\varphi}^k \text{ is non-decreasing and } \sum_{k=\underline{\numPossVals}+2}^{\underline{\numPossVals}+\overline{\numPossVals}} \tilde{\typepmf}(\type^k) = e^{-\lambda(\underline{\numPossVals}+1)\evenSpacing}.
	\end{gather*}
	This system of constraints has at most $\overline{\numPossVals}$ constraints and $2(\overline{\numPossVals}-1)$ free variables, so the system can be satisfied. 
	Thus, we have constructed a regular and sparse $\tilde{\typeDist}$ that has the required values of $\optimalReserve$,$\underline{\numPossVals}$, and $\overline{\numPossVals}$. Then, we can apply \cref{lem:dutch_sparse_wsp} to conclude the proof.

\subsection{Strategy-Proof Auctions (\cref{ss:strategy_proof}) Appendix}
\begin{definition}\label{def:discrete}
	Let $\typeDist$ be a discrete distribution with ordered atoms $0 = \type^1<\ldots<\type^M$ and $\contDist$ be a continuous distribution with p.d.f.~$\contpdf$.
	If $Y_{\contDist} \sim \contDist$, then $\typeDist$ is a \textbf{discrete approximation} of $\contDist$ when $Y_{\typeDist} \sim \typeDist$ is defined as
	\begin{equation} \label{eq:discretization}
		Y_{\typeDist} = \begin{cases}
			\type^1 & Y_{\contDist} \le \type^1 \\
			\type^k & Y_{\contDist} \in (\type^{k-1},\type^{k}] \\
			\type^{\numPossVals} & Y_{\contDist} > \type^{\numPossVals-1}
		\end{cases}.
	\end{equation}
\end{definition}
For such a distribution $\typeDist$, let $\overline{\Delta} = \max_k \cbr{\type^k - \type^{k-1}}$.
As convention, let $\typeDist^{-1}$ be the left pseudo-inverse: $\typeDist^{-1}(x) = \max\cbr{\type^k : x \ge \typeDist(\type^k)}$.

\begin{definition} \label{def:mhr_dist}
	Let $\typeDist$ be a discrete approximation of $\contDist$. The distribution $\typeDist$ is a \textbf{monotone hazard rate (MHR) distribution} if $\frac{\typepmf(w^k)}{1-\typeDist(w^k)}$ is monotonically increasing in $k$ and $h(x) = \frac{\contpdf(x)}{1-\contDist(x)}$ is monotonically increasing in $x$.
\end{definition}

\begin{lemma} \label{lem:wsp_sp_bound}
	If the value distribution is a discrete MHR distribution $\typeDist$, then for all 
	\begin{equation} \label{eq:req_screen_level}
		\screenVal \ge \typeDist^{-1}\p{ \typeDist\p{\optimalReserve} + \max_{1 \le n < \numBidders}\cbr{\p{\max\cbr{1 - \frac{\optimalReserve}{\optimalReserve+2\overline{\Delta}}\p{\frac{\typepmf(\optimalReserve)}{1-\typeDist(\optimalReserve)}}^{n},0}}^{1/n} }},
	\end{equation}
	the ascending, screening auction with screening level $\screenVal$ is a \weakly~shill-proof, \strategyproof, and optimal auction.
\end{lemma}

\subsubsection*{Proof of \cref{lem:wsp_sp_bound}.}
	\textbf{The Ascending, Screening Auction is Orderly and Optimal. }
	We first prove that the auction is well-defined, orderly, and optimal.
	The transfer and allocation function are orderly and optimal, so we only have to show that the menu rule can induce this outcome function.
	Let us examine the English auction phase first.
	The auction ends if and only if $\highValue < \type^{\numPossVals}$.
	When that occurs, the auction has determined $\type_i$ for all $i$ given $\type_i \ge \optimalReserve$.
	Thus, the outcome is fully determined.
	In the second-price auction phase, each value weakly greater than $\screenVal$ is determined precisely (and there are at least two players with values weakly greater than $\screenVal$) and so the outcome rule is determined.

	\textbf{The Ascending, Screening Auction is Ex-Post Incentive Compatible. }
	Observe that the definition of \strategyproofness~(\cref{def:strategy_proof}) is a function solely of the direct mechanism $(\quantity,\transfer)$ and not of the menu rule $\menuRule$.
	Both the English auction phase and the second-price auction use the same transfer function $\transfer^2$.
	The entire auction has the optimal allocation rule $\quantity^*$ and so the ascending, screening auction is \strategyproof.
	
	\textbf{The Ascending, Screening Auction is \Weakly~Shill-Proof. }
	By assumption that the screen level is at least $\optimalReserve$, any potential shill bidder will be asked to play at least once in the English auction before being able to play in the second-price auction.
	If the optimal shill bid is $0$ in the first round of the English auction, then the auction is \weakly~shill-proof because once a bidder reports $0$, she ``drops out'' and does not take another action.
	
	Observe that given the form of the transfer rule, the maximum amount that a bidder $i$ with value $\type_i$ would have to pay is $\type_i$.
	Thus, the maximum possible gain in revenue from a shill bidder deviating is at most the difference between the first and second moment of $\typeDist$.
	Next, note that MHR distributions are regular.
	Regularity implies that if a shill bidder reports a non-zero value in the English auction stage and the auction concludes before reaching the second-price stage, the expected gain must be weakly less than $0$.
	So, when considering the expected gain of misreporting, we can think of the expected gain from manipulating outcomes in the English auction component as at most $0$ and can focus on manipulating outcomes in the second-price stage.
	Therefore, the total gains from misreporting as a shill bidder must be bounded above by the probability that a shill bidder is able to manipulate the outcome of the second-price auction multiplied by the expected difference between the first and second moments of the value distribution conditional on reaching the second-price auction stage.
	
	Let $\contDist$ be the continuous distribution for which $\typeDist$ is a discrete approximation. 
	For an exponential distribution with rate $\lambda$, the expected difference between the first and second moments of $T$ independent draws is $\frac{1}{\lambda}$.
	The exponential distribution, with its constant hazard rate, has the thickest right tail of any MHR distribution and so has the largest expected difference between its first and second moments (see proof of Theorem 5.1 in \citet{bahraniEtAl24}).
	In particular, since we are only interested in value draws above the reserve $\rho_{\contDist}^*$ and $\contDist$ has a non-decreasing hazard rate, we can take the rate $\lambda = h(\rho_{\contDist}^*) = \frac{1}{\rho_{\contDist}^*}$ and conclude that the maximum difference between the first and second moments of $\contDist$ must be bounded above by $\rho_{\contDist}^*$.
	Recall that $h(\rho_{\contDist}^*) = \frac{1}{\rho_{\contDist}^*}$ because $\contDist$ is regular and $\rho_{\contDist}^*$ is the optimal reserve of $\contDist$.
	Examining \cref{eq:discretization}, we can see that our discrete approximation pools draws from a continuous distribution upwards to atoms and so, if the absolute difference between two samples of the continuous distribution is $\kappa$, the absolute difference between the discrete approximation samples would be at most $\kappa + \overline{\Delta}$.
	Thus, the maximum possible expected difference between the first and second moments of $\typeDist$ conditional on being above the reserve is at most $\rho_{\contDist}^* + \overline{\Delta}$.
	Note that we also know that $\card{\optimalReserve-\rho_{\contDist}^*} \le \overline{\Delta}$.
	
	Suppose bidder $i \in \shillBidders$ and it is the first time she is taking an action.
	Then, under the rules of the auction, she has not indicated that her value is greater than $\optimalReserve$ yet.
	For any real bidder $j \ne i$, there are two cases: either bidder $j$ has indicated her value is weakly greater than $\optimalReserve$ $\p{\prob\sbr{\type_j < \screenVal} = \typeDist(\screenVal) - \typeDist(\optimalReserve)}$ or she has not yet taken an action $\p{\prob\sbr{\type_j < \screenVal} = \typeDist(\screenVal)}$.
	So, if $K \le \numBidders$ real bidders have not dropped out yet (i.e., indicated that their value is less than $\optimalReserve$), then the probability that the auction would continue to the second-price auction is at most $1-\p{\typeDist\p{\screenVal} - \typeDist\p{\optimalReserve}}^K$.
	Therefore, the maximum expected gain for a shill bidder from misreporting in her first action of the English auction phase when $K$ bidders have not dropped is at most
	\begin{equation} \label{eq:shill_upper_bound}
		\p{1-\p{\typeDist\p{\screenVal} - \typeDist\p{\optimalReserve}}^K}\p{\rho_{\contDist}^* + \overline{\Delta}}.
	\end{equation}
	
	We now turn to bounding the loss from reporting a non-zero value as a shill bidder.
	If shill bidder $i$ misreports her value as $\type^m$ at some point in the English auction phase and then she wins the item without taking another action, then the transfer the seller would have received had $i$ not misreported is at least $\max\cbr{\optimalReserve,\type^{m-1}} \ge \optimalReserve$, assuming at least one real bidder has value weakly above the reserve.
	To bound the probability that a real bidder $j$ would have won the item if not for shill bidder $i$'s misreport, we can consider the probability that bidder $j$ has indicated her value is at least $\lowValue_j \ge \optimalReserve$.
	By \cref{def:mhr_dist}, the hazard rate of $\contDist$ is non-decreasing.
	So,
	\[ \prob\sbr{\type_j \le \type^m} \ge \frac{\displaystyle  \sum_{\cbr{k:\lowValue_j \le \type^k < \type^m}}\typepmf(\type^k)}{1-\typeDist(\lowValue_j)} \ge \frac{\typepmf(\lowValue_j)}{1-\typeDist(\lowValue_j)} \ge \frac{\typepmf(\optimalReserve)}{1-\typeDist(\optimalReserve)}. \]
	Combining the preceding inequality with our hypothesis that $K$ bidders have not dropped out yet, the expected loss for a shill bidder of misreporting is at least
	\begin{equation} \label{eq:shill_lower_bound}
		\optimalReserve \cdot \p{\frac{\typepmf(\optimalReserve)}{1-\typeDist(\optimalReserve)}}^{K}.
	\end{equation}
	
	We conclude the proof by showing that $\screenVal$ satisfying \cref{eq:req_screen_level} implies that the expected revenue loss from misreporting as a shill is weakly larger than the expected gain.
	Beginning with \cref{eq:req_screen_level}, we can see that for all $K < N$,
	\begin{align*}
		\screenVal &\ge \typeDist^{-1}\p{ \typeDist\p{\optimalReserve} + \max_{1 \le n < \numBidders}\cbr{\p{\max\cbr{1 - \frac{ \optimalReserve}{\optimalReserve+2\overline{\Delta}}\p{\frac{\typepmf(\optimalReserve)}{1-\typeDist(\optimalReserve)}}^{n},0}}^{1/n} }} \\
		&\ge \typeDist^{-1}\p{ \typeDist\p{\optimalReserve} + \p{\max\cbr{1 - \frac{ \optimalReserve}{\optimalReserve+2\overline{\Delta}}\p{\frac{\typepmf(\optimalReserve)}{1-\typeDist(\optimalReserve)}}^{K},0}}^{1/K} } \\
		&\ge \typeDist^{-1}\p{ \typeDist\p{\optimalReserve} + \p{\max\cbr{1 - \frac{ \optimalReserve}{\rho_{\contDist}^*+\overline{\Delta}}\p{\frac{\typepmf(\optimalReserve)}{1-\typeDist(\optimalReserve)}}^{K},0}}^{1/K} }.
	\end{align*}
	This implies that
	\[ \optimalReserve \cdot \p{\frac{\typepmf(\optimalReserve)}{1-\typeDist(\optimalReserve)}}^{K} \ge \p{1-(\typeDist(\screenVal) - \typeDist(\optimalReserve))^K}\p{\rho_{\contDist}^* + \overline{\Delta}}. \]
	The left-hand side of the preceding equation corresponds to \cref{eq:shill_lower_bound}, the lower bound on the expected loss from misreporting as a shill bidder, and the right-hand side corresponds to \cref{eq:shill_upper_bound}, the upper bound on the expected gain from misreporting and thus we have shown that it is equilibrium not to shill when $\screenVal$ is sufficiently high. 
\qed

\subsubsection*{Proof of \cref{prop:wsp_sp_upper_bound}.}
Let $\typeDist_m$ be the discrete approximation of the exponential distribution with rate $\lambda=1$ and atoms at $\cbr{0,2,\ldots,2m}$.
Then, using the argument from the proof of \cref{prop:dutch_robust_wsp}, $\typeDist_m$ is regular and MHR.
For all $m>2$, an optimal reserve is $\optimalReserve = 4$. 
Observe that $\screenValSuper^*=3$ satisfies \cref{eq:req_screen_level} because
\[\frac{\optimalReserve}{\optimalReserve+2\overline{\Delta}}\p{\frac{\typepmf(\optimalReserve)}{1-\typeDist(\optimalReserve)}}^{n} = \frac{1}{2}(e^2-1)^n > 1 \text{ for all } n \ge 1.  \]
We apply \cref{lem:wsp_sp_bound} to conclude that the ascending, screening auction with screen level $\type^{\screenValSuper^*}$ is \weakly~shill-proof, \strategyproof, and optimal for all $\typeDist_m$.
Then, $\worstCaseQueries^{AS,\screenValSuper^*}(\typeDist_m) = 2$ and $\worstCaseQueries^E(\typeDist_m) = m-2$.
Thus, $\worstCaseQueries^{AS}(\typeDist_m,\screenValSuper^*)/\worstCaseQueries^E(\typeDist_m) \to 0$ as $m \to \infty$, concluding the proof.
\qed

\section{Affiliation and Interdependence (\cref{sec:aff_and_sp}) Appendix}
\begin{lemma}[\cite{karlinRinott80}]\label{lem:aff_exp_order}
	For any non-decreasing function $g$, if $\typeDist' \affOrder \typeDist$, then $\expect_{\type \sim \typeDist'}\sbr{g(\type)} \ge \expect_{\type \sim \typeDist}\sbr{g(\type)}$.
\end{lemma}

\begin{lemma}\label{lem:aff_dec_diff_order}
	For any function $g :\R^{\numBidders} \to \R$ with decreasing differences, if $\typeDist' \affOrder \typeDist$, then for any index $i$ and for all $\type^1_i > \type^2_i > \type^3_i$,
	\begin{align*}
		\expect_{\type \sim \typeDist'}&\sbr{\p{g(\type^1_i,\type_{-i}) - g(\type^2_i,\type_{-i})} - \p{g(\type^2_i,\type_{-i}) - g(\type^3_i,\type_{-i})}} \\
		&\le \expect_{\type \sim \typeDist}\sbr{\p{g(\type^1_i,\type_{-i}) - g(\type^2_i,\type_{-i})} - \p{g(\type^2_i,\type_{-i}) - g(\type^3_i,\type_{-i})}}.
	\end{align*}
\end{lemma}
\begin{proof}
	Consider any $i$, $\Delta^1$, and $\Delta^2$.
	Then, define
	\[\tilde{g}(\type_i,\type_{-i}) = \p{g(\type_i+\Delta^1+\Delta^2,\type_{-i}) - g(\type_i+\Delta^1,\type_{-i})} - \p{g(\type_i+\Delta^1,\type_{-i}) - g(\type_i,\type_{-i})}.\]
	By the assumption that $g$ has decreasing differences, $\tilde{g}$ is non-increasing in $\type_i$.
	We can then immediately apply \cref{lem:aff_exp_order} to complete the proof.
\end{proof}

\subsubsection*{Proof of \cref{lem:aff_ex_int_funcs}}
	To prove that $\exinttransfer$ has the claimed form, we will consider a specific non-truthful reporting: if a bidder has value $\type^m$, she commits to mis-reporting (selecting partitions) $\type^{m'}$ for the rest of the game.
	Since our direct mechanism is an equilibrium for real bidders, we must have
	\begin{align*}
		\expect\Big[\Big(\exintquantity_i(\type^{j_m};&\possibleValues)\val(\type^{j_m},\type_{-i}) - \exinttransfer_i(\type^{j_m};\possibleValues)\Big) - \\
		 &\p{\exintquantity_i(\type^{j_{m-1}};\possibleValues)\val(\type^{j_m},\type_{-i}) - \exinttransfer_i(\type^{j_{m-1}};\possibleValues)} \mid \type_i=\type^{j_m}, \type_{-i} \in \possibleValues_{-i} \Big] \ge 0, \text{ and} \\
		 \expect\Big[\Big(\exintquantity_i(\type^{j_{m-1}};&\possibleValues)\val(\type^{j_{m-1}},\type_{-i}) - \exinttransfer_i(\type^{j_{m-1}};\possibleValues)\Big) - \\
		 &\p{\exintquantity_i(\type^{j_{m}};\possibleValues)\val(\type^{j_{m-1}},\type_{-i}) - \exinttransfer_i(\type^{j_{m}};\possibleValues)} \mid \type_i=\type^{j_{m-1}}, \type_{-i} \in \possibleValues_{-i} \Big] \ge 0.
	\end{align*}
	Defining $U_i$ to be the ex-interim utility for bidder $i$, the preceding expressions become:
	\begin{gather*}
		U_i(\type^{j_m};\possibleValues) \ge U_i(\type^{j_{m-1}};\possibleValues) + (\val^*_i(\type^{j_{m-1}},\type^{j_m};\possibleValues) - \val^*_i(\type^{j_{m-1}},\type^{j_{m-1}};\possibleValues))\exintquantity_i(\type^{j_{m-1}};\possibleValues), \text{ and} \\
		U_i(\type^{j_{m-1}};\possibleValues) \ge U_i(\type^{j_{m}};\possibleValues) + (\val^*_i(\type^{j_{m}},\type^{j_m};\possibleValues) - \val^*_i(\type^{j_{m}},\type^{j_{m-1}};\possibleValues))\exintquantity_i(\type^{j_{m}};\possibleValues).
	\end{gather*}
	Thus,
	\begin{align*}
		(\val^*_i(\type^{j_{m-1}},&\type^{j_m};\possibleValues) - \val^*_i(\type^{j_{m-1}},\type^{j_{m-1}};\possibleValues))\exintquantity_i(\type^{j_{m-1}};\possibleValues) \\
		&\le U_i(\type^{j_m};\possibleValues) - U_i(\type^{j_{m-1}};\possibleValues) \\
		&\le (\val^*_i(\type^{j_{m}},\type^{j_m};\possibleValues) - \val^*_i(\type^{j_{m}},\type^{j_{m-1}};\possibleValues))\exintquantity_i(\type^{j_{m}};\possibleValues).
	\end{align*}
	Then, by the same logic as in \cref{lem:ex_int_funcs}, a mechanism with an optimal transfer rule is such that
	\begin{equation*}
		U_i(\type^{j_m};\possibleValues) = \sum_{k=2}^{m} (\val^*_i(\type^{j_{k-1}},\type^{j_k};\possibleValues) - \val^*_i(\type^{j_{k-1}},\type^{j_{k-1}};\possibleValues))\exintquantity_i(\type^{j_{k-1}};\possibleValues).
	\end{equation*}
	Therefore, we have that
	\begin{equation}
		\exinttransfer_i(\type^{j_m};\possibleValues) = \exintquantity_i(\type^{j_m};\possibleValues)\val^*_i(\type^{j_{m}},\type^{j_m};\possibleValues) - \sum_{k=2}^{m} (\val^*_i(\type^{j_{k-1}},\type^{j_k};\possibleValues) - \val^*_i(\type^{j_{k-1}},\type^{j_{k-1}};\possibleValues))\exintquantity_i(\type^{j_{k-1}};\possibleValues).
	\end{equation}
	Re-arranging concludes the proof.
\qed

\subsubsection*{Proof of \cref{lem:aff_set_order}}
	There are two cases to consider.
	In the first case, suppose $\max_{\type \in S} \cbr{\type} > \max_{\type \in S'} \cbr{\type}$.
	In this case, we can form subsets $S_1,\ldots,S_K$ and $S'_1,\ldots,S'_K$ such that
	\begin{enumerate}[(i)]
		\item $S = \bigcup_k S_k$ and $S' = \bigcup_k S'_k$;
		\item For all $k$, $x \in S_k$, and $y \in S'_{k}$: $x > y$; and
		\item For all $k' > k$, $x \in S_k$, and $y \in S'_{k'}$: $x < y$.
	\end{enumerate}
	Such partitions can be formed inductively.
	First, define the base case as 
	\[ S_K = \cbr{\type \in S : \type > \max_{\tilde{\type} \in S'} \cbr{\tilde{\type}}} \text{ and } S'_K = \cbr{\type \in S' : \type > \max_{\tilde{\type} \in S \backslash S_K} \cbr{\tilde{\type}}}. \]
	Then, define the inductive case as 
	\begin{gather*}
		S_{k} = \cbr{\type \in S \ \backslash \ \p{\bigcup_{k' > k} S_{k'}} : \type > \max_{\tilde{\type} \in S' \backslash \p{\bigcup_{k' > k} S'_{k'}}} \cbr{\tilde{\type}}} \text{ and } \\
		S'_{k} = \cbr{\type \in S' \ \backslash \ \p{\bigcup_{k' > k} S_{k'}} : \type > \max_{\tilde{\type} \in S \backslash \p{\bigcup_{k' \ge k} S_{k'}}} \cbr{\tilde{\type}}}.
	\end{gather*}
	Note that these two partitions have the same number of elements because $\min_{\type \in S} \cbr{\type} > \min_{\type \in S'} \cbr{\type}$.
	Then, we can re-write the right-hand side of \cref{eq:aff_set_order} as
	\[ \sum_{k=1}^K \p{\prob_F\sbr{\type_i \in S_k \mid S}\expect_{\type \sim \typeDist}\sbr{g(\type) \mid \type_i \in S_k} - \prob_F\sbr{\type_i \in S'_k  \mid S'}\expect_{\type \sim \typeDist}\sbr{g(\type) \mid \type_i \in S'_k}}. \]
	Note that by assumption, the marginals are equal: $\prob_{\typeDist'}\sbr{\type_1 \in S_k} = \prob_{\typeDist}\sbr{\type_1 \in S_k}$ for all $k$.
	We can then recall that $g$ is non-decreasing, weakly super-modular, and that affiliation is equivalent to log-supermodularity of the type distribution to conclude that for all $k$,
	\begin{align*}
		\prob_{F'}&\sbr{\type_i \in S_k \mid S}\expect_{\type \sim \typeDist'}\sbr{g(\type) \mid \type_i \in S_k} - \prob_{F'}\sbr{\type_i \in S'_k \mid S'}\expect_{\type \sim \typeDist'}\sbr{g(\type) \mid \type_i \in S'_k} \\
		&\ge \prob_F\sbr{\type_i \in S_k \mid S}\expect_{\type \sim \typeDist}\sbr{g(\type) \mid \type_i \in S_k} - \prob_F\sbr{\type_i \in S'_k \mid S'}\expect_{\type \sim \typeDist}\sbr{g(\type) \mid \type_i \in S'_k}
	\end{align*}
	Then, we can sum over all $k$ to conclude the proof in this case.
	
	Now, we consider the case where $\max_{\type \in S} \cbr{\type} < \max_{\type \in S'} \cbr{\type}$.
	In this case we can form subsets $S_1,\ldots,S_K$ and $S'_1,\ldots,S'_{K+1}$ such that
	\begin{enumerate}[(i)]
		\item $S = \bigcup_k S_k$ and $S' = \bigcup_k S'_k$;
		\item For all $k$, $x \in S_k$, and $y \in S'_{k}$: $x > y$; and
		\item For all $k' > k$, $x \in S_k$, and $y \in S'_{k'}$: $x < y$.
	\end{enumerate}
	The inductive construction is the same as the first case with the base case changed to
	\[ S'_{K+1} = \cbr{\type \in S' : \type > \max_{\tilde{\type} \in S} \cbr{\tilde{\type}}} \text{ and } S_K = \cbr{\type \in S : \type > \max_{\tilde{\type} \in S' \backslash S'_{K+1}} \cbr{\tilde{\type}}}. \]
	We then re-write the right-hand side of \cref{eq:aff_set_order} as
	\begin{align*}
		\frac{1}{2}&\prob_{\typeDist}\sbr{\type_i \in S_1}\expect_{\type \sim \typeDist}\sbr{g(\type) \mid \type_i \in S_1} - \prob_{\typeDist}\sbr{\type_i \in S'_1}\expect_{\type \sim \typeDist}\sbr{g(\type) \mid \type_i \in S'_1} + \\
		\frac{1}{2}&\prob_{\typeDist}\sbr{\type_i \in S_K}\expect_{\type \sim \typeDist}\sbr{g(\type) \mid \type_i \in S_K} - \prob_{\typeDist}\sbr{\type_i \in S'_{K+1}}\expect_{\type \sim \typeDist}\sbr{g(\type) \mid \type_i \in S'_{K+1}} + \\
		& \frac{1}{2}\sum_{k=2}^K \Bigg(\p{\prob_F\sbr{\type_i \in S_k}\expect_{\type \sim \typeDist}\sbr{g(\type) \mid \type_i \in S_k}  - \prob_{\typeDist}\sbr{\type_i \in S'_k}\expect_{\type \sim \typeDist}\sbr{g(\type) \mid \type_i \in S'_k}} \\
		& \qquad\quad - \p{\prob_{\typeDist}\sbr{\type_i \in S'_k}\expect_{\type \sim \typeDist}\sbr{g(\type) \mid \type_i \in S'_k} - \prob_{\typeDist}\sbr{\type_i \in S_{k-1}}\expect_{\type \sim \typeDist}\sbr{g(\type) \mid \type_i \in S_{k-1}}}\Bigg)
	\end{align*}
	Now, observe that the first two differences must decrease under $\typeDist'$ because $g$ is non-decreasing and weakly super-modular.
	We can then recall that $g$ has decreasing differences, apply \cref{lem:aff_dec_diff_order}, and then sum over all $k$ to conclude that the sum also must decrease and therefore obtain the desired conclusion.
\qed

\subsubsection*{Proof of \cref{prop:aff_shill_order}}

	\paragraph{Affiliation.}
	Towards contrapositive, suppose that the auction is not shill-proof (strong or weak, we will highlight where the proofs diverge) under $\typeDist$ and consider some $\typeDist' \affOrder \typeDist$.
	There exists some realization of $\type$ under $\typeDist$ such that a shill bidder $i$ deviates at $\possibleValues$.
	Since $\quantity$, $\menuRule$, and $\player_0$ are all fixed, we know from \cref{lem:aff_ex_int_funcs} that under weak shill-proofness that the only way for expected revenue under weak shill-proofness to change is for the ex-interim expected value of some bidder $j \ne i$ to change.
	We can also apply \cref{lem:ssp_epab} to see that the same is true under strong shill-proofness.
	Thus, in order to prove that the auction is not shill-proof for either weak or strong shill-proofness under $\typeDist'$, it is sufficient to show that for all $j \ne i$, if the expected transfers from $\exinttransfer_j$ from manipulation $\hat{\possibleValues}$ compared to truthful play $\possibleValues^*$ under $\typeDist$, then the same manipulation is profitable under $\typeDist'$.
	Fix $j$ and consider the profitably manipulated possible values $\hat{\possibleValues}$.
	We can re-express the expected transfers from bidder $j$ from \cref{lem:aff_ex_int_funcs} as
	\begin{align*}
		\expect&_{\type \sim \typeDist}\sbr{ \exinttransfer_j(\type_j;\possibleValues^*) \mid \type \in \possibleValues^*} \\
		&=\expect_{\type \sim \typeDist}\Bigg[\expect\sbr{\quantity_j(\type_j,\type_{-j})\val(\type_j, \type_{-j}) \mid \type_j = \type_j,\quantity_j(\type_j,\type_{-j}) = 1} \\
		&\qquad - \sum_{m : \type^{k_m} < \type_j} \prob\sbr{\quantity_j(\type^{k_m},\type_{-j}) = 1} \big( \expect\sbr{\val(\type^{k_m}, \type_{-j}) \mid \type_j = \type^{k_{m+1}},\quantity_j(\type^{k_m},\type_{-j}) = 1} \\
		&\qquad\qquad\qquad\qquad - \expect\sbr{\val(\type^{k_m}, \type_{-j}) \mid \type_j = \type^{k_{m}},\quantity_j(\type^{k_m},\type_{-j}) = 1} \big) \mid \type_{-j} \in \possibleValues_{-j}^*, \type_j \in \possibleValues^*_j  \Bigg].
	\end{align*}
	Now observe that with manipulation only changes that the expectation is taken with respect to $\type_{-j} \in \hat{\possibleValues}_{-j}$.
	Note that we are not worrying about the true distribution from bidder $j' \ne j$ because when considering revenue from bidder $j$, we consider her expected transfers, and then integrate over $j$'s true expected distribution, $\type_j \in \possibleValues^*_j$.
	Given that $\val$ satisfies \cref{assum:interdep} and $\quantity$ is orderly, we know that
	\begin{equation}\label{eq:transfer_aff}
		\begin{aligned}
			\expect&\sbr{\quantity_j(\type_j,\type_{-j})\val(\type_j, \type_{-j}) \mid \type_j = \type_j,\quantity_j(\type_j,\type_{-j}) = 1} \\
			& - C_m \cdot \p{\expect\sbr{\val(\type^{k_m}, \type_{-j}) \mid \type_j = \type^{k_{m+1}},\quantity_j(\type^{k_m},\type_{-j}) = 1} - \expect\sbr{\val(\type^{k_m}, \type_{-j}) \mid \type_j = \type^{k_{m}},\quantity_j(\type^{k_m},\type_{-j}) = 1}}
		\end{aligned}
	\end{equation}
	also has the same properties (with respect to $\type_j$) for any $m$ such that $\type^{k_m} < \type_j$ and constant $C_m$.
	Therefore, we can conclude that the necessary conditions for \cref{lem:aff_set_order} are satisfied and apply it to conclude that
	\begin{equation*}
		\expect_{\type \sim \typeDist'}\sbr{ \exinttransfer_j(\type_j;\hat{\possibleValues}) - \exinttransfer_j(\type_j;\possibleValues^*) \mid \type \in \possibleValues^*} \ge \expect_{\type \sim \typeDist}\sbr{ \exinttransfer_j(\type_j;\hat{\possibleValues}) - \exinttransfer_j(\type_j;\possibleValues^*) \mid \type \in \possibleValues^*}.
	\end{equation*}
	Thus, the auction is not shill-proof under $\typeDist'$.

	\paragraph{Common Values.}
	Towards contrapositive, suppose that the auction is not shill-proof (strong or weak, we will highlight where the proofs diverge) under $\val$ and consider some $\val' \valOrder \val$.
	There exists some realization of $\val$ under $\val$ such that a shill bidder $i$ deviates at $\possibleValues$.
	Since $\quantity$, $\menuRule$, and $\player_0$ are all fixed, we know from \cref{lem:aff_ex_int_funcs} that under weak shill-proofness that the only way for expected revenue under weak shill-proofness to change is for the ex-interim expected value of some bidder $j \ne i$ to change.
	We can also apply \cref{lem:ssp_epab} to see that the same is true under strong shill-proofness.
	Thus, in order to prove that the auction is not shill-proof for either weak or strong shill-proofness under $\valDist'$, it is sufficient to show that for all $j \ne i$, if the expected transfers from $\exinttransfer_j$ from manipulation $\hat{\possibleValues}$ compared to truthful play $\possibleValues^*$ under $\val$, then the same manipulation is profitable under $\val'$.
	This is immediate: since $\val'$ is point-wise higher than $\val$ and more super-modular, then we can see that \cref{eq:transfer_aff} is larger under $\val'$ than $\val$ no matter the constant and therefore the auction is not shill-proof.

\subsubsection*{Proof of \cref{prop:ssp_pab}}
	The proof that a strongly shill-proof auction must be pay-as-bid is the same as in \cref{lem:ssp_pab} because that argument is ex-post and bidders' values are weakly increasing in all types.
	To prove that the conjectured efficient allocation rule is efficient, observe that by the assumption that a bidder values her own signal more than that of others, the efficient ex-post allocation is such that $\quantity_i = 0$ for $i \notin \argmax_j \cbr{\type_j}$ and $\sum_j \quantity_j(\type) = 1$ for all $\type$.
	The assumption that the seller has $0$ value for the good completes the proof that $\quantity^E$ is the efficient allocation rule.
	
	Recall that a shill-proof auction must be pay-as-bid, and observe that bids must be strictly increasing as a function of own signal in order to be incentive compatible.
	Next, because the bidder who has the highest ex-ante signal will have the highest ex-post valuation, under any allocation rule, the bidder with the highest signal who is allocated will have the highest bid.
	Therefore, to maximize transfers the optimal allocation rule is such that $\quantity_i = 0$ for $i \notin \argmax_j \cbr{\type_j}$.
	To conclude the proof, observe that by \cref{lem:aff_ex_int_funcs}, the fact that changing the allocation rule conditional on only the maximum value will not change that bidder's value, and by the standard arguments, the optimal allocation rule uses a reserve type.

\newpage
\section{Supplemental Appendix}
The following definition for an extensive-form auction is taken\footnote{We modify the definition to remove notation we do not use and to make it specific to auctions.} from \citet{li17}:
\begin{definition} \label{def:game}
	An \textbf{extensive form auction} $\game$ is defined as the tuple $\p{\possibleHistories,\prec,\actionSet,\recentAction,\playerFunction,\cbr{\informationSet_i}_{i \in \potentialBidders},(\actionQuantity,\actionTransfer)}$ such that:
	\begin{enumerate}[(i)]
		\item $\possibleHistories$ is a set of histories, along with a binary relation $\prec$ on $\possibleHistories$ that represents precedence. In addition:
		\begin{enumerate}[(a)]
			\item $\prec$ forms a partial order and $(\possibleHistories,\prec)$ forms an arborescence.
			\item There exists an initial history $\history_{\emptyset}$ such that there does not exist $\history'$ where $\history' \prec \history_{\emptyset}$.
			\item The set of terminal histories is $\possibleTerminalHistories \equiv \cbr{\history : \neg\exists\history \st \history \prec \history' }$.
			\item The set of immediate successors to $\history$ is $\succz(\history)$.
		\end{enumerate}
		\item $\actionSet$ is the set of possible actions.
		\item $\recentAction:\possibleHistories\backslash\history_{\emptyset} \to \actionSet$ maps histories to the most recent action taken to reach it. In addition:
		\begin{enumerate}[(a)]
			\item For all $\history$, $\recentAction(\history)$ is one-to-one on $\succz(\history)$.
			\item The set of actions available at $\history$ is
			\[ \actionSet(\history) \equiv \bigcup_{\history' \in \succz(\history)} \recentAction(\history'). \]
		\end{enumerate}
		\item $\playerFunction: \possibleHistories\backslash\possibleTerminalHistories \to \potentialBidders$ is the player function for any given non-terminal history.
		\item $\informationSet_i$ is a partition of $\cbr{\history : \playerFunction(\history) = i}$ such that:
		\begin{enumerate}[(a)]
			\item $\actionSet(\history) = \actionSet(\history')$ when $\history$ and $\history'$ are in the same cell of the partition, and
			\item $\actionSet(\history) \cap \actionSet(\history') = \emptyset$ when $\history$ and $\history'$ are not in the same cell of the partition.
		\end{enumerate}
		\item For every $\terminalState \in \possibleTerminalHistories$, $\terminalState = (\actionQuantity,\actionTransfer)$, such that $\sum_{i=1}^{\numBidders} \actionQuantity_i \le 1$, $x_i \in [0,1]$, and $t_i \in \R$.
	\end{enumerate}
\end{definition}

In order to define an information order, we will use the notation that when the current set of possible values is $\possibleValues$, a player $i$'s knowledge of what values are possible is $\informationSet^i_{\possibleValues}$.
\begin{definition}
	A menu rule $\menuRule'$ is \textbf{more informative} than $\menuRule$, $\menuRule' \infoOrder \menuRule$, if $\menuRule(\cdot,\cdot) = \menuRule'(\cdot,\cdot)$ for all game states $\possibleValues$ and all possible information sets $\informationSet, \informationSet'$, when $\informationSet \subseteq \informationSet'$, it is the case that
	\[ \prob\sbr{\informationSet_{\possibleValues} = \informationSet \mid \menuRule'} - \prob\sbr{\informationSet_{\possibleValues} = \informationSet' \mid \menuRule'} \ge \prob\sbr{\informationSet_{\possibleValues} = \informationSet \mid \menuRule} - \prob\sbr{\informationSet_{\possibleValues} = \informationSet' \mid \menuRule}. \]
\end{definition}

\begin{proposition}\label{prop:info_shill_order}
	Consider any affiliated type distribution $\typeDist$ and value function $\val$ satisfying \cref{assum:interdep}.
	Suppose $(\quantity,\optTransfer,\menuRule,\player_0)$ is orderly and strongly shill-proof. Then, for any $\menuRule'$ such that $\menuRule \infoOrder \menuRule'$, it is the case that $(\quantity,\transfer^*(\quantity,\menuRule',\player_0,\val,\typeDist),\menuRule',\player_0)$ is strongly shill-proof.
\end{proposition}
\begin{proof}
	Towards contrapositive, suppose that the auction is not strongly shill-proof under $\menuRule$ and consider some $\menuRule' \infoOrder \menuRule$.
	There exists some realization of $\type$ under $\typeDist$ such that a shill bidder $i$ deviates at $\possibleValues$.
	Since $\valDist$ is affiliated and $\quantity$, $\player_0$, and $\menuRule$ are all fixed, we can apply \cref{lem:aff_exp_order,lem:aff_ex_int_funcs,lem:ssp_epab} to conclude that the only way for expected revenue under strong shill-proofness to change is for ex-interim expected values for the winner $j \ne i$ to change.
	There can be no change in allocation since $j$ has the highest value among the real bidders and so since the allocation is orderly, any change in the allocation would result in a shill bidder winning the item.
	Recall that the winning bidder conditions her valuation on the realization that she wins, i.e., she conditions on the assumption that $(\type_j,j) \ordering (\type_k,k)$ for all $k \ne j$.
	Thus, it is without loss with regards to how the winning bidder $j$ will estimate her valuation to assume that $\informationSet_{\possibleValues}$ is such that for all $y \in \informationSet_{\possibleValues}$ and $k \ne j$, $(\type_j,j) \ordering (y_k,k)$.
	There are now two cases to consider.
	In the first case, consider that under $\menuRule'$, the information sets in the support are a super-set of those in the support under $\menuRule$.
	We then condition on the event of the same realization of the information set that leads to a profitable deviation to find that the auction is not strongly shill-proof under $\menuRule'$.
	In the second case, we do not assume that the support under $\menuRule$ contains elements that $\menuRule'$ does not.
	It may be that the realization of the information set for the winner's last move is not in the support of $\menuRule$.
	However, because $\menuRule' \infoOrder \menuRule$, we can know that there is some realization of the information structure such that $0$ is not contained in it.
	This is because we know that there exists an information set that is strictly smaller than that realization and by Bayes plausibility, the complement must also be in it.
	Further, a shill bidding deviation means selecting a partition which does not include $0$ and by \cref{lem:aug_rev_prin} Condition 2a, we know that said partition does not intersect the partition containing $0$.
	Since the allocation is fixed, one of the partitions must increase transfers from the winner in order to have the average transfer under every realization of the information set average out to the transfer under the coarser information set.
	We then select the highest transfer among these realizations to find a profitable deviation and complete the proof.
\end{proof}

\begin{example}\label{ex:fpa}
	Consider the sealed-bid, first-price (pay-as-bid) optimal auction with equilibrium bids $b^1$ in the regular, IPV environment.
	The na\"ive implementation of allocation and transfer rule $(\quantity^*,\transfer^1)$ in a public auction would be to query each bidder sequentially on what her value is and then have the payment rule be $\transfer^1$, but $\transfer^1$ is not the direct transfer rule of any equilibrium of this game. 
	Indeed, consider the last bidder who takes a move, and label that bidder $\numBidders$. 
	If $\type_{\numBidders} > \max_{i < \numBidders} \cbr{b_{i}}$, then the only possible equilibrium bid---and therefore the transfer function---is $\max_{i < \numBidders} \cbr{b_{i}}$.
	
	If we modify the direct transfer rule to represent the bid that each bidder submits in equilibrium in this sequential form (as one could solve for inductively), the auction would be \weakly~shill-proof by regularity. In particular, while a shill bid can force later bidders to pay a higher price, the probability that no one will want to pay that higher price outweighs the benefit by regularity.
	However, such an auction is not \strongly~shill-proof because, given full knowledge of real bidders' valuations, a shill bidder will be incentivized to bid just below the highest valuation of a subsequent bidder.
\end{example}

\begin{example}\label{ex:robust_wsp}
	Let $\contDist(x) = 1 - e^{-0.1 x}$ and let $\typeDist_1,\typeDist_2$ be discrete approximations of $\contDist$ with atoms $\cbr{0,5,9,14,20}$ and $\cbr{0,3,7,14,20}$, respectively. 
	It can be verified that both these distributions are regular and have optimal reserve $\optimalReserve = 14$.
	Consider a variant of the efficient Dutch auction, with the modification that, if all bidders have indicated values less than $20$, then the auction queries bidders from lowest-to highest-priority as to whether their value is at least $9$.
	If no one indicates that their value is at least $9$, then the Dutch auction continues.
	If at least one person does indicate that their value is at least $9$, then bidders are queried from lowest- to highest-priority as to whether their value is $14$, and the transfer is $14$ if at least two people have value $14$, and $9$ if only one person does.
	It can be verified that if the value distribution is $\typeDist_1$, the auction just described is \weakly~shill-proof, but if the value distribution is $\typeDist_2$, then the auction is not \weakly~shill-proof.
	When the value distribution is $\typeDist_2$, in expectation, a shill bidder will want to report that her value is $9$.
	In fact, \cref{lem:dutch_sparse_wsp} (see appendix) implies that if the value distribution is $\typeDist_2$, then the auction in this example must be a semi-Dutch auction with cutoff at least $14$.	
\end{example}

\subsection{Credibility} \label{ss:cred}
The following definitions are adapted from \cite{akbarpourLi20} to match our notation and specialized to the auction setting:

For any extensive form game $\game$, we can define a messaging game as follows:
\begin{enumerate}[1.]
	\item The auctioneer chooses to:
	\begin{enumerate}[(a)]
		\item Select an outcome and end the game; or
		\item Go to step 2.
	\end{enumerate}
	
	\item The auctioneer chooses some bidder $i \in \potentialBidders$ and sends a message $m= I_i \in \mathcal{I}_i$.
	
	\item Bidder $i$ privately observes message $m=I_i$ and chooses reply $r \in \actionSet(I_i)$.

	\item The auctioneer privately observes $r$.
	
	\item Go to step 1.
\end{enumerate}
We can now write bidder $i$'s observations in the game as $((m_i^k,r_i^k)_{k=1}^{\tau_i},\omega_i)$ where $\tau_i$ is the number of observations $i$ has and $\omega_i$ is the information partition over outcomes that $i$ observes.
Let $o_i(\strategy_0,\strategy,\type)$ be $i$’s observation when the auctioneer plays $\strategy_0$, the bidders play $\strategy$, and the type profile is $\type$.

\begin{definition}[\cite{akbarpourLi20}]
	Let $\strategy_0^G$ be the \textbf{rule-following auctioneer strategy}.
	Formally, $\strategy_0^G$ is defined by the following algorithm: Initialize $\hat{\history} := \history_{\emptyset}$. At each step, if $\hat{\history} \in \terminalState$, end the game and choose outcome $(\actionQuantity,\actionTransfer)(\hat{\history})$.
	Else, contact agent $\playerFunction(\hat{\history})$ and send message $m= \informationSet_{\playerFunction(\hat{\history})}$ such that $\playerFunction(\hat{\history}) \in \informationSet_{\playerFunction(\hat{\history})}$. 
	Upon receiving reply $r$, update $\hat{\history}$ to $\history$ such that $\history \in \succz(\hat{\history})$ and $\recentAction(\history) = r$, then iterate.
\end{definition}
\begin{definition}[\cite{akbarpourLi20}, Definition 3]\label{def:safe_strat}
	Given some promised strategy profile $(\strategy_0,\strategy)$, auctioneer strategy $\hat{\strategy}_0$ is safe if, for all agents $i \in \potentialBidders$ and all type profiles $\type$, there exists $\type'_{-i}$ such that
	$o_i(\hat{\strategy}_0,\strategy,\type) = o_i(\strategy_0,\strategy,(\type_i,\type'_{-i}))$.
	We denote by $\Sigma^{*}_{0} (\strategy_0,\strategy)$ the set of \textbf{safe strategies}.
\end{definition}
\begin{definition}[\cite{akbarpourLi20}, Definition 4] \label{def:credible}
	$(\game,\strategy)$ is \textbf{credible} if
	\[ \strategy^G_0 \in \argmax_{\strategy_0 \in \Sigma^*_0(\strategy_0^G,\strategy)} \cbr{ \expect_{\type}\sbr{\sum_{i \in \potentialBidders} \actionTransfer_i(\strategy_0,\strategy,\type)}} \]
\end{definition}

\subsubsection*{Proof of \cref{prop:cred_nesting}}
\textbf{\Strong~Shill-Proofness $\to$ Credibility. } We prove the contrapositive: Suppose $(\game,\strategy)$ is not credible.
Let $\hat{\strategy}_0 \in \Sigma^*_0(\strategy_0^G,\strategy)$ be a profitable and safe deviation by the auctioneer.
By \cref{def:safe_strat}, there exists $\type$ and $\cbr{\type'_{-i}}_i$ such that $o_i(\hat{\strategy}_0,\strategy,\type) = o_i(\strategy_0,\strategy,(\type_i,\type'_{-i}))$ for all $i$.
By \cref{lem:winner_paying}, only one bidder $i^*$ has $\actionTransfer_{i^*}(\strategy(\type_{i^*},\type'_{-i^*})) \ne 0$ and so $\actionTransfer_{i^*}(\strategy(\type_{i^*},\type'_{-i^*})) > \actionTransfer_{i^*}(\strategy(\type))$ because $\hat{\strategy}_0$ is profitable.
Then, let $\realBidders = \cbr{i^*}$, and by ex-post monotonicity, 
\[ \actionTransfer_{i^*}(\strategy(\type_{i^*},\type'_{-i^*})) > \actionTransfer_{i^*}(\strategy(\type)) \ge \actionTransfer_{i^*}(\strategy(\type_{i^*},0)), \]
and so the auction is not \strongly~shill-proof.

\textbf{Credibility $\to$ \Weak~Shill-Proofness. }  We prove the contrapositive: Suppose $(\game,\strategy)$ is not \weakly~shill-proof.
Let $\hat{\strategy} \in \Sigma_{\shillBidders}$ be a profitable shilling strategy.
Then, by \cref{def:game_eq}, for all $\realBidders$, there exists $\type,\type'$ such that $\hat{\strategy}(\type;\realBidders) = \strategy(\type;\potentialBidders)$.
Consider the following reporting strategy $\hat{\strategy}_0$: for all $i \in \realBidders$, report play as if $i \in \realBidders$ is following $\hat{\strategy}$; and for all $i \notin \realBidders$, report in the rule-following manner.
This strategy is safe because $\hat{\strategy} \in \Sigma_{\shillBidders}$.
To see that it is profitable compared to $\strategy_0^{\game}$, consider what happens when the winning bidder $i$ is or is not in $\realBidders$.\footnote{Note that if no one has value above the optimal reserve, there will be no winner under any safe strategy, so let us only consider the case where the good is allocated.}
Conditional on $i \in \realBidders$ winning, $\hat{\strategy}_0$ increases expected revenue because $\hat{\strategy}$ is a profitable shill bidding strategy.
Conditional on $i \notin \realBidders$ winning, shill bidding would have led $0$ revenue for the seller and, by \cref{lem:winner_paying}, $\hat{\strategy}_0$ must have non-negative revenue.
Thus, our described reporting strategy is a profitable, safe strategy and therefore the auction is not credible.
\qed

\subsection{Generalizing Credibility in the Single-Action Case}
\begin{definition} \label{def:safe_dev}
	Fix a single-action auction with exogenous signal $\experiment$ and a set of real bidders $\realBidders$. The set of safe deviations to report to $i \in \potentialBidders$ is
	\begin{align*}
		\devActionSet_i&(\type_{j \le i}) = \\
		&\cbr{ \action : \exists \tilde{\type}_{-i} \st \sbr{j < i, \type_j \ne 0 \implies \tilde{\type}_j = \type_j} \text{ and } \action = \p{\strategy_i(\type_i;\experiment_i(\tilde{\type}_{j < i})), \strategy_{-i}\p{\type_i,\tilde{\type}_{-i};\potentialBidders}} }.
	\end{align*}
	The total set of \textbf{safe deviations} is
	\begin{equation*}
		\devActionSet(\type) = \cbr{ \cbr{\reportedAction{i}} : \reportedAction{i} \in \devActionSet_i(\type_{j \le i}) \text{ and } \sum_{i=1}^{\numBidders} \actionQuantity_i\p{\reportedAction{i}} \le 1 }.
	\end{equation*}
\end{definition}
\cref{def:safe_dev} allows the auctioneer to report any value she chooses when a bidder's declared valuation is $0$.
Note that if we take the canonical setting where there are no exogenous signals, the above assumption is without loss.

\begin{definition} 
	\label{def:gen_credible}
	A single-action auction is \textbf{$\experiment$-credible} if for all $\type$ and $\cbr{\reportedAction{i}} \in \devActionSet(\type)$, we have
	\begin{equation*}
		\sum_{i} \actionTransfer_i(\reportedAction{i}) \le \sum_{i} \actionTransfer_i(\strategy(\type; \potentialBidders)).
	\end{equation*}
\end{definition}

\begin{lemma} \label{lem:direct_cred}
	For a single-action auction, define the augmented (direct) inverse $\widecheck{\experiment}_{i}^{-1}$ as $\widecheck{\experiment}_{i}^{-1}(\type) = \cbr{0} \cup \experiment_i^{-1}(\type_{j<i})$.
	Then for
	\[ \devValSet(\type) = \cbr{ \cbr{\reportedValue{i}} : \reportedValue{i} \in \widecheck{\experiment}^{-1}_i(\type), \sum_{i} \quantity_i(\reportedValue{i}) \le 1 },\]
	the auction is credible if and only if for all $\type$, and $\cbr{\reportedValue{i}} \in \devValSet(\type)$,
	\[ \sum_i \transfer_i(\reportedValue{i}) \le \sum_i \transfer_i(\type). \]
\end{lemma}
\begin{proof}
	Apply \cref{lem:one_shot_rev_principle}, specifically the unique mapping between $(\quantity,\transfer)$ and $(\actionQuantity,\actionTransfer)$ to \cref{def:safe_dev,def:credible} to see that the lemma holds.
\end{proof}

\begin{lemma} \label{lem:wsp_not_ssp}
	Suppose a single-action auction is \weakly~shill-proof, but not \strongly~shill-proof.
	Then, there exist $\realBidders$, $\type_{\realBidders}$, and $\type_{-\realBidders}$ such that
	\begin{equation} \label{eq:ssp_prof_deviation}
		\sum_{k \in \realBidders} \transfer_k(\type_{\realBidders},\type_{-\realBidders}) > \sum_{k \in \realBidders} \transfer_k(\type_{\realBidders},0).
	\end{equation}
\end{lemma}
\begin{proof}
	Suppose that $(\game,\strategy)$ is \weakly~shill-proof, but not \strongly~shill-proof.
	Because $(\game,\strategy)$ is \weakly~shill-proof, for all $\type$ and $\realBidders,\realBidders'$, we can define $\hat{\strategy}(\type) \equiv \strategy(\type;\realBidders) = \strategy(\type;\realBidders')$.
	Since $(\game,\strategy)$ is not \strongly~shill-proof, $\hat{\strategy}$ must not be an ex-post strategy for the shill bidders.
	So, for some realization of $\realBidders$ and $\type_{\realBidders}$ there exists a profitable deviation for the shill bidders; examining the set of possible deviations $\Sigma_{\shillBidders}$ in \cref{def:game_eq}, we see that any profitable deviating actions induces a profitable misreport $\type_{-\realBidders}$ in the direct mechanism for some $\realBidders$ and $\type_{\realBidders}$; proving \cref{eq:ssp_prof_deviation} can be satisfied.
\end{proof}

\begin{lemma} \label{lem:ssp_val_equiv}
	If a single-action auction is \strongly~shill-proof, then for all $\realBidders$, $i \notin \realBidders$, $\type_i$, and $\type_{-i}$,
	\iftoggle{compressLines}{$\sum_{k \in \realBidders} \transfer_k(\type_i,\type_{-i}) \le \sum_{k \in \realBidders} \transfer_k(0,\type_{-i}).$}{\[ \sum_{k \in \realBidders} \transfer_k(\type_i,\type_{-i}) \le \sum_{k \in \realBidders} \transfer_k(0,\type_{-i}). \]}
\end{lemma}
\begin{proof}
	Towards contradiction, suppose that there exists $\realBidders$, $i \notin \realBidders$, $\type_i$, and $\type_{-i}$ such that
	\iftoggle{compressLines}{$\sum_{k \in \realBidders} \transfer_k(\type_i,\type_{-i}) > \sum_{k \in \realBidders} \transfer_k(0,\type_{-i}).$}{\[ \sum_{k \in \realBidders} \transfer_k(\type_i,\type_{-i}) > \sum_{k \in \realBidders} \transfer_k(0,\type_{-i}). \]}
	So in the direct game reporting $0$ is not a dominant strategy for shill bidders.
	This implies, from \cref{cond:one_shot_corr} of \cref{lem:one_shot_rev_principle}, that there exists a deviation in the auction such that for some value vectors, the seller raises more revenue.
	Therefore, the auction is not \strongly~shill-proof.
\end{proof}

\subsubsection*{Proof of \cref{prop:sa_cred_nesting}}
\textbf{\Weak~Shill-Proofness $\to$ $(\publicExperiment)$-Credibility. }
Suppose the auction is not $(\publicExperiment)$-credible.
Then, combining \cref{lem:direct_cred} with the ex-post IR condition, there exist $\type,\cbr{\reportedValue{i}} \in \devValSet(\type)$ and $k^*$ such that $\transfer_{k^*}(\reportedValue{k^*}) > \transfer_{k^*}(\type)$.
Applying the definition of orderly and the winner-paying property, for all $j \ne k^*$, $\transfer_j(\reportedValue{k^*}) = 0$. 
Since $\publicExperiment$, for all $j \le k$, it is the case that $\reportedValue{k^*}_j = \type_j$ or $\type_j = 0$.
Let $\realBidders = \set{k^*}$.
Then,
\begin{align*}
	\sum_{i \in \realBidders} \transfer_i(\reportedValue{k^*}) &= \transfer_{k^*}(\reportedValue{k^*}) \\
	&\ge \transfer_{k^*}\p{\type_1,\ldots,\type_{k^*},\reportedValue{k^*}_{k^*+1},\ldots,\reportedValue{k^*}_{\numBidders}} \\
	&\ge \transfer_{k^*}\p{\type_1,\ldots,\type_{k^*},0,\ldots,0} \\
	&= \sum_{i \in \realBidders} \transfer_i(\type_{\realBidders},0).
\end{align*}
Thus, we can apply \cref{lem:wsp_val_equiv_general} to conclude that the auction is not \weakly~shill-proof.

\textbf{$\experiment$-Credibility $\to$ \Weak~Shill-Proofness. }
Suppose the auction is $\experiment$-credible.
Towards contradiction, suppose the auction is not \weakly~shill-proof.
So, there exists $\realBidders$ and $\type$ such that $\strategy(\type;\realBidders) \ne \strategy(\type;\potentialBidders)$.
In particular, this means that shill bidders have, in expectation, a profitable deviation relative to acting as real bidders with valuation $0$.
If this is true in expectation, there must then exist $\type=(\type_{\realBidders},0)$ and $\tilde{\type}_{-\realBidders}$ such that
\[ \sum_{i \in \realBidders} \transfer_i((\type_{\realBidders},\tilde{\type}_{-\realBidders})) > \sum_{i \in \realBidders} \transfer_i((\type_{\realBidders},0)). \]
Now, let us consider the messaging deviation
\[ \cbr{\reportedValue{i}}_{i \in \potentialBidders} = \begin{cases}
	\p{\type_{\realBidders},\tilde{\type}_{-\realBidders}} & i \in \realBidders \\
	\p{\type_{\realBidders},0} & \ow
\end{cases} .\]
By the definition of credibility, the auctioneer can report any value to other bidders when the value reported to him is $0$ and bidders with value $0$ are told the other bidders' true reports.
Therefore, $\cbr{\reportedValue{i}} \in \devValSet(\type_{\realBidders},0)$ and
\[ \sum_i \transfer_i(\reportedValue{i}) = \sum_{i \in \realBidders} \transfer_i(\type_{\realBidders},\tilde{\type}_{-\realBidders}) + \sum_{i \notin \realBidders} \transfer_i(\type_{\realBidders},0) > \sum_i \transfer_i((\type_{\realBidders},0)). \]
This contradicts \cref{lem:direct_cred}, and so the auction must be \weakly~shill-proof.

\textbf{$(\privateExperiment)$-Credibility $\to$ \Strong~Shill-Proofness.}
Suppose that the auction is not \strongly~shill-proof and $\privateExperiment$. 
There are two cases to consider: either the auction is not \weakly~shill-proof or it is.
If the auction is not \weakly~shill-proof, then we can apply the previous case to conclude the auction is not $(\privateExperiment)$-credible.
If the auction is \weakly~shill-proof, but not \strongly~shill-proof, then by \cref{lem:wsp_not_ssp}, there exists $\realBidders,k^* \in \realBidders$, and $\type \st \transfer_{k^*}(\type) > \transfer_{k^*}(\type_{\realBidders},0)$.
Thus, we can construct the following profitable auctioneer reporting deviation:
\[ \cbr{\reportedValue{i}}_{i \in \potentialBidders} = \begin{cases}
	\p{\type_{\realBidders},\type_{-\realBidders}} & i = k^* \\
	\p{\type_{\realBidders},0} & \ow
\end{cases} .\]
Since $\privateExperiment$, we know that $\cbr{\tilde{\type}^{\leadsto i}} \in \devValSet(\type_{\realBidders},0)$.
The total transfers are then
\[ \sum_{i} \transfer_i\p{\reportedValue{i}} = \transfer_{k^*}(\type) + \sum_{i \ne k^*} \transfer_i\p{\type_{\realBidders},0} > \sum_{i} \transfer_i\p{\type_{\realBidders},0}. \]
Hence, by \cref{lem:direct_cred}, we see that the auction is not credible.

\textbf{\Strong~Shill-Proofness $\to$ $\experiment$-Credibility. }
Suppose that the auction is not $\experiment$-credible.
Then, combining \cref{lem:direct_cred} with the ex-post IR condition, there exists $\type,\cbr{\reportedValue{i}} \in \devValSet(\type)$ and $k^*$ such that $\transfer_{k^*}(\reportedValue{k^*}) > \transfer_{k^*}(\type)$.
Recall, by the definition of $\experiment^{-1}$, that $\reportedValue{k^*}_{k^*} = \type_{k^*}$.
Suppose $\realBidders = \cbr{k^*}$.
Then,
\[ \sum_{i \in \realBidders} \transfer_{i}(\reportedValue{k^*}) = \transfer_{k^*}(\reportedValue{k^*}) > \transfer_{k^*}(\type) \ge \transfer_{k^*}(\type_{k^*},0) = \sum_{i \in \realBidders} \transfer_{k^*}(\type_{k^*},0). \]
Therefore, by \cref{lem:ssp_val_equiv}, the auction is not \strongly~shill-proof.
\qed

\end{document}